%% file: main.tex
\documentclass[a4paper,UKenglish,cleveref,thm-restate]{lipics-v2021}
\pdfoutput=1
\usepackage{cite}
\usepackage{amsmath,amssymb,amsfonts}
\usepackage{algorithmic}
\usepackage{graphicx}
\usepackage{enumitem}
\usepackage{textcomp}
\usepackage{etoc}
\usepackage{xcolor}

\usepackage{amssymb}
\usepackage{amsmath}
\usepackage{amsthm}
\usepackage{hyperref}
\crefname{equation}{}{}
\crefname{equations}{}{}
\crefname{figure}{Fig.}{Fig.}
\crefname{definition}{Definition}{Definitions}
\crefname{remark}{Remark}{Remarks}
\crefname{lemma}{Lemma}{Lemmas}
\crefname{section}{Section}{Sections}
\crefname{app}{Appendix}{Appendices}
\usepackage{todonotes}
\usepackage{stmaryrd}
\usepackage{quiver}
\usepackage{bussproofs}
\newenvironment{bprooftree}
{\leavevmode\hbox\bgroup}
{\DisplayProof\egroup}

\usepackage{tikz}
\usetikzlibrary{shapes.geometric}
\usepackage{mathcommand,mathtools}

\usepackage{arydshln}

\usepackage{tikzit}
\input{sample.tikzstyles}

\newtheorem{assumption}[theorem]{Assumption}
\newcommand{\id}{\mathrm{id}}
\newcommand{\tensor}{\otimes} 

\newcommand{\FStoch}{\mathsf{FStoch}}

\newcommand{\FS}{\FStoch_{\scriptscriptstyle \mathsf{tv}}}

\newcommand{\Set}{\mathsf{Set}}

\newcommand{\Mat}{\mathsf{Mat}}
\newcommand{\EnrMat}{\mathsf{Mat}^{\scriptscriptstyle\leq}}

\newcommand{\bN}{\mathbb{N}}
\newcommand{\bR}{\mathbb{R}}

\newcommand{\VV}{\mathcal{V}}
\newcommand{\qV}{V}

\newcommand{\enCat}[1]{{#1}\mathsf{HMet}}
\newcommand{\enCatsym}[1]{{#1}\mathsf{PMet}}
\newcommand{\PMet}{\enCatsym{\zeroinfQ}}
\newcommand{\qVCat}{\enCat{\qV}}
\newcommand{\qVCatsym}{\enCatsym{\qV}}
\newcommand{\qVCatsum}{\qVCat_{\vten}}
\newcommand{\qVCatsymsum}{\qVCatsym_{\vten}}
\newcommand{\qVCatmax}{\qVCat_{\vmeet}}
\newcommand{\qVCatsymmax}{\qVCatsym_{\vmeet}}
\newrobustcmd{\Hclosesum}[1]{#1^{\mathsf{H}\vten}}
\newrobustcmd{\Pclosesum}[1]{#1^{\mathsf{P}\vten}}
\newrobustcmd{\Hclosemax}[1]{#1^{\mathsf{H}\vmeet}}
\newrobustcmd{\Pclosemax}[1]{#1^{\mathsf{P}\vmeet}}
\newrobustcmd{\Hclosestar}[1]{#1^{\mathsf{H}\star}}
\newrobustcmd{\Pclosestar}[1]{#1^{\mathsf{P}\star}}

\newrobustcmd{\Hclosesumsum}[1]{#1^{\mathsf{H}\vten\vten}}
\newrobustcmd{\Pclosesumsum}[1]{#1^{\mathsf{P}\vten\vten}}
\newrobustcmd{\Hclosesummax}[1]{#1^{\mathsf{H}\vten\vmeet}}
\newrobustcmd{\Pclosesummax}[1]{#1^{\mathsf{P}\vten\vmeet}}
\newrobustcmd{\Hclosemaxsum}[1]{#1^{\mathsf{H}\vmeet\vten}}
\newrobustcmd{\Pclosemaxsum}[1]{#1^{\mathsf{P}\vmeet\vten}}
\newrobustcmd{\Hclosemaxmax}[1]{#1^{\mathsf{H}\vmeet\vmeet}}
\newrobustcmd{\Pclosemaxmax}[1]{#1^{\mathsf{P}\vmeet\vmeet}}

\newrobustcmd{\anyclose}[1]{\overline{#1}}
\newcommand{\vle}{\sqsubseteq}
\newcommand{\vge}{\sqsupseteq}
\newcommand{\vjoin}{\sqcup}
\newcommand{\vJoin}{\bigsqcup}

\newcommand{\vten}{\oplus}
\newcommand{\vmeet}{\sqcap}
\newcommand{\C}{\mathcal{C}}
\newcommand{\D}{\mathcal{D}}
\newcommand{\Ob}{\mathop{Ob}}

\newcommand{\twoQ}{{2_\vmeet}}

\newcommand{\zeroinfQ}{{[0,\infty]_+}}

\newrobustcmd{\preordqmat}[1]{\mathsf{POHA}_{#1}}
\newrobustcmd{\metqmat}[1]{\mathsf{MHA}_{#1}}

\newrobustcmd{\RuleRefl}{\hyperlink{rulerefl}{\textnormal{\textsc{Refl}}}}
\newrobustcmd{\RuleTriang}{\hyperlink{ruletriang}{\textnormal{\textsc{Triang}}}}
\newrobustcmd{\RuleBot}{\hyperlink{rulebot}{\textnormal{\textsc{Bot}}}}
\newrobustcmd{\RuleMon}{\hyperlink{rulemon}{\textnormal{\textsc{Mon}}}}
\newrobustcmd{\RuleJoin}{\hyperlink{rulejoin}{\textnormal{\textsc{Join}}}}
\newrobustcmd{\RuleSeqsum}{\hyperlink{ruleseqsum}{$\textnormal{\textsc{Seq}}_{\vten}$}}
\newrobustcmd{\RuleTenssum}{\hyperlink{ruletenssum}{$\textnormal{\textsc{Par}}_{\vten}$}}
\newrobustcmd{\RuleSeqmeet}{\hyperlink{ruleseqmeet}{$\textnormal{\textsc{Seq}}_{\vmeet}$}}
\newrobustcmd{\RuleTensmeet}{\hyperlink{ruletensmeet}{$\textnormal{\textsc{Par}}_{\vmeet}$}}
\newrobustcmd{\RuleCont}{\hyperlink{rulecont}{$\textnormal{\textsc{Cont}}$}}
\newrobustcmd{\RuleSym}{\hyperlink{rulesym}{$\textnormal{\textsc{Symm}}$}}
\newrobustcmd{\letter}{f}
\newrobustcmd{\letterbox}[1]{\renewrobustcmd{\letter}{#1}\input{letterbox.tikz}}

\makeatletter
\DeclareFontEncoding{LS1}{}{}
\DeclareFontEncoding{LS2}{}{\noaccents@}
\DeclareFontSubstitution{LS1}{stix}{m}{n}
\DeclareFontSubstitution{LS2}{stix}{m}{n}
\makeatother

\newrobustcmd{\Dset}{\mathscr{D}}
\newrobustcmd{\Dmet}{\overline{\mathscr{D}}}
\newrobustcmd{\dist}{\varphi}
\newrobustcmd{\distb}{\psi}
\newrobustcmd{\distc}{\tau}
\newrobustcmd{\Dist}{\Phi}
\newrobustcmd{\Distb}{\Psi}
\newrobustcmd{\dirac}[1]{\delta_{#1}}
\newrobustcmd{\supp}[1]{\mathrm{supp}(#1)}
\newrobustcmd{\tv}{\mathsf{tv}}
\newrobustcmd{\tvmax}{\mathsf{tv}_{\times}}
\newrobustcmd{\tvplus}{\mathsf{tv}_{\otimes}}
\newrobustcmd{\Kant}[1]{#1_{\mathrm{K}}}
\newrobustcmd{\Cpl}{\mathfrak{C}}
\newrobustcmd{\discrete}{d_{\top}}
\newrobustcmd{\fset}[1]{\underline{\mathbf{#1}}}
\newrobustcmd{\drel}[1]{\Delta_{#1}}
\newrobustcmd{\ndrel}[1]{\Delta^{\!\mathsf{c}}_{#1}}

\newrobustcmd{\qenrcats}[1]{{#1}\mathsf{HMet}}

\newrobustcmd{\ConvAlg}{\mathsf{Conv}}
\newrobustcmd{\HA}[1]{\mathsf{HA}_{#1}}
\newrobustcmd{\BarAlg}{\mathsf{Lib}}
\newcommand{\qU}{\mathcal{U}}

\newrobustcmd{\del}{\raisebox{0.3ex}{\scalebox{0.6}{\input{del.tikz}}}}
\newrobustcmd{\cc}{\raisebox{0.3ex}{\scalebox{0.6}{\input{cc.tikz}}}}
\newrobustcmd{\cop}{\raisebox{0.3ex}{\scalebox{0.6}{\input{cop.tikz}}}}
\newrobustcmd{\ccOneMinus}{\raisebox{0.3ex}{\scalebox{0.6}{\input{ccOneMinus.tikz}}}}

\newrobustcmd{\swap}{\mathsf{s}}
\newrobustcmd{\sigCA}{\Sigma_{\mathsf{CA}}}
\newrobustcmd{\img}[1]{\vcenter{\hbox{\includegraphics[scale=0.15]{#1}}}}
\newrobustcmd{\one}{I}

\makeatletter
\newrobustcmd{\bigplus}{%
	\DOTSB\mathop{\mathpalette\mattos@bigplus\relax}\slimits@
}
\newcommand\mattos@bigplus[2]{%
	\vcenter{\hbox{%
			\sbox\z@{$#1\sum$}%
			\resizebox{!}{0.9\dimexpr\ht\z@+\dp\z@}{\raisebox{\depth}{$\m@th#1+$}}%
	}}%
	\vphantom{\sum}%
}
\makeatother

\newrobustcmd{\syncat}[2]{\mathcal{S}_{\scriptstyle #1,#2}}
\newrobustcmd{\qsyncat}[2]{\widehat{\mathcal{S}}_{\scriptstyle #1,#2}}
\newrobustcmd{\sync}[1]{\mathcal{S}_{\scriptstyle #1}}
\newrobustcmd{\LT}[1]{\mathcal{L}_{\scriptstyle #1}}

\newrobustcmd{\add}{q}
\newrobustcmd{\zer}{j}
\newrobustcmd{\scalar}[1]{#1}

\newcommand{\diagbox}[3]{
	\begin{tikzpicture}
		\begin{pgfonlayer}{nodelayer}
			\node [style=basic box] (0) at (0, 0) {$#1$};
			\node [style=none] (1) at (1.5, 0) {};
			\node [style=none] (2) at (-1.5, 0) {};
			\node [style=none] (3) at (1.5, 0.5) {\scriptsize $#3$};
			\node [style=none] (4) at (-1.5, 0.5) {\scriptsize $#2$};
		\end{pgfonlayer}
		\begin{pgfonlayer}{edgelayer}
			\draw (2.center) to (0);
			\draw (0) to (1.center);
		\end{pgfonlayer}
	\end{tikzpicture}
}

\title{Quantitative Monoidal Algebra: Axiomatising Distance with String Diagrams} 

\titlerunning{Quantitative Monoidal Algebra} 

\author{Gabriele Lobbia}{Department of Computer Science,  Universit\`a di Bologna, Bologna, Italy}{}{https://orcid.org/0000-0002-2732-0317}{}

\author{Wojciech R\'{o}\.{z}owski}{Department of Computer Science, University College London, London, UK}{}{https://orcid.org/0000-0002-8241-7277}{}
\author{Ralph Sarkis}{Department of Computer Science, University College London, London, UK}{}{https://orcid.org/0000-0002-9037-2435}{}
\author{Fabio Zanasi}{Department of Computer Science, University College London, London, UK}{}{https://orcid.org/0000-0001-6457-1345}{}

\authorrunning{G. Lobbia, W. R\'{o}\.{z}owski, R. Sarkis, and F. Zanasi} 

\Copyright{John Q. Public and Joan R. Public} 



\acknowledgements{The authors would like to thank the organisers and the participants of the 2023 and 2025 Bellairs Workshop on \textit{Quantitative Logic and Reasoning} for helpful discussion. This research was supported by the ARIA Safeguarded AI programme and by the ERC grant Autoprobe (grant agreement 101002697).}

\nolinenumbers 

\hideLIPIcs  

\EventEditors{John Q. Open and Joan R. Access}
\EventNoEds{2}
\EventLongTitle{42nd Conference on Very Important Topics (CVIT 2016)}
\EventShortTitle{CVIT 2016}
\EventAcronym{CVIT}
\EventYear{2016}
\EventDate{December 24--27, 2016}
\EventLocation{Little Whinging, United Kingdom}
\EventLogo{}
\SeriesVolume{42}
\ArticleNo{23}

\begin{document}

\maketitle

\begin{abstract}
String diagrammatic calculi have become increasingly popular in fields such as quantum theory, circuit theory, probabilistic programming, and machine learning, where they enable resource-sensitive and compositional algebraic analysis. Traditionally, the equations of diagrammatic calculi only axiomatise exact semantic equality. However, reasoning in these domains often involves approximations rather than strict equivalences.

In this work, we develop a quantitative framework for diagrammatic calculi, where one may axiomatise notions of distance between string diagrams. Unlike similar approaches, such as the quantitative theories introduced by Mardare et al., this requires us to work in a monoidal rather than a cartesian setting. We define a suitable notion of monoidal theory, the syntactic category it freely generates, and its models, where the concept of distance is established via enrichment over a quantale. To illustrate the framework, we provide examples from probabilistic and linear systems analysis.
\end{abstract}
\section{Introduction}

Traditionally, formal semantics models equivalence between programs as equality of their interpretation in a certain mathematical domain. A fundamental question is the one of \emph{axiomatisation}: finding a set of equations between programs that hold precisely when they are semantically equivalent. Such an axiomatisation allows reasoning about semantics purely by syntactic manipulation of programs: this offers a structured, scalable approach to designing protocols (such as refinement and optimisation), automatisation, and formal verification. 

In the last few decades, increasingly prominent paradigms of computation such as quantum theory, probabilistic programming, and deep learning have challenged formal semantics, as they demand reasoning about systems that are partially defined, approximate, or sensitive to perturbations. This has led to a resource-sensitive enhancement of the above picture, along two axes: quantitative semantics and monoidal syntax (string diagrams). In a sense, our work is about reconciling both perspectives. 

\textbf{Quantitative Semantics.} The limitations of `exact' semantics are particularly evident when modelling probabilistic computation~\cite{Panangaden2009,vBW2001,DBLP:conf/concur/BreugelW01,Larsen2011}: rather than asking if probabilistic programs $P$ and $Q$ yield the same outputs with the same probabilities, it is more informative to ask \emph{how far} $P$s behaviour is from $Q$s behaviour, according to a certain metric. Similar considerations apply to other research areas, such as differential privacy~\cite{Dwork06}, and approximate computing~\cite{Mittal2016}. To address this form of analysis, program semantics has embraced quantitative reasoning, leading to advancements in areas like bisimulation metrics~\cite{Desharnais99,vBWorrell2005}, coeffectful computation~\cite{Orchard2019}, program distances~\cite{Crubille15,DalLago2022b}, and quantitative rewriting~\cite{Gavazzo2023}. The focus of our contribution is on \emph{axiomatising} quantitative semantics, for which \emph{quantitative algebraic theories} are particularly relevant. In this line of work, initiated by Mardare et al.~\cite{Mardare2016}, equations of the form $s = t$ are replaced with judgments of the form $s =_{\varepsilon} t$, which should be read as: ``$s$ is at distance at most $\varepsilon$ from $t$''. Among their examples, the authors show complete axiomatisations of the total variation and Kantorovich--Wasserstein distances~\cite{Villani2009} between probability distributions. Quantitative algebraic theories have been developed extensively, including a variety theorem~\cite{Mardare2017}, sum and tensor of theories~\cite{Bacci2021}, higher-order extensions~\cite{DalLago2022}, and the development of significant examples such as Markov processes~\cite{Bacci2018a}. They were also generalised to the setting of categorical algebra in various ways~\cite{MiliusU19,Rosicky2023,Rosicky2024,Jurka2024,Adamek2022}.

\textbf{Cartesian vs.\ Monoidal Syntax.} A fundamental feature of the aforementioned approaches is that the syntax of programs (or, more generally, computational processes) is represented by terms of a \emph{cartesian} algebraic theory. The terminology is due to the usual categorical perspective on abstract algebra, initiated by Lawvere~\cite{Lawvere1963}: the `syntactic' category freely generated by an algebraic theory $(\Sigma,E)$ is a \emph{cartesian} category, and models are functors preserving the cartesian structure. Whereas abstractly being cartesian just means to have finite products, via Fox's theorem~\cite{Fox1976} this is equivalent to each object $X$ of the category having a `copy' and a `discard' map. If we interpret these objects as variables of our programs, or resources of our systems, the assumption of cartesianity means that these entities may be duplicated or eliminated at will. In other words, the theory is \emph{insensitive} to such resources.

These assumptions are unsuitable in many contexts. A notable example is quantum theory, with its `no-cloning' and `no-deleting' theorems~\cite{Wootters1982Single}. Also probabilistic computation is inherently non-cartesian: duplicating the outcome of a die roll is not the same as rolling that die twice. There are many more instances in computer science where algebraic modelling needs to be attentive to resource consumption, e.g.~in concurrency theory~\cite{ABRAMSKY200637} and cryptography~\cite{BroadbentK23}.  
	
These examples motivated the development of \emph{monoidal} algebra. Processes are studied in (symmetric) monoidal categories, which allow for algebraic reasoning but do not assume a cartesian structure, meaning consumption of resources (variables) becomes explicit in the theory. Because the fundamental operations of a monoidal category are \emph{sequential} and \emph{parallel} composition, process syntax is depicted two-dimensionally, as \emph{string diagrams}~\cite{Selinger_2010,PiedeleuZanasi2025}. The pictorial representation is not just aesthetically pleasing, but allows for a clearer understanding of how information flows and is exchanged within the process components. For these reasons, string diagrams have been applied in quantum theory~\cite{Coecke2017}, concurrency~\cite{Bonchi2019b}, probabilistic programming~\cite{Piedeleu2025b}, machine learning~\cite{Cruttwell2022,Wilson2022}, cyber-physical systems~\cite{Bonchi2015,Capucci2022}, and even areas further removed from computer science such as linguistics~\cite{Coecke2010,lambeq}, epidemiology~\cite{Libkind2022,Baez2023}, and chemistry~\cite{Lobski2023,Lobski2023b}. When it comes to \emph{axiomatising} semantics of string diagrammatic calculi, tools analogous to those of (non-quantitative) cartesian algebra are available, such as a notion of freely generated `syntactic' category~\cite{HylandPowerSketches,Baez2018,Bonchi2018} and of model~\cite{Bonchi2018}. 
	
\textbf{Towards Quantitative Monoidal Algebra.} Similarly to cartesian algebra, it has become apparent that monoidal algebra urges for a quantitative extension. There is an increasing body of work developing the theory of probabilistic processes and Bayesian reasoning in symmetric monoidal categories called Markov categories, see e.g.~\cite{Fritz2020,JacobsKZ21,LorenzTull-causalmodels}. Similarly, several categorical models for machine learning algorithms are being proposed, in which string diagrams play a major role, see e.g. the surveys~\cite{shiebler2021categorytheorymachinelearning,crescenzi2024categoricalfoundationdeeplearning}. However, quantitative analysis has received very limited attention so far---one such example is~\cite{Perrone2024}, which studies notions of mutual information in Markov categories via relative entropy of string diagrams. In quantum theory, the works~\cite{kissinger2017pictureperfectquantumkeydistribution,Breiner2019} use distances between string diagrams to express noise tolerance in quantum protocols, and \cite{HLarsen2021} studies distances between (quantum) channels represented within a monoidal theory. What all these instances are missing is an \emph{axiomatic framework} to reason about distance of string diagrams, playing a role analogous to the one served by quantitative algebra for cartesian computation.
	
\textbf{Our Contribution.} In this work, we lay the mathematical foundations of quantitative monoidal algebra. To capture a wider range of models, we develop our framework not just for real-valued metric spaces, but for the more general notion of spaces with distances valued in a quantale $\qV$~\cite{Lawvere73}. Examples include preorders, pseudometric spaces, ultrametric spaces, etc. We introduce the notion of (symmetric) $\qV$-quantitative monoidal theory $\qU$ as a triple $(\Sigma,E,E_q)$, where $(\Sigma,E)$ is a monoidal theory and $E_q$ is a set of $\qV$-quantitative equations, for which we use the same notation $=_{\varepsilon}$ introduced in~\cite{Mardare2016}. We present the construction of the freely generated syntactic category over $\qU$, as an enriched monoidal category $\sync{\qU}$. Morphisms of $\sync{\qU}$ are depicted as string diagrams, composable sequentially and in parallel. Distances between string digrams, induced by $E_q$, are modelled in the enrichment of $\sync{\qU}$. Being a `variable-free' approach, the interaction between the enrichment and the generating rules of string diagrams in the syntactic category poses additional challenges compared to the cartesian setting. 
The final piece of our foundations is a suitable notion of model, which is defined à la Lawvere, as enriched functors from the syntactic category to `semantic' categories. We are then able to conclude with an analogue of the completeness theorem of equational logic, for the rules of quantitative diagrammatic reasoning. Our last contribution is a more in-depth comparison with related work~\cite{Mardare2016,Rosicky2024,Power2005,Sarkis2024,Bonchi2018}, which clarifies the relationship between cartesian approaches and our monoidal framework.
	
We provide two basic examples for our framework, related to linear and probabilistic computation respectively. The first is an axiomatisation for matrices over an ordered semiring with entrywise ordering (\cref{sec:preoder-matrix}). The theory of matrices appears ubiquitously in monoidal algebra (see e.g.~\cite{Baez2018,Bonchi2019b,Bonchi2015,Coecke2017,zanasi:tel-01218015}), and the order enrichment naturally appears in many such research threads. In particular, matrices on the Boolean semiring $\{0,1\}$ represent relations, and the ordering is set-theoretic inclusion. 		Our second example (\cref{sec:qaxtv}) is an axiomatisation of discrete probabilistic processes with the total variation metric, a distance fundamental in optimisation, learning theory, statistical inference, etc. Such an example may be thought as the `monoidal version' of an analogous result in the cartesian setting~\cite[Section~8]{Mardare2016}. 

	\textbf{Synopsis.} \cref{sec:prelim} provides background on quantales and enriched category theory. \cref{sec:quant-mon-alg} contains our main theoretical contributions. In \cref{sec:smt}, we recall monoidal theories and string diagrams. In \cref{sec:inferencerules} we define quantitative monoidal theories and the construction of the freely generated syntactic categories. We prove the latter are monoidal enriched in \cref{sec:syncatenriched}. In \cref{sec:models}, we define enriched models and give a sufficient condition for a (classical) model to be enriched. \cref{sec:mat,sec:axiomatisetotalvariation} are devoted to the examples outlined above. We compare formally to related work in \cref{sec:related}, and conclude in \cref{sec:conclusion} with future work. The appendix contains complete proofs for our results.

	\section{Preliminaries}
	\label{sec:prelim}

		\textbf{Quantale-Valued Generalised Metric Spaces.}		Following Lawvere~\cite{Lawvere73}, we allow distances to be valued not just in the positive reals, but in any quantale. We now recall quantales, as well as hemimetric and pseudometric spaces.
	
	\begin{definition}
														A \emph{quantale} is a tuple $(\qV,\vle,\vten,k)$, where $(\qV,\vle)$ is a partial order that has all joins and meets (supremums and infimums), i.e.~a complete lattice, $(\qV,\vten,k)$ is a monoid, and $\vten$ is \emph{join-continuous}, that is, $a\vten\vJoin S = \vJoin_{x \in S} a\vten x$ for any $a\in\qV$ and subset $S \subseteq \qV$.

		We write $\vJoin S$ or $\vJoin_{x \in S} x$ for the join of a subset $S \subseteq \qV$. In particular, $\qV$ has a \emph{bottom~($\bot$)} and a \emph{top ($\top$)} element that satisfy $\bot = \vJoin \emptyset \sqsubseteq x \sqsubseteq \vJoin \qV = \top$ for any $x \in \qV$. We call a quantale \emph{integral} if the monoidal unit is the top element of the underlying lattice (i.e.~$k=\top$).
											\end{definition}
	Throughout the paper, we will make the assumption that quantales are integral. This is a common requirement when studying quantales. In \cref{sec:mat}, we use the \emph{Boolean quantale} $\twoQ$ consisting of two elements $\bot\vle\top$ with underlying monoid $(\twoQ,\vmeet,\top)$. In \cref{sec:axiomatisetotalvariation}, we use the \emph{Lawvere quantale} $\zeroinfQ$ where the underlying lattice is the interval $[0,\infty)$ with the reversed order extended with $\infty$ as a bottom element, and the monoid operation is addition.													
	\begin{remark}\label{rem:meet_quantale}
		Let $(\qV, \sqsubseteq)$ be a complete lattice. If meets distribute over infinite joins, that is, for any $x \in \qV$ and family $\{x_i\}_{i \in I}$ in $\qV$,
$x \vmeet {\vJoin_{i \in I} x_i} = \vJoin_{i \in I} \left( x \vmeet x_i \right)$,
		then $\qV$ is called \emph{infinitely join distributive (IJD)} (see e.g.~\cite{Dilworth52,Hoffmann1981}). 
					\end{remark}
	
	\begin{definition}\label{defn:vmets}
		Let $\qV$ be a quantale. A \emph{$\qV$-hemimetric space} $(X, d)$ consists of a set $X$ and a function $d \colon X \times X \to \qV$ satisfying, for all $x,y,z\in X$, $k \vle d(x,x)$ \textbf{(reflexivity)}, and $d(x,y) \vten d(y,z) \vle d(x,z)$ \textbf{(triangle inequality)}.
		We call $(X,d)$ a \emph{$\qV$-pseudometric space} if it additionally satisfies, for all $x,y \in X$, $d(x,y)=d(y,x)$ \textbf{(symmetry)}.
				A function $f\colon X \to Y$ between $\qV$-hemimetric spaces $(X,d_X)$ and $(Y,d_Y)$ is called \emph{nonexpansive} if for all $x,x' \in X$,
		\(d_X(x,x')  \vle d_Y(f(x),f(x')).\)
	\end{definition}

	\begin{example}\label{ex:preorder-as-hemimet}
								In order to make better sense of \cref{defn:vmets}, note that when considering $\qV=\zeroinfQ$ we get back the standard definition of hemimetric and pseudometric spaces with possibly infinite distances.
		Over other quantales, we recover well-known structures. 						Setting $\qV=\twoQ$, $\twoQ$-hemimetrics are preorders, while $\twoQ$-pseudometrics are equivalence relations, and nonexpansive maps are order/relation-preserving functions.

									\end{example}
	
	\begin{definition}
				We denote the category of $\qV$-hemimetric spaces and nonexpansive functions with $\qVCat$, and its full subcategory of $\qV$-pseudometric spaces with $\qVCatsym$.	\end{definition}
	
	Our primary goal is to study categories where morphisms have a distance between them. We will model this extra structure on hom-sets with enriched categories, and this requires us to provide a (symmetric) monoidal product of $\qV$-hemi/pseudometric spaces. 			We consider two monoidal products inspired from well-known products of real-valued metric spaces.

									\begin{example}		\label{ex:quantalic_sum_metric}
		Let $(\qV, \vten, k)$ be a commutative quantale (i.e.~$\vten$ is commutative) and $(X,d_X)$, $(Y,d_Y)$ be two $\qV$-hemimetric spaces. We define the \emph{sum hemimetric} $d_X \boxtimes_{\vten} d_Y$ on the cartesian product $X\times Y$ by $(d_X \boxtimes_{\vten} d_Y)((x,y), (x',y'))\coloneqq d_X(x,x') \vten d_Y(y,y')$. This yields a monoidal product defined by $(X,d_X) \boxtimes_{\vten} (Y,d_Y) \coloneqq (X \times Y, d_X \boxtimes_{\vten} d_Y)$ whose monoidal unit is $1_{\boxtimes} \coloneqq (\{\bullet\}, \top)$, where $\top(\bullet, \bullet) = \top$. 		The symmetries $\sigma_{X,Y} \coloneqq (x,y) \mapsto (y,x)$ are nonexpansive maps $(X,d_X) \boxtimes_{\vten} (Y,d_Y) \rightarrow (Y,d_Y) \boxtimes_{\vten} (X,d_X)$, and they make $\qVCat$ into a symmetric monoidal category. Since $\boxtimes_{\vten}$ preserves symmetry, $\qVCatsym$ is a full symmetric monoidal subcategory of $\qVCat$.
			\end{example}
	\begin{example}		\label{ex:quantalic_max_metric}
		If $(\qV, \oplus, k)$ is a quantale and $(V,\vle)$ is IJD (see \Cref{rem:meet_quantale}), then we can define the \emph{max hemimetric}. Given $(X, d_X), (Y,d_Y) \in \qVCat$, we define $(X,d_X) \boxtimes_{\vmeet} (Y,d_Y) \coloneqq (X \times Y, d_X \boxtimes_{\vmeet} d_Y)$, where $(d_X \boxtimes_{\vmeet} d_Y)((x,y), (x',y'))\coloneqq d_X(x,x') \vmeet d_Y(y,y')$. The monoidal unit is given by $1_{\boxtimes} \coloneqq (\{\bullet\}, \top)$. Once again, the evident symmetries are nonexpansive, and we get another symmetric monoidal structure on $\qVCat$. This definition also restricts to a symmetric monoidal product on $\qVCatsym$.
	\end{example}

		\noindent\textbf{Categories Enriched over \texorpdfstring{$\qVCat$}{VHMet}.} In this paper, distances between morphisms of monoidal categories will be cast in terms of \emph{enriched categories}. 	We are only interested in categories enriched over hemi/pseudometric spaces, so we defer to \cite{KellyBook} for the general details.
	
	We will work with categories enriched in $(\qVCat, \boxtimes, 1_{\boxtimes})$ or $(\qVCatsym, \boxtimes, 1_{\boxtimes})$, where $\boxtimes$ is defined as $\boxtimes_{\vten}$ in \Cref{ex:quantalic_sum_metric} or $\boxtimes_{\vmeet}$ in \Cref{ex:quantalic_max_metric}. To be explicit but concise on which base of enrichment we are considering, we will use the notations $\qVCatsum$, $\qVCatmax$, $\qVCatsymsum$, and $\qVCatsymmax$. Enrichment over these categories boils down to equipping hom-sets with hemi/pseudometrics and requiring a nonexpansiveness property of composition. In other words, any $\qVCat$-enriched category is determined by an underlying category $\C$, where every hom-set $\C(a,b)$ has a $\qV$-hemimetric space structure $(\C(a,b), d_{a,b})$ such that for all $f,f' \in \C(a,b)$ and $g,g' \in \C(b,c)$, we have that $(d_{b,c} \boxtimes d_{a,b})((g,g'),(f,f')) \vle d_{a,c}(g \circ f, g' \circ f')$.

	Furthermore, any $\qVCat$-functor $F\colon \C \rightarrow \D$ is determined by a functor between the underlying categories, which is locally nonexpansive, in the sense that the assignment $f \mapsto Ff$ is a nonexpansive map $\C(a,b) \rightarrow \D(Fa,Fb)$ for all $a,b \in \Ob(\C)$. Similarly, an enriched isomorphism $F\colon \C \rightarrow \D$ is an isomorphism between the underlying categories, which in addition is locally an isometry, namely, the assignment $f \mapsto Ff$ is an isometry $\C(a,b) \rightarrow \D(Fa,Fb)$ for all $a,b \in \Ob(\C)$. We carefully develop these claims in \cref{app:proofs-prelim}, and we note that they can also be applied to enrichment over $\qVCatsym$. 	
	
	This concrete characterisation of enrichment allows us to give a convenient definition of enriched monoidal categories for our purposes. It instantiates the more general definition that appears in e.g. \cite[Definition~2.1]{Morrison2017}, \cite[Definition~4.1]{KongLiang2024}. Morally, we define enriched monoidal categories to be monoidal ($\Set$-)categories equipped with hemi/pseudometrics on their hom-sets such that both composition and monoidal product are nonexpansive.	\begin{definition}\label{def:enrichedmoncat}
		A \emph{$\qVCat$-enriched symmetric monoidal category} $\C$ is a category that is both symmetric monoidal and $\qVCat$-enriched, and such that the bifunctor $\otimes \colon \C \times \C \to \C$ is a $\qVCat$-functor. 		It is called \emph{strict} if the underlying monoidal category is strict. A \emph{$\qVCat$-enriched symmetric strict monoidal functor} $F\colon\C\to\D$ is a strict monoidal functor between the underlying monoidal categories which is also $\qVCat$-enriched (as a functor).
	\end{definition}
	Unrolling this definition according to our discussion above, a $\qVCat$-enriched symmetric strict monoidal category (SMC) is just a $\qVCat$-category whose underlying category $\C$ is equipped with a symmetric strict monoidal product $\otimes \colon \C \times \C \to \C$ that is nonexpansive with respect to $\boxtimes$, in the sense that for all $a,b,c,d \in \Ob(\C)$, $f,f' \in \C(a,b)$, $g,g' \in\C(c,d)$, we have $(d_{a,b} \boxtimes d_{c,d})((f,f'),(g,g'))\sqsubseteq d_{a \otimes c, b \otimes d}(f \otimes g, f'\otimes g')$. All the above also applies to $\qVCatsym$.

	\section{Quantitative Monoidal Algebra}
	\label{sec:quant-mon-alg}
	
	In this section we fix a commutative integral quantale $\qV$ and introduce the notion of $\qV$-quantitative symmetric monoidal theory.  We recall ordinary monoidal theories first, following \cite{Bonchi2018}. We omit the adjective `symmetric' for monoidal theories, as it will be assumed.

	\subsection{Background: Monoidal Theories}\label{sec:smt}

	\begin{definition}		A \emph{monoidal signature} $\Sigma$ is a set of \emph{generators}, each with an arity $n \in\bN$ and a coarity $m \in \bN$, which we often indicate simply with a type $n \to m$. In preparation to representing $\Sigma$-terms using string diagrams, we adopt a graphical representation for generators, as boxes with dangling wires on the left and the right to indicate arity and coarity. For instance, $\raisebox{0.3ex}{\scalebox{.45}{\input{ex-operation.tikz}}}$ has arity $2$ and coarity $3$.  We write $\raisebox{0.3ex}{\scalebox{.7}{\input{generator-nm.tikz}}}$ for a generic generator with arity $n$ and coarity $m$. We use $\raisebox{0.3ex}{\scalebox{.7}{\input{generator-notype.tikz}}}$ when the type is irrelevant or clear from context.
		
		The set of \emph{$\Sigma$-terms} (and their (co)arities) is defined inductively as follows:
		\begin{itemize}[nosep]
			\item all generators $\raisebox{0.3ex}{\scalebox{0.7}{\input{generator-nm.tikz}}} \in\Sigma$, $\raisebox{0.3ex}{\scalebox{0.7}{\input{id.tikz}}} \colon 1\to 1$, $\raisebox{0.3ex}{\scalebox{0.7}{\input{empty-diag.tikz}}} \colon 0 \rightarrow 0$, and $\raisebox{0.3ex}{\scalebox{.6}{\input{sym.tikz}}}\colon 2\to 2$ are $\Sigma$-terms;
			\item if $s\colon n\to m$ and $t\colon m\to \ell$ are $\Sigma$-terms, then $s ; t \colon n\to \ell$ is a $\Sigma$-term;
			\item if $t \colon n\to n'$ and $s\colon m\to m'$ are $\Sigma$-terms, then $t\tensor s\colon n+m\to n'+m'$ is a $\Sigma$-term. 
		\end{itemize}
		
		A \emph{monoidal theory} $(\Sigma,E)$ consists of a signature $\Sigma$ and a set $E$ of pairs $(s,t)$ of $\Sigma$-terms of the same type, which we call \emph{equations} and write $s=t$.
	\end{definition}
	We represent $\Sigma$-terms graphically using the same conventions introduced for the generators. Given $\Sigma$-terms $s \colon n \to m$ and $t \colon m \to \ell$, we write $s ; t$ as $\raisebox{0.3ex}{\scalebox{0.7}{\input{horizontal-comp.tikz}}}$. Similarly, given $t \colon n\to n'$ and $s\colon m\to m'$, we write $t \tensor s$ as $\raisebox{0ex}{\scalebox{0.7}{\input{vertical-comp.tikz}}}$. 
	Arbitrary identities $\scalebox{.7}{\input{idn.tikz}} \colon n \to n$ and symmetries $\raisebox{0ex}{\scalebox{.7}{\input{symmn.tikz}}} \colon m+n \to n+m$ may be defined as $\Sigma$-terms, by pasting together in the expected way copies of the `basic' identity $\scalebox{.6}{\input{id.tikz}}$ and symmetry $\scalebox{.6}{\input{sym.tikz}}$. 
	
	When organised into a category (\cref{def:freesmc} below), operations $;$ and $\tensor$ become associative and obey the so-called `exchange law', meaning we can paste together diagrams without worrying about priority of application. $\Sigma$-terms modulo the axioms of symmetric strict monoidal categories (drawn in \Cref{App:SMC}) are called \emph{string diagrams}, see e.g.~\cite{Selinger_2010,PiedeleuZanasi2025}.
		
	\begin{definition}\label{def:freesmc}
		The symmetric strict monoidal category (SMC) $\syncat{\Sigma}{E}$ freely generated by $(\Sigma,E)$, called the \emph{syntactic category}, is defined as follows. Its objects are natural numbers. A morphism $n \rightarrow m$ is a $\Sigma$-term of arity $n$ and coarity $m$ modulo the equations in $E$ and the axioms of SMCs. Formally, two $\Sigma$-terms $s$ and $t$ are equal in $\syncat{\Sigma}{E}$ if and only if they are in the same equivalence class of the smallest congruence (with respect to $;$ and $\tensor$) that contains the pairs in $E$ and the axioms of SMCs. Monoidal product on objects is given by addition. Regarding morphisms, composition, monoidal product, identities, and symmetries are defined by their counterparts on $\Sigma$-terms.
	\end{definition}

	\subsection{Quantitative Monoidal Theories}\label{sec:inferencerules}
	
	In the envisioned applications of our work (and in \cref{sec:mat,sec:axiomatisetotalvariation}), the string diagrams represent processes for which equality is too coarse a relation to be meaningful. 	To achieve a finer comparison, we reuse a central idea in \cite{Mardare2016}, that is to replace the `exact' equality relation with equality `up to' some quantity $\varepsilon$ in $\qV$. This new relation is denoted with $=_{\varepsilon}$, and  $s =_{\varepsilon} t$ means that the processes represented by $s$ and $t$ are at distance at most $\varepsilon$.
		A quantitative monoidal theory is a monoidal theory with additional axioms of this shape.	
	\begin{definition}		\label{def:quantitative-monoidal-theory}
	A \emph{$\qV$-quantitative (symmetric) monoidal theory} is a triple $(\Sigma,E,E_q)$, where $(\Sigma,E)$ is a monoidal theory, and $E_q$ is a set of triples $(s,t,\varepsilon)$ comprising two $\Sigma$-terms $s$ and $t$, and an element $\varepsilon \in V$, which we call \emph{quantitative equations} and denote with $s=_{\varepsilon} t$.
			\end{definition}
	We want to construct a syntactic category associated to a $\qV$-quantitative monoidal theory $(\Sigma,E,E_q)$. It must have extra structure describing distance between morphisms, thus it will be $\qVCat$-enriched or $\qVCatsym$-enriched. To this aim, we will start with $\syncat{\Sigma}{E}$, the SMC freely generated by the underlying monoidal theory, and define a $\qV$-hemimetric (or pseudometric) on all the hom-sets of $\syncat{\Sigma}{E}$ making sequential and parallel compositions nonexpansive. 	
		
	Just like equality between $\Sigma$-terms in $\syncat{\Sigma}{E}$ was inferred from the equations in $E$ and the axioms of SMCs, the distance between $\Sigma$-terms will be inferred from the quantitative equations in $E_q$ and the axioms of enriched SMCs. This process is more involved than building the smallest congruence, but similar in spirit. It is also inspired from quantitative equational logic in \cite{Mardare2016}. Note the infinitary \RuleJoin\ rule that mirrors the \textbf{Arch} rule of \cite{Mardare2016}.
				
	\begin{definition}
		\label{def:qnt-closure}
		Let $(\Sigma,E,E_q)$ be a $\qV$-quantitative monoidal theory and $\syncat{\Sigma}{E}$ be the SMC generated by $(\Sigma,E)$. 
		We define the $\qVCatsum$-closure of $E_q$, denoted $\Hclosesumsum{E_q}$, as the smallest set of quantitative equations containing $E_q$ and closed under the following inference rules.
		\def\defaultHypSeparation{\hskip0in}
		\def\labelSpacing{1pt}
		\begin{itemize}[nosep]
			\item For any $\Sigma$-terms $f,g,h\colon n \rightarrow m$, and $\varepsilon,\varepsilon' \in \qV$, we have the following rules. They ensure that the distances defined later in \cref{lem:loc-hem-space} are $\qV$-hemimetrics on the hom-sets of $\syncat{\Sigma}{E}$.
			{\small
			\begin{gather*}
				\begin{bprooftree}
					\AxiomC{\scalebox{0.85}{$\input{diagf.tikz} = \input{diagg.tikz}$} { \scriptsize is provable from $E$}}
					\RightLabel{\hypertarget{rulerefl}{$\textnormal{\textsc{Refl}}$}}
					\UnaryInfC{\scalebox{0.85}{$\input{diagf.tikz} =_{\top} \input{diagg.tikz}$}}
				\end{bprooftree} \begin{bprooftree}
					\AxiomC{\scalebox{0.85}{$\input{diagf.tikz} =_{\varepsilon} \input{diagg.tikz}$}}
					\AxiomC{\scalebox{0.85}{$\input{diagg.tikz} =_{\varepsilon'} \input{diagh.tikz}$}}
					\RightLabel{\hypertarget{ruletriang}{$\textnormal{\textsc{Triang}}$}}
					\BinaryInfC{\scalebox{0.85}{$\input{diagf.tikz} =_{\varepsilon \vten \varepsilon'}\input{diagh.tikz}$}}
				\end{bprooftree}
						\\ 
					\begin{bprooftree}
			\AxiomC{$\phantom{E}$}
			\RightLabel{\hypertarget{rulebot}{$\textnormal{\textsc{Bot}}$}}
			\UnaryInfC{\scalebox{0.85}{$\input{diagf.tikz} =_{\bot} \input{diagg.tikz}$}}
		\end{bprooftree}
				\begin{bprooftree}
					\AxiomC{\scalebox{0.85}{$\input{diagf.tikz} =_{\varepsilon} \input{diagg.tikz}\ \varepsilon' \vle \varepsilon$}}
										\RightLabel{\hypertarget{ruletmon}{$\textnormal{\textsc{Mon}}$}}
					\UnaryInfC{\scalebox{0.85}{$\input{diagf.tikz} =_{\varepsilon'}\input{diagg.tikz}$}}
				\end{bprooftree}\ 
				\begin{bprooftree}
					\AxiomC{\scalebox{0.85}{$\input{diagf.tikz} =_{\varepsilon_i} \input{diagg.tikz}\ \forall i \in I$}}
										\RightLabel{\hypertarget{rulejoin}{$\textnormal{\textsc{Join}}$}}
					\UnaryInfC{\scalebox{0.85}{$\input{diagf.tikz} =_{\vjoin_{i}\varepsilon_i}\input{diagg.tikz}$}}
				\end{bprooftree}
			\end{gather*}
			}
			\item For any two pairs of composable $\Sigma$-terms $(f_0,g_0)$ and $(f_1,g_1)$ and any $\varepsilon,\varepsilon' \in \qV$, the rule \RuleSeqsum\ ensures that the syntactic category is $\qVCatsum$-enriched.																					\item For any two pairs of $\Sigma$-terms $(f_0,f_1)$ and $(g_0,g_1)$ with matching arities and any $\varepsilon,\varepsilon' \in \qV$, the rule \RuleTenssum\ ensures that the syntactic category is monoidal $\qVCatsum$-enriched.			{\small	
			\[\begin{bprooftree}
				\AxiomC{\scalebox{0.85}{$\input{diagf0.tikz} =_{\varepsilon} \input{diagf1.tikz}$}}
				\AxiomC{\scalebox{0.85}{$\input{diagg0.tikz} =_{\varepsilon'} \input{diagg1.tikz}$}}
				\RightLabel{\hypertarget{ruleseqsum}{$\textnormal{\textsc{Seq}}_{\vten}$}}
				\BinaryInfC{\scalebox{0.85}{$\input{seqletters.tikz} =_{\varepsilon\vten \varepsilon'} \input{seqletterstwo.tikz}$}}
			\end{bprooftree} \quad \begin{bprooftree}
				\AxiomC{\scalebox{0.85}{$\input{diagf0.tikz} =_{\varepsilon} \input{diagf1.tikz}$}}
				\AxiomC{\scalebox{0.85}{$\input{diagg0.tikz} =_{\varepsilon'} \input{diagg1.tikz}$}}
				\RightLabel{\hypertarget{ruletenssum}{$\textnormal{\textsc{Par}}_{\vten}$}}
				\BinaryInfC{\scalebox{0.85}{$\input{parletters.tikz} =_{\varepsilon\vten\varepsilon'} \input{parletterstwo.tikz}$}}
			\end{bprooftree}\]
			}
		\end{itemize}

	\end{definition}
	\begin{remark}
		If $(\qV,\vle)$ is IJD (\cref{rem:meet_quantale}), then we can consider enriching $\syncat{\Sigma}{E}$ over $\qVCat$ with the monoidal product $\boxtimes_{\vmeet}$ from \cref{ex:quantalic_max_metric}. This requires defining a different closure of $E_q$, that we denote with $\Hclosemaxmax{E_q}$. It is the smallest set of quantitative equations containing $E_q$ and closed under the inference rules above, but \RuleSeqsum\ and \RuleTenssum\ are replaced by \RuleSeqmeet\ and \RuleTensmeet\ below. These ensure that the syntactic category is $\qVCatmax$-enriched monoidal.		\def\defaultHypSeparation{\hskip0in}
		\def\labelSpacing{1pt}
		{\small
		\begin{gather*}			
			\begin{bprooftree}
				\AxiomC{\scalebox{0.85}{$\input{diagf0.tikz} =_{\varepsilon} \input{diagf1.tikz}$}}
				\AxiomC{\scalebox{0.85}{$\input{diagg0.tikz} =_{\varepsilon'} \input{diagg1.tikz}$}}
				\RightLabel{\hypertarget{ruleseqmeet}{$\textnormal{\textsc{Seq}}_{\vmeet}$}}
				\BinaryInfC{\scalebox{0.85}{$\input{seqletters.tikz} =_{\varepsilon\vmeet \varepsilon'} \input{seqletterstwo.tikz}$}}
			\end{bprooftree}\ 
			\begin{bprooftree}
				\AxiomC{\scalebox{0.85}{$\input{diagf0.tikz} =_{\varepsilon} \input{diagf1.tikz}$}}
				\AxiomC{\scalebox{0.85}{$\input{diagg0.tikz} =_{\varepsilon'} \input{diagg1.tikz}$}}
				\RightLabel{\hypertarget{ruletensmeet}{$\textnormal{\textsc{Par}}_{\vmeet}$}}
				\BinaryInfC{\scalebox{0.85}{$\input{parletters.tikz} =_{\varepsilon\vmeet\varepsilon'} \input{parletterstwo.tikz}$}}
			\end{bprooftree}
		\end{gather*}
		}
		\def\defaultHypSeparation{\hskip0.2in}\def\labelSpacing{3pt}
	\end{remark}

		\begin{remark}\label{rem:mixed_rules}
		There are two additional possible closures of $E_q$ which we can consider, motivated by the example in \cref{sec:axiomatisetotalvariation}. Let $\Hclosemaxsum{E_q}$ be the closure of $E_q$ under the same inference rules as $\Hclosesumsum{E_q}$ except \RuleSeqsum\ is replaced by \RuleSeqmeet. Similarly, let $\Hclosesummax{E_q}$ be the closure of $E_q$ under the same inference rules as $\Hclosesumsum{E_q}$ except \RuleTenssum\ is replaced by \RuleTensmeet.

				Since $\qV$ is integral, $\forall a,b \in \qV, a \oplus b \sqsubseteq a \sqcap b$, so we can infer that the rules \RuleSeqmeet\ and \RuleTensmeet\ are tighter than \RuleSeqsum\ and \RuleTenssum\ (respectively). Namely, any quantitative equation in $\Hclosesumsum{E_q}$ belongs to $\Hclosemaxsum{E_q}$ and $\Hclosesummax{E_q}$, and any quantitative equation in $\Hclosemaxsum{E_q}$ or $\Hclosesummax{E_q}$ also belongs to $\Hclosemaxmax{E_q}$. Consequently, $\Hclosesummax{E_q}$ can be used to enrich $\syncat{\Sigma}{E}$ over $\qVCatsum$,  
														where the monoidal product satisfies an additional nonexpansiveness property. 		We make use of this in \cref{sec:totalvarenrichment}. In the case of $\Hclosemaxsum{E_q}$, $\syncat{\Sigma}{E}$ will be enriched over $\qVCatmax$, but monoidal enriched only over $\qVCatsum$.

	\end{remark}
	
	\begin{remark}
		We can also enforce the distance between morphisms to satisfy the symmetry property, so that $\syncat{\Sigma}{E}$ will be enriched over $\qVCatsym$. It suffices to add the following rule.
		{\small
		\[\begin{bprooftree}
			\AxiomC{\scalebox{0.85}{$\input{diagf.tikz} =_{\varepsilon} \input{diagg.tikz}$}}
			\RightLabel{\hypertarget{rulesym}{\textsc{Symm}}}
			\UnaryInfC{\scalebox{0.85}{$\input{diagg.tikz}=_{\varepsilon} \input{diagf.tikz}$}}
		\end{bprooftree}\]}We write $\Pclosesumsum{E_q}$, $\Pclosemaxsum{E_q}$, $\Pclosesummax{E_q}$ and $\Pclosemaxmax{E_q}$ the corresponding closures of $E_q$ after adding \RuleSym.		
					\end{remark}
		There is a total of eight possible closures depending on the choice of inference rules considered. They are summarised in \cref{tab:closures-quant-eq}. Our logic offers this flexibility to allow axiomatisation of various examples. We use the closures $\Hclosesumsum{E_q}$ and $\Pclosesummax{E_q}$ respectively in \cref{sec:preoder-matrix,sec:qaxtv}.

	\subsection{Enrichment of the Syntactic Category}\label{sec:syncatenriched}

	In this section, we will show how to use the inference rules in \Cref{def:qnt-closure} (or more precisely the different closures they induce) to define enrichments of the syntactic category $\syncat{\Sigma}{E}$. We work with a generic closure $\anyclose{E_q}$ that can be instantiated with any closure in \cref{tab:closures-quant-eq}. We explicitly mention what inference rules are needed to show each item. 	
	First, we equip each hom-set of $\syncat{\Sigma}{E}$ with a $\qV$-hemimetric, essentially mirroring the definition of $d_{\qU}$ in \cite[Section~5]{Mardare2016}. Note that \RuleRefl, \RuleTriang\ and \RuleSym\ correspond, respectively, to reflexivity, triangle inequality, and symmetry for $d^\qU_{n,m}$.
	
	\begin{lemma}
		\label{lem:loc-hem-space}
		Let $\qU=(\Sigma,E,E_q)$ be a $\qV$-quantitative monoidal theory. For any $n,m \in \bN$ and $\Sigma$-terms $f,g\colon n \rightarrow m$, let $ d^{\qU}_{n,m}(f,g) \coloneqq \vJoin \{ \varepsilon \mid f =_{\varepsilon} g \in\anyclose{E_q} \}$. 
										This defines a $\qV$-hemimetric on $\syncat{\Sigma}{E}(n,m)$, which is a $\qV$-pseudometric if \RuleSym\ was used in the closure $\anyclose{E_q}$.					\end{lemma}

		Now that $\syncat{\Sigma}{E}$ is equipped with $\qV$-hemimetrics (resp.~pseudometrics) on its hom-sets, we show it is monoidal enriched over $\qVCat$ (resp.~$\qVCatsym$). This relies on two lemmas showing that sequential and parallel composition are nonexpansive.				
				\begin{lemma}
			\label{lem:subst-non-exp}
			Let $f_0, f_1\colon n \rightarrow m$ and $g_0,g_1\colon m \rightarrow \ell$ be $\Sigma$-terms. When $\ast$ is $\vten$ and $\anyclose{E_q}$ is closed under \RuleSeqsum, or when $\ast$ is $\vmeet$ and $\anyclose{E_q}$ is closed under \RuleSeqmeet, the following holds: 			$d^\qU_{n,m}(f_0,f_1)\ast d^\qU_{m,\ell}(g_0,g_1) \vle d^\qU_{n,\ell}( f_0;g_0\,,\,f_1;g_1).$
		\end{lemma}
						
		\begin{lemma}\label{lem:tens-non-exp}
			Let $f_0,g_0\colon n \rightarrow n'$ and $f_1,g_1\colon m \rightarrow m'$ be $\Sigma$-terms. When $\ast$ is $\vten$ and $\anyclose{E_q}$ is closed under \RuleTenssum, or when $\ast$ is $\vmeet$ and $\anyclose{E_q}$ is closed under \RuleTensmeet, the following holds: 
			$d^\qU_{n,n'}(f_0,f_1)\ast d^\qU_{m,m'}(g_0,g_1)\vle d^\qU_{n+m,n'+m'}(f_0\tensor g_0,f_1\tensor g_1).$
		\end{lemma}

														By the characterisation of $\qVCat$ enrichment in \cref{sec:prelim}, we conclude that the category of string diagrams equipped with the distances $d^{\qU}_{n,m}$ is an enriched SMC.
								
		\begin{proposition}\label{def:synt-enr-cat}
			Let $\qU = (\Sigma,E,E_q)$ be a $\qV$-quantitative monoidal theory, and $\anyclose{E_q}$ be one of the closures defined above. Then, $\syncat{\Sigma}{E}$ equipped with the $\qV$-hemimetrics defined in \cref{lem:loc-hem-space} is a $\VV$-enriched SMC, where $\VV$ is the base of enrichment corresponding to $\anyclose{E_q}$ in \cref{tab:closures-quant-eq}. We denote this enriched category with $\sync{\qU}$ to distinguish it from its underlying category $\syncat{\Sigma}{E}$.
		\end{proposition}
				
																	The inference rules are central to the definition of the syntactic category, and it will be convenient for us to reify them in other categories through the notion of validity.
		\begin{definition}\label{defn:soundness}
			Let $\C$ be an SMC equipped with $\qV$-hemimetrics $(\C(a,b), d^{\C}_{a,b})$ on each of its hom-sets. The rules \RuleRefl, \RuleBot, \RuleTriang, \RuleMon, and \RuleJoin\ are \emph{valid} in $\C$. \RuleSym\ is \emph{valid} if all $d^{\C}_{a,b}$ are $\qV$-pseudometrics. \RuleSeqsum\ is \emph{valid} if $\scalebox{1}{${;}\colon\C(a,b) \boxtimes_{\vten} \C(b,c) \rightarrow \C(a,c)$}$ (sequential composition) is nonexpansive. \RuleTenssum\ is \emph{valid} if ${\tensor}\colon\C(a,a') \boxtimes_{\vten} \C(b,b') \rightarrow \C(a\otimes a',b\otimes b')$ (parallel composition) is nonexpansive. \RuleSeqmeet\ is \emph{valid} if $\scalebox{1}{${;}\colon\C(a,b) \boxtimes_{\vmeet} \C(b,c) \rightarrow \C(a,c)$}$ is nonexpansive. \RuleTensmeet\ is \emph{valid} if $\scalebox{1}{${\tensor}\colon\C(a,a') \boxtimes_{\vmeet} \C(b,b') \rightarrow \C(a\otimes a',b\otimes b')$}$ is nonexpansive.
																													\end{definition}

		\subsection{Models}\label{sec:models}
		Defining syntactic categories allows us to study models as functors à la Lawvere~\cite{Lawvere1963}. Recall that a model of a monoidal theory $(\Sigma,E)$ is a symmetric strict monoidal functor from $\syncat{\Sigma}{E}$ to another SMC $\C$ (\emph{cf.}~\cite{Bonchi2018}). Central to this approach is the fact that because models must preserve the structure used to generate $\syncat{\Sigma}{E}$, they are entirely determined by their action on the generators. Also, one may check that an assignment of the generators from $\Sigma$ into $\C$ extends to a model simply by verifying that the equations in $E$ are satisfied. 		
		Our goal in this section is to define models of quantitative monoidal theories as functors from the syntactic categories, prove that they are determined by their action on generators, and finally give sufficient conditions for when a model of a monoidal theory can be enriched.

					Because there are multiple syntactic categories that can be constructed from a quantitative monoidal theory (depending on the inference rules that are invoked), there are different notions of models. We can unify their definition using the notion of validity (\cref{defn:soundness}).
		
						\begin{definition}
			\label{def:models-quant-theory}
			Let $\qU = (\Sigma,E,E_q)$ be a $\qV$-quantitative monoidal theory, and $\anyclose{E_q}$ be one of the closures listed in \cref{tab:closures-quant-eq}. An $\anyclose{E_q}$-model of $\qU$ is a $\qVCatsum$-enriched SMC $\C$ wherein all the inference rules used to generate $\anyclose{E_q}$ are valid, along with a strict monoidal $\qVCatsum$-functor $M\colon \sync{\qU} \rightarrow \C$, where $\sync{\qU}$ is constructed according to \cref{def:synt-enr-cat}.
																																															\end{definition}

								Independently of the choice of inference rules, the underlying category of $\sync{\qU}$ is always $\syncat{\Sigma}{E}$, where $(\Sigma,E)$ is the underlying monoidal theory of $\qU$. Therefore, any model of $\qU$ is always built on top of a model of $(\Sigma,E)$. The enrichment is merely a property on a strict monoidal functor $M\colon \syncat{\Sigma}{E} \rightarrow \C$. In analogy to how assignments on the generators of $\Sigma$ can be extended to models of $(\Sigma,E)$ when they satisfy $E$, we can give a sufficient condition, in terms of the quantitative equations in $E_q$, for $M$ to be an enriched model.

		\begin{definition}
			\label{def:sat-quant-eq}
			Let $\qU = (\Sigma,E,E_q)$ be a $\qV$-quantitative monoidal theory, $\sync{\qU}$ be constructed according to \cref{def:synt-enr-cat} with a closure $\anyclose{E_q}$, and $\C$ be a $\qVCatsum$-enriched SMC wherein all the inference rules used to generate $\anyclose{E_q}$ are valid.
															A quantitative equation $f =_{\varepsilon} g$ is \emph{true} in a model of $(\Sigma,E)$, $M\colon \syncat{\Sigma}{E} \rightarrow \C$, if $\varepsilon \vle d^{\C}(Mf,Mg)$,
			where $d^{\C}$ is the $\qV$-hemimetric (or pseudometric) on the hom-sets of $\C$.
		\end{definition}
		
		This definition allows us to define models of $(\Sigma,E,E_q)$ from certain models of $(\Sigma,E)$.
								
		\begin{theorem}
			\label{thm:model-to-enriched-model}
									Let $M$ be a model of a monoidal theory $(\Sigma,E)$. If all the quantitative equations in $E_q$ are true in $M$, then it is an $\anyclose{E_q}$-model of $(\Sigma,E,E_q)$. In particular, $M$ is an enriched functor.
		\end{theorem}

	The theorem below is an analogue of the completeness theorem of equational logic adapted to the case of quantitative monoidal reasoning. It intuitively means that the rules of quantitative monoidal reasoning suffice for proving all the quantitative equations that hold generally in the semantics. The proof relies on the canonical model $\id\colon \sync{\qU} \rightarrow \sync{\qU}$.	
	\begin{theorem}\label{thm:completeness}
		Let $\qU=(\Sigma,E,E_q)$ be a $\qV$-quantitative monoidal theory, $\anyclose{E_q}$ be a closure from \cref{tab:closures-quant-eq}, and $f,g \colon n \to m \in \sync{\qU}$. If $f =_{\varepsilon} g$ is true in all $\anyclose{E_q}$-models $M \colon \sync{\qU} \to \C$ of the theory $\qU$, meaning that $\varepsilon \vle d^{\C}(M(f), M(g))$, then $f =_{\varepsilon} g$ is in the closure $\anyclose{E_q}$. 
	\end{theorem}

	\section{Case Study I: Order on Matrices}
	\label{sec:mat}

		Throughout the section we fix a semiring $R$ (a ring without additive inverses) and write $1_R$ and $0_R$ respectively for its multiplicative and additive identities. Also, we write $\fset{n}$ for $\{0,\dots,n-1\}$. We will show that, when $R$ is ordered,  the entrywise ordering of $R$-matrices can be axiomatised using the framework of \cref{sec:quant-mon-alg}. A notable example is the Boolean semiring $R=\{0,1\}$, in which case $R$-matrices represent relations, and the order we axiomatise is set-theoretic inclusion.

	\subsection{Background: (non-Quantitative) Axiomatisation of \texorpdfstring{$\Mat_R$}{Mat\_R}}
	\label{sec:mon-th-matrices}
	
	Before considering the quantitative theory, we recall the category of $R$-matrices and the monoidal theory axiomatising it. The axiomatisation result seems to be folklore, see e.g.~\cite{LafontCombinators,HylandPowerSketches} for the Boolean case. We follow the presentation of~\cite[Section~3.2]{zanasi:tel-01218015}, which is for a generic ring, but is applicable for semirings as well.
	
	\begin{definition}		The SMC $\Mat_R$ has objects $\bN$, and morphisms $n \rightarrow m$ the $m\times n$ matrices with entries in $R$. When $n$ or $m$ is $0$, there is a unique empty $m\times n$ matrix $[]$. Composition is by matrix multiplication and the monoidal product is by direct sum: $A \oplus A' = \left[\begin{smallmatrix}
			A & \mathbf{0} \\ \mathbf{0} & A'
		\end{smallmatrix}\right]$.   																												
		
			The monoidal theory $\HA{R}$ (\emph{H}opf \emph{A}lgebras) has generators
					$\raisebox{0.2ex}{\scalebox{0.5}{\input{delete.tikz}}} \colon 1 \rightarrow 0$,			$\raisebox{0.2ex}{\scalebox{0.5}{\input{copy.tikz}  }} \colon 1 \rightarrow 2$,			$\raisebox{0.2ex}{\scalebox{0.5}{\input{add.tikz}   }} \colon 2 \rightarrow 1$,			$\raisebox{0.2ex}{\scalebox{0.5}{\input{zer.tikz}   }} \colon 0 \rightarrow 1$, and 			$\raisebox{0.2ex}{\scalebox{0.5}{\input{scalar.tikz}}} \colon 1 \rightarrow 1$,
				for each scalar $k \in R$, and equations 
			\begin{gather*}
				\scalebox{0.45}{\input{a1_addassoc.tikz}} \quad \scalebox{0.5}{\input{a2_addcomm.tikz}} \quad \scalebox{0.45}{\input{a4_copassoc.tikz}} \quad \scalebox{0.55}{\input{a5_copcomm.tikz}}\\
				\scalebox{0.5}{\input{bw_unit.tikz}} \qquad \scalebox{0.55}{\input{a7_deladd.tikz}} \qquad \scalebox{0.55}{\input{a8_copadd.tikz}}\\
				\scalebox{0.5}{\input{a9_zercop.tikz}} \quad \scalebox{0.5}{\input{a10_delzer.tikz}} \quad \scalebox{0.55}{\input{a11_scalid.tikz}} \quad \scalebox{0.55}{\input{a12_scalscal.tikz}} \quad \scalebox{0.55}{\input{a13_addscal.tikz}}\\
				\scalebox{0.55}{\input{a14_zerscal.tikz}} \quad \scalebox{0.55}{\input{a15_scalcop.tikz}} \quad \scalebox{0.55}{\input{a16_scaldel.tikz}} \quad \scalebox{0.5}{\input{a17_zero.tikz}} \quad \scalebox{0.45}{\input{a18_addingscalars.tikz}}
			\end{gather*}
																																																																																																																																																																																																																																																																																																																																													We write $\sync{\HA{R}}$ for the SMC freely generated by $\HA{R}$, defined according to \cref{def:freesmc}.
	\end{definition}

						\begin{proposition}\label{prop:mataxiom}\!\!\!\cite[Proposition~3.9]{zanasi:tel-01218015}
		The following assignment of a matrix to each generator of $\HA{R}$, $F_R(\raisebox{0.3ex}{\scalebox{0.4}{\input{delete.tikz}}}) = []$, $F_R(\raisebox{0.3ex}{\scalebox{0.4}{\input{copy.tikz}  }}) = \left[\begin{smallmatrix}
			1_R\\1_R
		\end{smallmatrix}\right]$, $F_R(\raisebox{0.3ex}{\scalebox{0.4}{\input{add.tikz}   }}) = [\begin{smallmatrix}
			1_R&1_R
		\end{smallmatrix}]$, $F_R(\raisebox{0.3ex}{\scalebox{0.4}{\input{zer.tikz}   }}) = []$, $F_R(\raisebox{0.3ex}{\scalebox{0.5}{\input{scalar.tikz}}}) = [k]$, 
																														yields an identity-on-objects freely generated symmetric monoidal functor $F_R\colon \sync{\HA{R}} \rightarrow \Mat_R$, which is furthermore an isomorphism of SMCs.
	\end{proposition}
	Showing that $F_R$ is faithful relies on a decomposition result for morphisms of $\sync{\HA{R}}$ \cite[Lemma~3.10]{zanasi:tel-01218015}. We prove a variant of this result that will be convenient to use later.
	
	\begin{lemma}\label{lem:matrix-representation}
		Given $n,m \in \bN$, there are two morphisms $b^n_m\colon n \rightarrow nm$ and $w^n_m\colon nm \rightarrow m$ in $\sync{\HA{R}}$ such that for any morphism $f\colon n \rightarrow m$ in $\sync{\HA{R}}$, there are scalars $\{f_{ij} \in R\}_{i \in \fset{m}, j \in \fset{n}}$ such that  $f =  b^n_m ; ( \bigotimes_{i \in \fset{m},j \in \fset{n}} f_{ij} ) ; w^n_m $ and the $(i,j)$-entry of the matrix $F(f)$ is $f_{ij}$.
	\end{lemma}

This decomposition is related to the `matrix canonical form' of \cite[Lemma~3.10]{zanasi:tel-01218015} because it can be shown that $b^n_m$ is represented by a string diagram containing only the generators $\raisebox{0.3ex}{\scalebox{0.35}{\input{delete.tikz}}}$ and $\raisebox{0.3ex}{\scalebox{0.35}{\input{copy.tikz}}}$, while $w^n_m$ only contains $\raisebox{0.3ex}{\scalebox{0.35}{\input{zer.tikz}}}$, $\raisebox{0.3ex}{\scalebox{0.35}{\input{add.tikz}}}$, and $\raisebox{0.3ex}{\scalebox{0.35}{\input{sym.tikz}}}$. For example, the morphism $f$ satisfying $F_R(f) = \left[\begin{smallmatrix}
	a & b\\c & d
\end{smallmatrix}\right]$ decomposes as $\scalebox{0.6}{\input{exmpmatdecompsmallbis.tikz}}$.

\subsection{Axiomatising the Preorder Relation for Matrices}
\label{sec:preoder-matrix}

When $R$ is an ordered semiring, there is a simple preorder on matrices of the same size defined by entrywise comparisons: $A\leq B$ if and only if, $A_{ij}\leq B_{ij}$ for all $i$, $j$.

In this section we consider the task of axiomatising this preorder. There are two main steps. First, formulate the preorder as an enrichment on $\Mat_R$. Second, identify a quantitative extension of the theory $\HA{R}$ and show it axiomatises the enriched version of $\Mat_R$. A key property for the enrichment is not just the existence of an order on the semiring elements, but also compatibility of this order with matrix multiplication.
\begin{assumption}\label{ass:orderring}
	Throughout this subsection we assume $R$ to be an ordered semiring such that, for each $a,a',b,b' \in R$,  $a \leq a'$ and $b \leq b'$ implies $a+b \leq a'+ b'$ and $ab \leq a'b'$.
\end{assumption}

For example, the Boolean semiring $\{0m1\}$ and the semiring $[0,\infty)$ of nonnegative reals satisfies this assumption, whereas $\bR$ does not.
As seen in \Cref{ex:preorder-as-hemimet}, $\twoQ$-hemimetrics are preorders and their nonexpansive maps are order-preserving functions. Thus, we seek to enrich $\Mat_{R}$ in $\qenrcats{\twoQ}$.
Since $\vten = \vmeet$ in $\twoQ$, the two monoidal products provided in \cref{ex:quantalic_sum_metric,ex:quantalic_max_metric} coincide, and are defined as: $(x,y) \leq (x',y')$ if and only if $x \leq x'$ and $y \leq y'$. Thanks to \cref{ass:orderring}, we can show that matrix multiplication and direct sum preserve this order, and we obtain the following.
\begin{theorem}\label{thm:enrichmentMatR}
	The category $\Mat_{R}$ equipped with the $\twoQ$-hemimetrics corresponding to entrywise comparison on its hom-sets, denoted with $\EnrMat_{R}$, is a $\qenrcats{\twoQ}$-enriched SMC. \end{theorem}

Next, we introduce a $\twoQ$-quantitative monoidal theory and show it axiomatises $\EnrMat_R$. 

\begin{definition}
	The $\twoQ$-quantitative monoidal theory $\preordqmat{R}$ is defined as the monoidal theory $\HA{R}$ extended with the following family of quantitative equations: $\forall k_1 \leq k_2\in R$,\begin{equation}\label{eqn:mat:orderscalars}
		\scalebox{0.85}{\input{orderscalars.tikz}}.
	\end{equation}
\end{definition}

As a reason for~\cref{eqn:mat:orderscalars}, recall that, in a $\twoQ$-hemimetric space, two objects $x,y$ having distance $\top$ corresponds to $x\leq y$ when seeing the space as a preorder.We may now form the syntactic category $\sync{\preordqmat{R}}$ on $\preordqmat{R}$ (\cref{def:synt-enr-cat}) using the closure $\Hclosesumsum{E_q}$, where $E_q$ contains the quantitative equations in \cref{eqn:mat:orderscalars}.  We may equivalently use any closure in \cref{tab:closures-quant-eq} since $\vten = \vmeet$ in $\twoQ$. Both $\sync{\preordqmat{R}}$ and $\EnrMat_{R}$ are $\qenrcats{\twoQ}$-enriched monoidal, and we now prove they are isomorphic as enriched SMCs. 

\begin{theorem}
	\label{thm:preord-mat-quant-mon-theory}
	The isomorphism $F_{R}\colon \sync{\HA{R}} \rightarrow \Mat_{R}$ from \cref{prop:mataxiom} induces an isomorphism of $\qenrcats{\twoQ}$-enriched SMCs $F_{R}\colon \sync{\preordqmat{R}}\to\EnrMat_{R}$.
\end{theorem}
\begin{proof}	It suffices to prove $F_{R}$ is locally an isometry, which for $\twoQ$-hemimetrics means that for any morphisms $f$ and $g$ in $\sync{\HA{R}}$, $f\leq g$ if and only if $F_{R}(f)\leq F_{R}(g)$. 
				
	The forward implication says that $F_R$ is enriched, and by \cref{thm:model-to-enriched-model}, we can prove this by checking the quantitative equations of $\preordqmat{R}$ are true in $F_R$. The latter are of the form \cref{eqn:mat:orderscalars} when $k_1 \leq k_2$, and those are clearly true in $F_R$ because $F_{R}(\raisebox{0.3ex}{\scalebox{0.5}{\input{scalar.tikz}}}) = [k]$. 

						It remains to prove the converse implication. 	By \Cref{lem:matrix-representation}, we can decompose $f$ and $g$ as $b^n_m ; ( \bigotimes_{i \in \fset{m},j \in \fset{n}} f_{ij} ) ; w^n_m$ and $b^n_m ; ( \bigotimes_{i \in \fset{m},j \in \fset{n}} g_{ij} ) ; w^n_m$.
				Now, $F_{R} (f)\leq F_{R} (g)$ means that each entry of $F_{R} (f)$ is less or equal than each corresponding entry of $F_{R} (g)$, so for any $i$ and $j$, $F_R(f_{ij})\leq F_R(g_{ij})$, hence $\raisebox{0.3ex}{\scalebox{0.6}{\input{scalarfij.tikz}}} \leq \raisebox{0.3ex}{\scalebox{0.6}{\input{scalargij.tikz}}}$ by \cref{eqn:mat:orderscalars}. Having established this relation between scalars appearing in $f$ and $g$, and exploiting the decompositions, 	we can conclude that $f \leq g$ holds in $\sync{\preordqmat{R}}$ by repeated application of the inference rules {\RuleTensmeet} and {\RuleSeqmeet}. \qedhere
															\end{proof}

\section{Case Study II: Total Variation Distance}\label{sec:axiomatisetotalvariation}

The total variation distance is one of the most widely studied metrics on probability distributions. It appears ubiquitously in various fields of applied mathematics, a prominent example being optimal transport theory~\cite{Villani2009}. In \cite{Mardare2016}, the authors axiomatise the total variation distance on probability distributions as a quantitative (cartesian) algebraic theory.

In this section, we achieve a similar characterisation result, but in the language of quantitative monoidal theories. Rather than just discrete probability distributions, we focus more generally on stochastic matrices. These form a category $\FStoch$, in which distributions are the $1 \to n$ morphisms.

Our first step is to recall $\FStoch$ and the (non-quantitative) monoidal theory axiomatising it (\cref{sec:fritz}). Second, we introduce the total variation distance and show that $\FStoch$ is enriched over metric spaces, so that we can study total variation between its morphisms meaningfully (\cref{sec:totalvarenrichment}). Thirdly, we expand the monoidal theory of \cref{sec:fritz} to a quantitative monoidal theory, and show that it axiomatises $\FStoch$ as an enriched SMC. Effectively, this means that two matrices in $\FStoch$ are at total variation distance $\varepsilon$ if and only if the corresponding string diagrams can be proven to be at distance $\varepsilon$ in the theory.

This can be understood as an axiomatisation of the `metric theory' in \cite[Example~3.2.7]{HLarsen2021}.

\subsection{Background: (non-Quantitative) Axiomatisation of \texorpdfstring{$\FStoch$}{FStoch}}\label{sec:fritz}

Here, we recall the axiomatisation result of $\FStoch$, only focusing on exact equality.

\begin{definition}	\label{def:finstoch}
	The SMC $\FStoch$ is the subcategory of $\Mat_{[0,1]}$ whose morphisms $n \to m$ are the \emph{stochastic matrices}, i.e. $m \times n$ matrices with entries in the interval $[0,1]$, such that the sum of the entries in a column always equals $1$. 										\end{definition}
In the sequel, we will often (implicitly) see columns of a stochastic matrix as probability distributions and vice-versa.
Recall that, given a set $X$, a \emph{(probability) distribution} on $X$ is a function $\dist\colon X \rightarrow [0,1]$ satisfying $\sum_{x \in X} \dist(x) = 1$. There is a monad mapping $X$ to the set $\Dset(X)$ of finitely supported probability distributions on $X$, and one may regard $\FStoch$ as a full subcategory of the Kleisli category of such monad. We refer to \cref{app:distrmonad} for details.
The following axioms were originally studied in~\cite{Stone1949PostulatesFT}, but we follow the more recent \cite{Fritz09}, casting it in the setting of monoidal categories.

\begin{definition}	\label{def:convtheory}
	The monoidal theory $\ConvAlg$ (standing for \emph{convex algebras}) has generators $\raisebox{0.3ex}{\scalebox{0.5}{\input{del.tikz}}} \colon 0 \to 1$,  $\raisebox{0.3ex}{\scalebox{0.5}{\input{cop.tikz}}} \colon 2 \to 1$, and $\raisebox{0.2ex}{\scalebox{0.6}{\input{cc.tikz}}} \colon 1 \to 2$	for each $\lambda \in [0,1]$, and equations as in \cref{fig:fritzaxiom} below. We write $\sync{\ConvAlg}$ for the SMC freely generated by $\ConvAlg$ (recall \cref{def:freesmc}).\end{definition}

\begin{figure}[!htb]
{	
	\small
	\begin{align*}
		\scalebox{0.45}{\input{assoc.tikz}} \qquad \scalebox{0.45}{\input{comm.tikz}} \qquad \scalebox{0.5}{\input{unit.tikz}}\\
		\scalebox{0.55}{\input{idemp.tikz}} \qquad \scalebox{0.45}{\input{convassoc.tikz}} \qquad \scalebox{0.5}{\input{convcomm.tikz}} \\
		\scalebox{0.55}{\input{natdel.tikz}} \qquad \scalebox{0.55}{\input{zprob.tikz}} \qquad \scalebox{0.55}{\input{cccop.tikz}}
	\end{align*}
																																																													}
												\caption{Axiomatisation of $\FStoch$. 		We write $\tilde{\lambda}$ for $\lambda\mu$ and $\tilde{\mu}$ for $\frac{\lambda -\lambda\mu}{1-\lambda\mu}$ (with $\frac{0}{0} = 1$). 	}
	\label{fig:fritzaxiom}
\end{figure}

\begin{proposition}\label{thm:fritzmain}\!\!\!\cite[Theorem~3.14]{Fritz09}
	The following assignment of a stochastic matrix to each generator of $\ConvAlg$, $F(\raisebox{0.3ex}{\scalebox{0.4}{\input{del.tikz}}}) = []$, $F(\raisebox{0.3ex}{\scalebox{0.4}{\input{cop.tikz}}}) = \left[\begin{smallmatrix}
		\scriptstyle 1 & \scriptstyle 1
	\end{smallmatrix}\right]$, $F(\raisebox{0.3ex}{\scalebox{0.4}{\input{cc.tikz}}}) = \left[ \begin{smallmatrix}\scriptstyle \lambda\\ \scriptstyle 1-\lambda\end{smallmatrix} \right]$, 
																	yields an identitity-on-objects freely generated symmetric monoidal functor $F\colon \sync{\ConvAlg} \rightarrow \FStoch$, which is furthermore an isomorphism of SMCs.
\end{proposition}
\begin{remark} A few caveats when comparing our presentation with the one of \cite{Fritz09}: the category $\FStoch$ is called $\mathtt{FinStoMap}$; the author uses different graphical conventions for reading sequential and parallel composition of diagrams (top-to-bottom instead of left-to-right); the author has the symmetric structure as explicit part of the presentation, rather than as something generated freely by the syntactic category of string diagrams. Furthermore, note that there are other ways to present $\FStoch$ axiomatically, e.g.~\cite[Example 6.2(c)]{Bonchi2018}. 
\end{remark}
Showing that $F$ is faithful relies on a decomposition for morphisms of $\sync{\ConvAlg}$ recalled below.

\begin{lemma}\!\!\!\cite[Propositions~3.12 and 3.13]{Fritz09}\label{lem:fritzdecompose}
	Given $n,m \in \bN$, there is a morphism $p_m^n\colon nm \rightarrow m$ such that $\forall f\colon n \rightarrow m \in \sync{\ConvAlg}$, there are morphisms $\{f_i \colon 1 \rightarrow m\}_{i\in \fset{n}}$ such that 	$f = (f_1 \otimes \cdots \otimes f_n);p_m^n$ and
$F(f_i)$ is the $i$th column of $F(f)$.
\end{lemma}

\subsection{Enrichment of \texorpdfstring{$\FStoch$}{FStoch} with the Total Variation Distance}\label{sec:totalvarenrichment}

In this section we define an enrichment on $\FStoch$ based on the \emph{total variation} distance $\tv\colon \Dset X \times \Dset X \rightarrow [0,1]$, which is defined by $
		\tv(\dist,\distb) \coloneqq \max_{S \subseteq X} |\scalebox{0.8}{$\sum_{x \in S}$} \dist(x) - \scalebox{0.8}{$\sum_{x \in S}$} \distb(x)|$.

In the context of this paper, $\tv$ is a metric that can be defined on the set of morphisms $\FStoch(1,m)$ for any positive $m \in \bN$. Now, in order to define an enrichment of $\FStoch$ over $\enCatsym{\zeroinfQ}$, we still need pseudometrics on the other hom-sets. The following definition is somewhat natural (see \cref{rem:relativeKleisli}): for every $n,m \in \bN$, the metric $\tvmax$ on $\FStoch(n,m)$ is 
\begin{equation}\label{eqn:tvmax}
	\tvmax(A,B) = \max_{i \in \fset{n}}\tv(A_i,B_i),
\end{equation}
where $A_i$ is the $i$th column of $A$, and $\tv(A_i,B_i)$ is the total variation distance between the corresponding distributions.

We write $\FS$ for $\FStoch$ equipped with the metric $\tvmax$ on its hom-sets. After showing that sequential and parallel composition are nonexpansive relative to the monoidal product~$\boxtimes_{\vten}$ (from \cref{ex:quantalic_sum_metric}), we conclude the following.

\begin{theorem}\label{lem:fsenriched}
$\FS$ is a $\PMet_\otimes$-enriched SMC.
\end{theorem}

\begin{remark}\label{rem:relativeKleisli}
	The definition of $\FS$ via $\tvmax$ \cref{eqn:tvmax} and the fact that it is enriched can be obtained more abstractly. First, we can show $\Dset$ lifts to an enriched relative monad on $\PMet_\otimes$ with the total variation metric by adapting \cite[Example~5.10]{Arkor2025} for the theory of LIB algebras in \cite[Section~8]{Mardare2016}. Then, we unroll the construction of the enriched relative Kleisli category in \cite[Proposition~8.21]{Arkor2024} to get the category (opposite to) $\FS$.
\end{remark}
\begin{remark}\label{rem:FSnotmaxenriched}
		The category $\FS$ is not enriched over $\PMet$ with the monoidal product $\boxtimes_{\vmeet}$ from \cref{ex:quantalic_max_metric}, because sequential composition is not nonexpansive for this monoidal product: 	with the matrices $A  = { \left[ \begin{smallmatrix}
		1&  0.5\\ 0& 0.5
   \end{smallmatrix} \right]}$, $B = A$ , $C = \left[ \begin{smallmatrix}
	   1\\0
   \end{smallmatrix} \right]$, and $C' = \left[ \begin{smallmatrix}
	   0.5\\0.5
   \end{smallmatrix} \right]$, we have $AC=\left[ \begin{smallmatrix}
	   1\\0
   \end{smallmatrix} \right]$ and $BC'=\left[ \begin{smallmatrix}
	   0.75\\0.25
   \end{smallmatrix} \right]$, thus $\tvmax(C;A, C';B) = \tfrac{3}{4} > \tfrac{1}{2} = \max\left\{ \tvmax(C,C'), \tvmax(A,B) \right\}$.																\end{remark}

\subsection{Quantitative Axiomatisation of \texorpdfstring{$\FS$}{FStoch\_tv}}\label{sec:qaxtv}

We introduce the quantitative monoidal theory that axiomatises $\FS$. 
\begin{definition}	\label{def:barycentric-theory}
	The $[0,\infty]$-quantitative monoidal theory $\BarAlg$ is defined as the tuple $(\Sigma, E, E_q)$ where $(\Sigma, E) = \ConvAlg$ (\cref{def:convtheory}) and $E_q$ contains, for each $\lambda \in [0,1]$,
	\begin{equation}\label{eqn:ca:tv}
						\scalebox{0.6}{\input{tvleft.tikz}}\ =_{\lambda} \scalebox{0.6}{\input{tvright.tikz}}\tag{TV}
	\end{equation}
	We write $\sync{\BarAlg}$ for the $\PMet_\otimes$-enriched SMC generated by $\BarAlg$ using the inference rules \RuleSeqsum, \RuleTensmeet, and \RuleSym, constructed according to \cref{def:synt-enr-cat}.
\end{definition}
Note that \cref{eqn:ca:tv} is adapted from the quantitative equations \textbf{LI} used in \cite[Definition~8.1]{Mardare2016}. Both sides of \cref{eqn:ca:tv} will give weight $1-\lambda$ to the second output, thus the parts where they differ will have weight at most $\lambda$. Hence, the distance between the results is at most $\lambda$. 
\begin{remark}
	The choice of inference rules used to generate $\sync{\BarAlg}$ is motivated by our goal to construct an enriched isomorphism $\sync{\BarAlg} \rightarrow \FS$. Indeed, by \cref{rem:FSnotmaxenriched} the rule \RuleSeqmeet\ would not be valid in $\FS$, and if we took the less strict  \RuleTenssum\ over \RuleTensmeet, nothing would guarantee validity of the latter in $\sync{\qU}$. But \RuleTensmeet\ is valid in $\FS$.\end{remark}

Our axiomatisation result now amounts to showing that the functor $F$ introduced in \cref{thm:fritzmain} is an isomorphism of enriched categories between $\sync{\BarAlg}$ and $\FS$. The following lemma, along with \cref{thm:model-to-enriched-model}, implies $F$ is enriched.

\begin{lemma}\label{lemma:tvF}
	The quantitative equations \cref{eqn:ca:tv} are true in~$F$.
\end{lemma}

To conclude that $F$ is actually an enriched isomorphism, it is enough to show it is locally an isometry. We first focus on the case of morphisms $1 \to m$, that is, probability distributions, and we recall a lemma used in the axiomatisation of Mardare et al.\begin{lemma}	\label{lem:splitting}\!\!\!\cite[Lemma~10.12]{BacciinBook}
	For any two distributions $\dist, \distb \in \Dset X$ with $\lambda = \tv(\dist,\distb)$, there exist three distributions $\dist',\distb', \distc \in \Dset X$ such that $\dist = \dist' +_{\lambda} \distc$ and $\distb = \distb' +_{\lambda} \distc.$
\end{lemma}

We also introduce thick wires that represent the tensor of multiple wires. For example, for any morphism $f\colon 1 \rightarrow m$ in $\sync{\BarAlg}$, we can choose a representative diagram that we draw as $\scalebox{0.7}{\input{distnode.tikz}}$, 
where $\dist$ is the distribution corresponding to $F(f)$ (we often omit the number on top of the thick wire). Moreover, there are also thick versions of $\raisebox{0.3ex}{\scalebox{0.55}{\input{del.tikz}}}$ and $\raisebox{0.3ex}{\scalebox{0.55}{\input{cop.tikz}}}$ drawn as $\raisebox{0.2ex}{\scalebox{0.55}{\input{thickdel.tikz}}}$ and $\raisebox{0.2ex}{\scalebox{0.55}{\input{thickcop.tikz}}}$ respectively, which obey the same equations as their thin counterparts. In particular, we can show the following equation is in $\sync{\BarAlg}$ by induction.
\begin{equation}\label{eqn:deldistrib}
	\scalebox{0.7}{\input{deldist.tikz}}
\end{equation}
Now, a convex combination $\dist +_{\lambda} \distb$ can be represented diagrammatically with $\scalebox{0.5}{\input{convexcomb.tikz}}$, which facilitates a diagrammatic proof of the following result.

\begin{lemma}\label{lem:isoondist}
		The function $f \mapsto F(f)$ is a bijective isometry $\sync{\BarAlg}(1,m) \rightarrow \FS(1,m)$. 
\end{lemma}
\begin{proof}[Proof sketch]
													Given two morphisms $f,g\colon 1 \rightarrow m$, let their corresponding distributions be $\dist,\distb \in \Dset \fset{m}$, and let $\dist'$, $\distb'$, and $\tau$ be given by \cref{lem:splitting}. We provide a derivation in $\BarAlg$ of $f =_{\lambda} g$, with $\lambda = \tv(\dist,\distb) = \tvmax(F(f),F(g))$.
	\begin{align*}
		\raisebox{0.2ex}{\scalebox{0.65}{\input{dist.tikz}}} &=_0 \scalebox{0.65}{\input{distsplitthreesmall.tikz}}
		\stackrel{\text{by \cref{eqn:ca:tv}}}{=_{\lambda}} \scalebox{0.65}{\input{distbsplitthreesmall.tikz}}
		=_0 \raisebox{0.2ex}{\scalebox{0.65}{\input{distb.tikz}}}
	\end{align*}
	Both $=_0$ steps follow from hypotheses $\dist = \dist' +_{\lambda} \distc$, $\distb = \distc +_{1-\lambda} \distb'$, and \cref{eqn:deldistrib}. This shows $d^{\BarAlg}(f,g) \leq \tv(\dist,\distb)$. The converse inequality holds because $F$ is an enriched functor (\cref{lemma:tvF}), hence $F$ is an isometry on the hom-sets with domain $1$. It is bijective by \cref{thm:fritzmain}.	\end{proof}

\cref{lem:fritzdecompose} allows us to extend our argument to arbitrary morphisms of $\FS$.
\begin{theorem}
	\label{thm:baralg-enriched-iso}
	The functor $F\colon \sync{\BarAlg} \rightarrow \FS$ is an isomorphism of enriched categories.
\end{theorem}

As we mentioned, our axiomatisation of total variation distance between stochastic matrices was inspired by Mardare et al.'s for distributions. We discuss the link between our work and quantitative algebraic theories in the following section.

\section{Comparison with Related Work: Cartesian vs Monoidal}\label{sec:related}

The work \cite{Bonchi2018} relates monoidal theories and (cartesian) algebraic theories, showing that terms of an algebraic theory $\qU$ correspond to string diagrams in a monoidal theory $\qU'$, where $\qU'$ only adds a natural commutative monoid structure to $\qU$. This follows by an isomorphism between the Lawvere category generated by $\qU$ and the SMC freely generated by $\qU'$ \cite[Theorem~6.1]{Bonchi2018}. In this section, we establish an analogous link between the unconditional quantitative algebraic theories of \cite{Mardare2016} and our quantitative monoidal theories (\cref{def:quantitative-monoidal-theory}), via the discrete enriched Lawvere theories of \cite{Power2005}. 

Recasting \cite{Mardare2016} (and \cite{Sarkis2024}, which generalises \cite{Mardare2016} to quantales), an \emph{unconditional $\qV$-quantitative algebraic theory} $\qU$ is a triple $(\Sigma,E,E_q)$, where $\Sigma$ is a signature of operations with coarity $1$, $E$ is a set of equations between cartesian terms (the standard terms from universal algebra), and $E_q$ is a set of $\qV$-quantitative equations between cartesian terms. Elements of $E_q$ correspond to quantitative equations of \cite{Mardare2016} with no premises ($\emptyset \vdash s =_{\varepsilon} t$), and they are called unconditional in \textit{loc.~cit.} Any such theory generates a discrete enriched Lawvere theory \cite[Definition~4]{Power2005} as follows.

\begin{definition}	\label{def:generated-dlawvth}
	The discrete $\qVCatmax$-Lawvere theory generated by $\qU$ is the $\qVCat$-category $\LT{\qU}$, where objects are natural numbers,  and morphisms $n \rightarrow m$ are $n$-tuples of cartesian terms with at most $m$ variables, e.g. $\langle f(x_1,x_3), x_3\rangle\colon 2 \rightarrow 3$, considered modulo the equations between terms derived in quantitative equational logic from the axioms in $E$ and $E_q$ (see \cref{fig:qeqlog} in the appendix inspired by \cite{Mardare2016,Sarkis2024}). Composition of morphisms is by substitution. 	The distance between morphisms is computed between terms as the join of derivable distances, and between tuples as the coordinatewise meet. Namely, if $\qU \vdash f =_{\varepsilon} g$ denotes that $f =_{\varepsilon} g$ is derivable from the axioms in $E$ and $E_q$, then $d_{\LT{\qU}}(\langle f_i \rangle, \langle g_i \rangle) = \vmeet_{i} \vJoin \left\{ \varepsilon \mid \qU \vdash f_i =_{\varepsilon} g_i \right\}$.
						\end{definition}
\begin{remark}\label{rem:enrichedLTiskleisli}
	Following \cite{Rosicky2024}, an equivalent description of $\LT{\qU}$ is as the restriction of the enriched Kleisli category for the free $\qU$-algebra monad to the discrete spaces on finite sets.\end{remark}
We will show $\LT{\qU}$ can be freely generated from the theory that combines $\qU$ with a natural cocommutative comonoid structure.
\begin{definition}	\label{def:associated-cart-quant-theory}
	The $\qV$-quantitative monoidal theory $\qU' = (\Sigma',E',E'_q)$ is defined by $\Sigma' \coloneqq \Sigma\sqcup\lbrace\raisebox{0.3ex}{\scalebox{0.4}{\input{delete.tikz}}},\,\raisebox{0.3ex}{\scalebox{0.4}{\input{copy.tikz}}}\rbrace$, $E'\coloneqq E\cup E^c\cup E^d$, and $E'_q \coloneqq E_q$, 
		where $E^c$ contains the equations making $\lbrace\raisebox{0.3ex}{\scalebox{0.4}{\input{delete.tikz}}},\,\raisebox{0.3ex}{\scalebox{0.4}{\input{copy.tikz}}}\rbrace$ into a cocommutative comonoid and $E^d$ the naturality equations $\scalebox{0.5}{\input{copynat.tikz}}$ and $\scalebox{0.5}{\input{deletenat.tikz}}$ for each in $f \in \Sigma$.				
				
		\end{definition}
Our constructions clearly treat the quantitative axioms and the distances separate from the rest. In other words, the underlying categories of $\LT{\qU}$ and $\sync{\qU'}$ are, respectively, the Lawvere category generated by $(\Sigma,E)$ and the prop generated by $(\Sigma',E')$. Hence, it readily follows from \cite[Theorem 6.1]{Bonchi2018} that the underlying categories are isomorphic. It remains to show this isomorphism is an isometry. At this point, it is important to note the choice of inference rules used to generate $\sync{\qU'}$: we take \RuleSeqsum\ and \RuleTensmeet. We also need to assume that $\qV$ is IJD at a technical point in the proof.

\begin{theorem}\label{thm:from-quant.mon.th-to-quant.mon.th}
	Let $\qU$ be an unconditional $\qV$-quantitative algebraic theory and $\qU'$ a $\qV$-quantitative monoidal theory constructed as in \cref{def:associated-cart-quant-theory}. There is an isomorphism of $\qVCat$-enriched categories between $\LT{\qU}$ and $\sync{\qU'}$.  
\end{theorem}
\begin{proof}[Proof sketch]
	We see the isomorphism between the underlying categories as a model of $(\Sigma',E')$ valued in $\LT{\qU}$. We apply \cref{thm:model-to-enriched-model} to show it is an enriched model. It remains to prove that the distance between cartesian terms in $\LT{\qU}$ is smaller than the distance between the corresponding diagrams in $\sync{\qU'}$. We do this by simulating all the rules in quantitative equational logic with the rules used to build the closure $\Hclosesummax{{E'_q}}$.
\end{proof}

In words, \cref{thm:from-quant.mon.th-to-quant.mon.th} shows that we can always extract the linear part of an unconditional quantitative algebraic theory just as we can extract the linear part of an algebraic theory.\begin{example}
	The theory of LIB algebras in \cite[Definition 8.1]{Mardare2016} is an unconditional quantitative algebraic theory, call it $\qU$. Unrolling \cref{rem:relativeKleisli,rem:enrichedLTiskleisli}, we find that the enriched Lawvere theory generated by $\qU$ is the opposite of $\FS$. Thus \cref{thm:from-quant.mon.th-to-quant.mon.th} provides a quantitative monoidal theory $\qU'$ and an enriched isomorphism $\sync{\qU'} \cong \FS^{\mathrm{op}}$.

	At first sight, this seems like another axiomatisation of the total variation distance complementary to \cref{thm:baralg-enriched-iso}. Further investigation shows that $\qU'$ and $\BarAlg$ are morally the same. This situation exactly mirrors the differences between the axiomatisations of $\FStoch$ in \cite[Example 6.2(c)]{Bonchi2018} and \cite{Fritz09}. Namely, the latter avoids redundant equations.
\end{example}

\section{Conclusions}\label{sec:conclusion}
Our work provides mathematical foundations to enhance monoidal algebra with quantitative equations. We are motivated by the increasing relevance of string diagrammatic calculi in areas such as quantum theory, machine learning, probabilistic programming, and circuit theory, in which quantitative reasoning plays a fundamental role. Our basic examples in \cref{sec:mat,sec:axiomatisetotalvariation} are intended merely as a proof-of-concept for our framework. More sophisticated examples, building on diagrammatic calculi for quantum \cite{Coecke2008,Coecke2017,lambeq}, probability theory \cite{Fritz2020,JacobsKZ21,Perrone2024,Piedeleu2025b}, and machine learning \cite{Cruttwell2022,Wilson2022}, deserve a separate development, which we will explore in future work.

A notable aspect of this work is the flexibility we provide to generate a syntactic category from a monoidal theory, where the different inference rules depend on which quantale operations we pick. This is due to how monoidal terms are formed differently from cartesian terms, and is motivated by the examples we developed. For instance, in $\BarAlg$, sequential and parallel composition are nonexpansive with respect to the sum and max metric respectively, so the rules \RuleSeqsum\ and \RuleTensmeet\ are used to generate $\sync{\BarAlg}$. To encompass more examples, one could devise other rules corresponding to enrichment over other monoidal products. One may also consider a logic whose judgments are inference rules (or implications) rather than quantitative equations, so that \RuleSeqsum\ and \RuleTensmeet\ become part of the theory $\BarAlg$.

Other questions concern the relation between cartesian and monoidal theories. Can \cref{thm:from-quant.mon.th-to-quant.mon.th} be obtained more abstractly via distributive laws like the non-quantitative result in \cite{Bonchi2018}? Also, in \cite{Mardare2016} and subsequent works,  distances between complex terms depend on the distances between variables used in those terms. For example, the construction of $\FS$ in \cref{rem:relativeKleisli} relies on the Kantorovich lifting of the distribution monad relative to the inclusion $\mathbf{FinSet} \hookrightarrow \mathbf{Met}$. If we use $\mathbf{FinMet} \hookrightarrow \mathbf{Met}$ instead, the distance between distributions depends on the finite metric space considered. Axiomatising this category would require the diagrammatic syntax to incorporate some quantitative information on the inputs and outputs: it is an open question how to represent it in monoidal algebra.
\bibliography{refs}
\appendix
\section*{Appendix}\label{appendix}
\renewcommand{\thesubsection}{\Alph{subsection}}
\counterwithin{theorem}{subsection}
\counterwithin{definition}{subsection}
\counterwithin{lemma}{subsection}
\counterwithin{remark}{subsection}
\counterwithin{proposition}{subsection}
\counterwithin{corollary}{subsection}
\counterwithin{table}{subsection}
\etocsettocstyle{}{}
\etocsettocdepth{2}
\localtableofcontents

\subsection{Axioms of Symmetric Strict Monoidal Categories (SMCs)}
\label[app]{App:SMC}

		\begin{equation}\label{eqn:smcax}
			\begin{gathered}
				\!\!\!\!\!\scalebox{.6}{\input{smc/sequential-associativity.tikz} \!=\! \input{smc/sequential-associativity-1.tikz}} \quad \scalebox{.6}{\input{smc/parallel-associativity.tikz} \!\!\!=\!\!\! \input{smc/parallel-associativity-1.tikz}} \\ \scalebox{.6}{\input{smc/interchange-law.tikz} \!\!=\!\!\input{smc/interchange-law-1.tikz} }
				\quad
				\scalebox{.6}{ \input{smc/parallel-unit-above.tikz} = \diagbox{c}{}{} =  \input{smc/parallel-unit-below.tikz}}
				\\
				\scalebox{.6}{\input{smc/unit-right.tikz} = \diagbox{c}{}{} = \input{smc/unit-left.tikz}}
				\\
				\scalebox{.6}{\input{smc/sym-natural.tikz}= \input{smc/sym-natural-1.tikz}}
				\qquad
				\scalebox{.6}{\input{smc/sym-iso.tikz} = \input{id2.tikz}}
			\end{gathered}
		\end{equation}

\subsection{Additional Background for \Cref{sec:prelim}}
\label[app]{app:proofs-prelim}
	
\begin{proposition}\label{rem:tensor_below_meet}
		In an integral quantale $\qV$, for all $a,b \in \qV$, $a \oplus b \sqsubseteq a \sqcap b$.
\end{proposition}
\begin{proof}
			We first prove that in any quantale $\vten$ is monotone in both arguments:
	\begin{align}
		x \vle y &\iff x \vjoin y = y \nonumber\\
		&\implies a \vten (x \vjoin y) = a \vten y \nonumber\\
		&\iff (a \vten x) \vjoin (a \vten y) = (a \vten y) \tag{Join distributivity}\\
		&\iff (a \vten x) \vle (a \vten y) \label{rem:monotone}
	\end{align}
	Second, we prove that, in any integral quantale, $a \vten b \vle a$:
			\begin{align*}
		b \vle \top &\implies a \vten b \vle a \vten \top \tag{by \cref{rem:monotone}}\\
		&\iff a \vten b \vle a \vten k \tag{$k = \top$}\\
		&\iff a \vten b \vle a
	\end{align*}
	Similarly, we have $a \vten b \vle b$, and thus $a \vten b \vle a \vmeet b$
\end{proof}

\subparagraph*{Enriched Categories}
	We only gloss over the technical details of enriched categories in the body of the paper, so we take some time here to formalise our discussion. Let us first recall the main definitions concerning enriched categories. The standard reference is \cite{KellyBook}. We fix an arbitrary monoidal category $\VV$. We denote its monoidal product $\boxtimes \colon \VV \times \VV \to \VV$, its monoidal unit $I$, its associators $\alpha_{a,b,c} \colon (a \boxtimes b) \boxtimes c \to a \boxtimes (b \boxtimes c)$, its left unitors $l_a \colon I  \boxtimes a \to a$, and its right unitors $r_a \colon a \boxtimes I  \to a$.
	
	\begin{definition}\label{def:v_enriched_cat}
		A small \emph{$\VV$-enriched category} $\C$ (or simply \emph{$\VV$-category}) consists of: 		\begin{itemize}
			\item A set $\Ob(\C)$ called the set of objects;
			\item For each pair $(a,b)$ of objects in $\C$, an object $\C(a,b) \in \Ob(\VV)$ called the hom-object;
			\item For each triple $(a,b,c)$ of objects in $\C$, a morphism $M_{a,b,c} \colon \C(b,c) \boxtimes \C(a,b) \to \C(a,c)$ in $\VV$ called composition;
			\item For each object $a$ in $\C$, a morphism $j_a \colon I  \to \C(a,a)$ in $\VV$ called the identity; 
		\end{itemize}
		They satisfy the following coherence conditions for all $a,b,c,d \in \Ob(\C)$:
		\begin{align}
			M_{a,c,d} \circ \left(\id_{\C(c,d)}  \boxtimes M_{a,b,c}\right) \circ \alpha _{\C(c,d),\C(b,c),\C(a,b)} &= M_{a,b,d} \circ \left(M_{b,c,d}  \boxtimes \id_{\C(a,b)}\right) \label{vcat:assoc} \\
			M_{a,b,b} \circ \left(j_b \boxtimes\id_{\C(a,b)}\right) &= l_{\C(a,b)}\label{vcat:lunitor}\\
			M_{a,a,b} \circ \left(\id_{\C(a,b)} \boxtimes j_a \right) &= r_{\C(a,b)}.\label{vcat:runitor}
		\end{align}
				
			\end{definition} 	
	\begin{definition}\label{def:vfunc}
		Let $\C$, $\D$ be small $\VV$-enriched categories. A \emph{$\VV$-enriched functor} (also called \emph{$\VV$-functor}) $F \colon \C \to \D$ consists of a function $F_0\colon \Ob(C) \to \Ob(D)$ and an $\left(\Ob(\C) \times \Ob(\C)\right)$-indexed collection of morphisms of $\VV$, written $F_{a,b} \colon \C(a,b) \to \D(F_0(a), F_0(b))$, such that for all $a,b,c \in \Ob(\C)$:
		\begin{align}
			F_{a,c} \circ M_{a,b,c} &= M_{F_0(a),  F_0(b), F_0(c)} \circ \left(F_{b,c} \boxtimes F_{a,c}\right)\label{vfunc:comp}\\
			j_{F_0(a)} &= F_{a,a} \circ j_a.\label{vfunc:id}
		\end{align}
	\end{definition}
	\begin{example}
		For $\VV = (\Set, \times, 1)$, the definitions above describe the usual small categories and functors between them. In many applications, $\VV$ is a category of sets equipped with structure (e.g.~groups, metric spaces, etc.). In such cases, $\VV$-categories are small categories whose hom-sets have extra structure preserved by composition, and $\VV$-functors are functors that locally (i.e.~on each hom-set) preserve this structure.
	\end{example}
		\begin{remark}\label{def:vcat_tensor}
	Monoidal categories are defined in ordinary category theory using the data of a bifunctor of type $\C \times \C \to \C$ called the monoidal product that obeys certain coherence conditions. Lifting this concept to enriched categories requires the construction of an enriched category $\C \times \C$ from an enriched category $\C$. When $\VV$ is a symmetric monoidal category, we can use the tensor product of $\VV$-categories defined in \cite[Section~1.4]{KellyBook} recalled below.
	\end{remark}
	\begin{definition}\label{def:app:vcat_tensor}
		Let $\VV$ be a symmetric monoidal category and let $\C$ be a $\VV$-category. The $\VV$-category $\C \times \C$ is defined as follows:
		\begin{itemize}
			\item The objects of $\C \times \C$ are pairs of objects of $\C$, i.e.~$\Ob(\C\times \C) = \Ob(\C)\times \Ob(\C)$.
			\item For each pair $\left((a,a'),(b,b')\right) \in \Ob(\C \times \C) \times \Ob\left(\C \times \C\right)$, we define $(\C\times\C)((a,a'),(b,b')) \coloneqq \C(a,b) \boxtimes \C(a',b')$.
			\item Let $(a,a'),(b,b'),(c,c') \in \Ob(C \times C)$. The composition morphism	$M_{(a,a'),(b,b'),(c,c')}$ 			is given by the following composite, 			where $\cong$ denotes an isomorphism constructed using the symmetric monoidal structure of $\VV$.
						\[\begin{tikzcd}[ampersand replacement=\&]
				\begin{array}{c} (\C\times\C)((b,b'),(c,c')) \otimes_{\VV} (\C\times\C)((a,a'),(b,b'))\\=\C(b,c)\otimes_{\VV} \C(b',c') \otimes_{\VV} \C(a,b) \otimes_{\VV} \C(a',b') \end{array} \\
				{\C(b,c)\otimes_{\VV} \C(a,b) \otimes_{\VV} \C(b',c') \otimes_{\VV} \C(a',b')} \\
				\begin{array}{c} \C(a,c) \otimes_{\VV} \C(a',c')\\= (\C\times\C)((a,a'),(c,c')) \end{array}
				\arrow["\cong", from=1-1, to=2-1]
				\arrow["{M_{a,b,c} \otimes_{\VV} M_{a',b',c'}}", from=2-1, to=3-1]
			\end{tikzcd}\]
			\item For each $(a,a') \in \Ob(\C \times \C)$ we define $j_{(a,a')}$ as the following composite.
															\[\begin{tikzcd}
	{I_{\VV}} & {I_{\VV}\otimes_{\VV} I_{\VV}} & \begin{array}{c} \C(a,a) \otimes_{\VV} \C(a',a') =(\C\times\C)((a,a'),(a,a')) \end{array}
	\arrow["\cong", from=1-1, to=1-2]
	\arrow["{j_a \otimes_{\VV} j_{a'}}", from=1-2, to=1-3]
	\end{tikzcd}\]
		\end{itemize}

	\end{definition}

	\subparagraph*{\texorpdfstring{$\qVCat$}{VHMet} Enrichment} We carefully develop our discussion on the characterisation of $\qVCat$-enriched categories and functors.
	\begin{lemma}\label{lem:nexpcomptoenriched}
		Let $\C$ be a small category such that for each $a,b \in \Ob(\C)$ there is a $\qV$-hemimetric space $(\C(a,b),d_{a,b})$. If for all $f,f' \in \C(a,b)$ and $g,g' \in \C(b,c)$,
				\begin{equation}\label{eqn:nexpcomposition}
			(d_{b,c} \boxtimes d_{a,b})((g,f),(g',f')) \vle d_{a,c}(g \circ f, g' \circ f'),
		\end{equation}
		then there is a $\qVCat$-category with underlying category $\C$. The same result holds for $\qVCatsym$.			\end{lemma}		
\begin{proof}
		
		Let $M_{a,b,c} \colon \C(b,c) \boxtimes \C(a,c) \to \C(a,c)$ be the composition morphism defined by $M_{a,b,c}(g,f) \coloneqq g \circ f$ for any $f \in \C(a,b)$ and $g \in \C(b,c)$. Note that \cref{eqn:nexpcomposition} ensures that $M_{a,b,c}$ is a nonexpansive map, i.e.~a morphism in $\qVCat$. For any $a \in \Ob(\C)$, let $j_a \colon 1_{\boxtimes} \to \C(a,a)$ be defined by $j_a(\bullet) \coloneqq \id_a$. We have $\top(\bullet,\bullet)=\top=k\sqsubseteq d_{a,a}(\id_a, \id_a)$ because $d_{a,a}$ satisfies reflexivity (and $\qV$ is integral), thus $j_a$ is nonexpansive as well.	
	It remains to check the coherence conditions of \Cref{def:v_enriched_cat}.
	\begin{itemize}
		\item For \cref{vcat:assoc}, we can explicitly compute how both sides act on $\C(c,d) \boxtimes \C(b,c) \boxtimes \C(a,b)$. They send $(h,g,f)$ to $(h \circ g) \circ f$ and $h \circ (g \circ f)$ which are equal because $\C$ is a category so its composition is associative.		\item For \cref{vcat:lunitor}, both sides act on $1_{\boxtimes} \boxtimes \C(a,b)$ by sending $(\bullet,f)$ to $\id_b \circ f$ and $f$ which are equal because $\C$ is a category so the identity morphism is neutral for composition.		\item The proof of \cref{vcat:runitor} is symmetric to the previous point.
	\end{itemize}
	By construction, $\C$ is the underlying category of the resulting $\qVCat$-category.
									\end{proof}

\begin{lemma}\label{lem:characVfct}
		Let $\C$ and $\D$ be two $\qVCat$-categories and $F\colon \C \rightarrow \D$ be a functor between their underlying categories. If $F$ is locally nonexpansive, namely, the assignment $f \mapsto Ff$ is a nonexpansive map $\C(a,b) \rightarrow \D(Fa,Fb)$ for all $a,b \in \Ob(\C)$, then $F$ is the underlying functor of a $\qVCat$-functor.	\end{lemma}
\begin{proof}
	Define $F_0 \colon \Ob(\C) \to \Ob(\D)$ by $F_0(a) \coloneqq F a$ for all $a \in \Ob(\C)$. Similarly, for any $a,b \in \Ob(\C)$ and a morphism $f \in \C(a,b)$, we define $F_{a,b}(f) \coloneqq Ff$. Since $F$ is locally nonexpansive, $F_{a,b}$ is indeed a morphism $\C(a,b) \rightarrow \D(F_0a,F_0b)$ in $\qVCat$. The coherence conditions \cref{vfunc:comp} and \cref{vfunc:id} are satisfied because $F$ is a functor between the underlying categories of $\C$ and $\D$ (explicit checks can be made as in the proof of \cref{lem:characVfct}). By construction, $F$ is the underlying functor of the resulting $\qVCat$-functor.
			\end{proof}

\begin{corollary}\label{cor:characenrichediso}
		Let $\C$ and $\D$ be two $\qVCat$-categories and $F\colon \C \rightarrow \D$ be an isomorphism between their underlying categories. If $F$ is locally an isometry, namely, the assignment $f \mapsto Ff$ is an isometry $\C(a,b) \rightarrow \D(Fa,Fb)$ for all $a,b \in \Ob(\C)$, then $F$ is an enriched isomorphism, namely, both $F$ and $F^{-1}$ are $\qVCat$-functors.
		\end{corollary}

\begin{proof}
	We know by \cref{lem:characVfct} that since $F$ is locally nonexpansive, it is a $\qVCat$-functor. Let $F^{-1} \colon \D \to \C$ be the inverse of the functor $F$. For any $f,g \in \D(a,b)$, since $F$ is locally an isometry, we have
	$$\scalebox{0.95}{$d_{F^{-1}a,F^{-1}b}(F^{-1}f,F^{-1}F) = d_{a,b}(FF^{-1}f,FF^{-1}g) = d_{a,b}(f,g),$}$$
	hence $F^{-1}$ is locally nonexpansive (it is even locally an isometry). Therefore, by \cref{lem:characVfct} again, $F^{-1}$ is a $\qVCat$-functor, and it is straightforward to check that it is the inverse of $F$ (the $\qVCat$-functors have the same action as their underlying functors).
	\end{proof}

	\begin{corollary}\label{lem:nexpmontoenrichedmon}
		Let $\C$ be a $\qVCat$-category, whose underlying category is equipped with a strict monoidal product $\otimes \colon \C \times \C \to \C$. Then, $\C$ is a strict $\qVCat$-enriched monoidal category if and only if  for all $a,b,c,d \in \Ob(\C)$, $f,f' \in \C(a,b)$, $g,g' \in\C(c,d)$,  
		\begin{equation}\label{eqn:nexptensorenrichedmon}
			(d_{a,b} \boxtimes d_{c,d})((f,f'),(g,g'))\sqsubseteq d_{a \otimes c, b \otimes d}(f \otimes g, f'\otimes g').
		\end{equation}
		The same result holds for $\qVCatsym$.
	\end{corollary}
\begin{proof}
	We simply need to show that the monoidal product bifunctor is enriched. We make $\C \times \C$ into a  $\qVCat$-category using \Cref{def:app:vcat_tensor}. The hemimetric on hom-sets of $\C \times \C$ is defined in terms of distances on hom-sets of $\C$, namely $d_{(a,c), (b,d)} = d_{a,b} \boxtimes d_{c,d}$. Then, \cref{eqn:nexptensorenrichedmon} precisely says that $\otimes$ is locally nonexpansive. Thus, it is a $\qVCat$-functor by \Cref{lem:characVfct}.
\end{proof}

\subsection{Proofs of \Cref{sec:quant-mon-alg}}\label[app]{app:proofquantmonalg}

Recall the definition of the \emph{syntactic distance} $d^{\qU}$.
Let $\qU=(\Sigma,E,E_q)$ be a $\qV$-quantitative monoidal theory. For any $n,m \in \bN$ and $\Sigma$-terms $f,g\colon n \rightarrow m$, let the distance from $f$ to $g$ be given by
\begin{equation}\label{eqn:syntacticdistance}
	d^{\qU}_{n,m}(f,g) \coloneqq \vJoin \{ \varepsilon \mid f =_{\varepsilon} g \in\anyclose{E_q} \}.
\end{equation}
\begin{proof}[Proof of \Cref{lem:loc-hem-space}]
	Let us first prove the reflexivity and triangle inequality on $\Sigma$-terms.
		\begin{itemize}
		\item For any $\Sigma$-terms $f\colon n \rightarrow m$, we know that $f = f$ is provable from $E$, thus $f=_{\top} f \in\anyclose{E_q}$ by \RuleRefl. Hence, 
		\[k = \top = \vJoin \{ \varepsilon \mid f =_{\varepsilon} f \in\anyclose{E_q} \}=d^{\qU}_{n,m}(f,f).\]
		\item For any $\Sigma$-terms $f,g,h\colon n \rightarrow m$, we have the following derivation.
		\begin{align*}
			d^{\qU}_{n,m}&(f,g)\vten d^{\qU}_{n,m}(g,h) \\
			&=\vJoin \lbrace \varepsilon \mid f =_{\varepsilon} g \in\anyclose{E_q} \rbrace\vten\vJoin \lbrace \varepsilon' \mid g =_{\varepsilon'} h \in\anyclose{E_q} \rbrace
			& \\ & (\textrm{def.~of}\,d^\qU\,) \\
															& =
			\vJoin \lbrace \varepsilon\vten\varepsilon' \mid f =_{\varepsilon} g,\,g =_{\varepsilon'} h \in\anyclose{E_q}\rbrace
			& \\ & ( \textrm{join-continuity}) 
		\end{align*}
																																								On the other hand, we recall the definition of $d^\qU_{n,m}(f,h)$.
		$$d^\qU_{n,m}(f,h)\coloneqq\vJoin \lbrace \vartheta\mid f =_{\vartheta} h \in\anyclose{E_q}\rbrace$$
		Now, by \RuleTriang, we have that if $f=_\varepsilon g,\,g=_{\varepsilon'}h\in\anyclose{E_q}$, then $f=_{\varepsilon\vten\varepsilon'}h\in\anyclose{E_q}$. Hence, for any such $\varepsilon$ and $\varepsilon'$, 
		$$\varepsilon\vten\varepsilon'\vle d^\qU_{n,m}(f,h).$$
				We conclude the desired inequality holds:
		\begin{align*}
			d^{\qU}_{n,m}(f,g)&\vten d^{\qU}_{n,m}(g,h) \\
			&=
			\vJoin \lbrace \varepsilon\vten\varepsilon' \mid f =_{\varepsilon} g,\,g =_{\varepsilon'} h \in\anyclose{E_q}\rbrace \\
			&\vle d^\qU_{n,m}(f,h).
		\end{align*}
					\end{itemize}
	We can now show that $d^{\qU}_{n,m}$ is a well-defined function on $\syncat{\Sigma}{E}(n,m) \times \syncat{\Sigma}{E}(n,m)$. Namely, if $f = f'$ and $g = g'$ are provable from $E$, then $d^{\qU}(f,g) = d^{\qU}(f',g')$. By the triangle inequality, we have 
	\[d^{\qU}_{n,m}(f,f') \vten d^{\qU}_{n,m}(f',g') \vten d^{\qU}_{n,m}(g',g) \vle  d^{\qU}_{n,m}(f,g).\]
	But by \RuleRefl, we also have that $f =_{\top} f', g=_{\top} g' \in \anyclose{E_q}$, so $d^{\qU}_{n,m}(f,f') =  d^{\qU}_{n,m}(g',g) = \top = k$. Thus, we obtain $d^{\qU}_{n,m}(f',g') \vle d^{\qU}_{n,m}(f,g)$ and the symmetric inequation is proven similarly. It follows that $d^{\qU}_{n,m}$ is a $\qV$-hemimetric on $\syncat{\Sigma}{E}(n,m)$.

		Moreover, if $\anyclose{E_q}$ is closed under \RuleSym, then for any $\Sigma$-terms $f,g\colon n \rightarrow m$, 
					\[f=_\varepsilon g\in\anyclose{E_q}\,\Longleftrightarrow\,g=_\varepsilon f\in\anyclose{E_q}.\]
				Hence, 
		\begin{align*}
			d^{\qU}_{n,m}(f,g) &= \vJoin \{ \varepsilon \mid f =_{\varepsilon} g \in\anyclose{E_q} \} \\
			& =\vJoin \{ \varepsilon \mid g =_{\varepsilon} f \in\anyclose{E_q} \}=d^{\qU}_{n,m}(g,f) ,
		\end{align*}
				and $d^{\qU}$ is indeed a $\qV$-pseudometric.
																									\end{proof}

	\begin{proof}[Proof of \Cref{lem:subst-non-exp}]
		If we explicitly compute the left-hand side similarly to the triangle inequality part of \Cref{lem:loc-hem-space}, we get the following.
		\begin{align*}
			d^\qU_{n,m}(f_0,f_1) \ast& d^\qU_{m,\ell}(g_0,g_1) \\
			&=  \vJoin \lbrace \varepsilon \ast \varepsilon' \mid g_0 =_{\varepsilon} g_1,\,f_0 =_{\varepsilon' \in \anyclose{E_q}} f_1\rbrace,
		\end{align*} 
				where $\ast$ is one of $\{\vten,\vmeet\}$.
																																								Now, by \RuleSeqsum\ or \RuleSeqmeet, we know that if $g_0 =_{\varepsilon} g_1,\,f_0 =_{\varepsilon'} f_1 \in\anyclose{E_q}$, then $g_0 \circ f_0=_{\varepsilon \ast \varepsilon'}g_1 \circ f_1\in\anyclose{E_q}$. 
		Hence, 
		\begin{align*}
			d^\qU_{n,m}(f_0,f_1)\ast& d^\qU_{m,\ell}(g_0,g_1)   \\
			&=\vJoin \lbrace \varepsilon\ast \varepsilon' \mid f_0 =_{\varepsilon'} f_1,\,g_0=_{\varepsilon} g_1 \in\anyclose{E_q}\rbrace\\
			&\vle \vJoin \lbrace \vartheta \mid f_0;g_0=_{\vartheta}  f_1;g_1 \in\anyclose{E_q}\rbrace\\
			&=d^\qU_{n,\ell}(f_0;g_0\,,\,f_1;g_1).\qedhere
		\end{align*}
		
																					\end{proof}
	
	\begin{proof}[Proof of \Cref{def:synt-enr-cat}]
		This is a simple application of \cref{lem:nexpcomptoenriched,lem:nexpmontoenrichedmon} combined with \cref{lem:subst-non-exp,lem:tens-non-exp}. Let us make an example with $\Pclosesummax{E_q}$, the closure we will need in \cref{sec:axiomatisetotalvariation}. Looking at its row in \cref{tab:closures-quant-eq}, we need to show $\syncat{\Sigma}{E}$ is $\qVCatsymsum$-enriched monoidal.
		
		Since $\anyclose{E_q}$ is closed under \RuleSeqsum, \cref{lem:subst-non-exp} says that for any $f_0,f_1\colon n \rightarrow m$ and $g_0,g_1\colon m \rightarrow \ell$,
		\[d^\qU_{n,m}(f_0,f_1)\vten d^\qU_{m,\ell}(g_0,g_1) \vle d^\qU_{n,\ell}( f_0;g_0\,,\,f_1;g_1).\]
		This implies that sequential composition is a nonexpansive map of type
		\[{;}\colon\syncat{\Sigma}{E}(n,m) \boxtimes_{\vten} \syncat{\Sigma}{E}(m,\ell) \rightarrow \syncat{\Sigma}{E}(n,\ell).\]
		Therefore, we conclude by \cref{lem:nexpcomptoenriched} that $\syncat{\Sigma}{E}$ is $\qVCatsymsum$-enriched.
		
		Next, by closure under \RuleTensmeet, \cref{lem:tens-non-exp} says that for any $f_0,f_1\colon n \rightarrow n'$ and $f_1,g_1\colon m \rightarrow m'$,
		\begin{equation}\label{eqn:nexpparcomppsummeet}
			d^\qU_{n,n'}(f_0,f_1)\vmeet d^\qU_{m,m'}(g_0,g_1)\vle d^\qU_{n+m,n'+m'}(f_0\tensor g_0,f_1\tensor g_1).
		\end{equation}
		By \cref{rem:tensor_below_meet} and transitivity of $\vle$, we obtain
		\[d^\qU_{n,n'}(f_0,f_1)\vten d^\qU_{m,m'}(g_0,g_1)\vle d^\qU_{n+m,n'+m'}(f_0\tensor g_0,f_1\tensor g_1),\]
		so parallel composition is a nonexpansive map of type
		\begin{equation}\label{eqn:parcomppsumsum}
			{\tensor}\colon\syncat{\Sigma}{E}(n,n') \boxtimes_{\vten} \syncat{\Sigma}{E}(m,m') \rightarrow \syncat{\Sigma}{E}(n+n',m+m').
		\end{equation}
		Therefore, we conclude by \cref{lem:nexpmontoenrichedmon} that $\syncat{\Sigma}{E}$ is $\qVCatsymsum$-enriched monoidal.
											\end{proof}
	
	\begin{proof}[Proof of \Cref{thm:model-to-enriched-model}]

						By \cref{lem:characVfct}, to show that $M$ is enriched, it suffices to show that for any $\Sigma$-terms $f,g$,
		\[d^{\qU}(f,g) \vle d^{\C}(Mf,Mg).\]
		By definition of $d^{\qU}$ \cref{eqn:syntacticdistance}, this inequality holds if and only if for any $\varepsilon$ such that $f=_{\varepsilon} g \in \anyclose{E_q}$, $\varepsilon \vle d^{\C}(Mf,Mg)$. Equivalently, we need to show that all of $\anyclose{E_q}$ is true in $M$. This readily follows from the fact that all of $E_q$ is true in $M$, and all the inference rules used to generate $\anyclose{E_q}$ are valid in $\C$.
	\end{proof}
	
		\begin{proof}[Proof of \Cref{thm:completeness}]
				Assume that $f =_{\varepsilon} g$ is true in all $\anyclose{E_q}$-models of the theory $\qU$. Observe that the canonical identity functor $\syncat{\Sigma}{E} \to \syncat{\Sigma}{E}$ can be extended to an enriched functor of type $\sync{\qU} \to \sync{\qU}$, which is a $\anyclose{E_q}$-model of $\qU$. By \Cref{lem:characVfct}, we have that $\varepsilon \vle d^{\qU}_{n,m}(f, g)$. Recall that \[d^{\mathcal{U}}_{n,m}(f,g) = \vJoin \{\varepsilon' \mid f =_{\varepsilon'} g \in\anyclose{E_q} \}.\]
				Using \RuleJoin, we can deduce that $f =_{d^{\mathcal{U}}_{n,m}(f,g)} g \in\anyclose{E_q}$.
				Finally, using \RuleMon, we can derive $f=_\varepsilon g\in\anyclose{E_q}$. \qedhere

			\end{proof}
			\begin{remark}
				We proved that \cref{eqn:nexpparcomppsummeet} implies nonexpansiveness of \cref{eqn:parcomppsumsum}, but the former is strictly stronger than the latter. Indeed, \cref{eqn:nexpparcomppsummeet} is equivalent to parallel composition being a nonexpansive map of type
				\begin{equation}\label{eqn:parcomppmeet}
					\scalebox{0.91}{${\tensor}\colon\syncat{\Sigma}{E}(n,n') \boxtimes_{\vmeet} \syncat{\Sigma}{E}(m,m') \rightarrow \syncat{\Sigma}{E}(n+n',m+m')$}.
				\end{equation}
				Equivalently, it says that \RuleTensmeet\ is valid.
							\end{remark}
			
			\subsubsection{Summary of Different Closures}
			In \Cref{tab:closures-quant-eq}, we give an overview of the various possible closures. The first column lists the closures. The second column lists the additional rules applied beyond \RuleRefl, \RuleBot, \RuleTriang, \RuleMon, and \RuleJoin. The third indicates whether we need the IJD property. 	The fourth specifies the base of enrichment of the syntactic category that will be constructed (in the next section) using each closure. The last column has references to examples in this paper employing the respective closure.
	\begin{center}
		\begin{table}[!ht]
			\begin{center}
				
			\renewcommand{\arraystretch}{1.5}
		
			\begin{tabular}{|c|c|c|c|c|}
				\hline
								\textbf{} & \textbf{Rules} &  \textbf{IJD} &  \textbf{Enrichment} & \textbf{Example} \\ \hline
				
				$\Hclosesumsum{E_q}$	
				& \RuleSeqsum ,\RuleTenssum        
				&  No   
				& $\qVCatsum$          
				& \cref{sec:preoder-matrix}            \\ \hdashline[0.5pt/5pt]
				
				$\Pclosesumsum{E_q}$	
				& '' + \RuleSym       
				&  No 
				& $\qVCatsymsum$          
				& 				\\ \hdashline[0.5pt/5pt]
				
				$\Hclosemaxsum{E_q}$ 
				& \RuleSeqmeet ,\RuleTenssum    
				& Yes
				& $\qVCatsum$                     
				&     \\ \hdashline[0.5pt/5pt]
				
				$\Pclosemaxsum{E_q}$ 
				& '' + \RuleSym   
				& Yes
				& $\qVCatsymmax$                     
				&      \\ \hdashline[0.5pt/5pt]
				
				$\Hclosesummax{E_q}$ 	
				&   \RuleSeqsum ,\RuleTensmeet                              
				& Yes      
				& $\qVCatsum$ 
				&            \\ \hdashline[0.5pt/5pt]
				
				$\Pclosesummax{E_q}$ 	
				&   '' + \RuleSym                              
				& Yes      
				& $\qVCatsymsum$ 
				& \cref{sec:totalvarenrichment}             \\ \hdashline[0.5pt/5pt]
				
				$\Hclosemaxmax{E_q}$ 	
				&   \RuleSeqmeet ,\RuleTensmeet                              
				&   Yes 
				& $\qVCatmax$ 
				&       \\ \hdashline[0.5pt/5pt]
				
				$\Pclosemaxmax{E_q}$ 	
				&   '' + \RuleSym                            
				&   Yes 
				& $\qVCatsymmax$ 
				& 				\\ \hline
			\end{tabular}
			\caption{Different choices of closures for $E_q$.}
			\label{tab:closures-quant-eq}
		
		\end{center}
		\end{table}
	\end{center}

\subsection{Proofs of \Cref{sec:mat}}
\label[app]{app:proofs-matrices}

\begin{proof}[Proof of \Cref{lem:matrix-representation}]
	Given $f\colon n \rightarrow m$, Item 2 gives no choice on the scalars $f_{ij}$, they must be the entries of the matrix $F(f)$.
	By definition, $F$ sends the tensoring $\bigotimes_{i \in \fset{m},j \in \fset{n}} f_{ij}$ to a diagonal matrix containing all the scalars in some fixed order $f_{11}, f_{21}, \dots f_{m1},f_{12},\dots, f_{mn}$.
		We define $w^n_m$ and $b^n_m$ as the preimages of the following matrices:
	\[F(w^n_m) = \begin{bmatrix}
		\mathbf{I}_m & \overset{n}{\cdots}  &\mathbf{I}_m
	\end{bmatrix} \qquad F(b^n_m) = \begin{bmatrix} 1_R\\ \vdots  \\1_R \end{bmatrix} \oplus \overset{n}{\cdots}\oplus  \begin{bmatrix} 1_R\\ \vdots \\1_R \end{bmatrix}.\]			One can verify (by direct computation) that the $(i,j)$-entry of the matrix $F(b^n_m);F\left( \bigotimes_{i \in \fset{m},j \in \fset{n}} f_{ij} \right);F(w^n_m)$ is $f_{ij}$. Therefore,  applying $F$ to both sides of the equation in Item 1 yields the same matrix. Since $F$ is fully faithful by \cref{prop:mataxiom}, Item 1 must hold.
\end{proof}

We restate the definition of the entrywise preorder for future reference:
\begin{equation}\label{eq:preorderdef}
	A\leq B \text{ if and only if, for all }i,j, \ A_{ij}\leq B_{ij} 
	\text{ in } R.
\end{equation}
\begin{proof}[Proof of \Cref{thm:enrichmentMatR}]
	As the enrichment is given by the preorder~\cref{eq:preorderdef}, all we need to check is that both composition and the monoidal product in $\Mat_R$ are order-preserving, then apply \cref{lem:nexpcomptoenriched} and \cref{lem:nexpmontoenrichedmon}.
	
	For composition, this boils down to showing that, for any matrices $A,A'\in\Mat_{R}(n,m)$ and $B,B'\in\Mat_{R}(m,t)$, if $(B,A)\leq (B',A')$ then $BA\leq B'A'$. Let us look at the $(i,j)$-entry of $BA$.
	\begin{align*}
		(BA)_{ij} & =\Sigma_{k=1}^mB_{ik}A_{kj} \\
		& \leq \Sigma_{k=1}^mB'_{ik}A'_{kj} \\
		& =(B'A')_{ij}
	\end{align*}
	The inequality holds by \cref{ass:orderring} using that $(B_{ik},A_{kj})\leq (B'_{ik},A'_{kj})$ implies $A_{kj}\leq A'_{kj}$ and $B_{ik}\leq B'_{ik}$. Therefore, we get $B_{ik}A_{kj}\leq B'_{ik}A'_{kj}$. 
	
	For the monoidal product, we need to check that, for any matrices $N,N'\in\Mat_{R}(n,n')$ and $M,M'\in\Mat_{R}(m,m')$, if $(N,M)\leq (N',M')$, then $N\oplus M\leq N'\oplus M'$. Since $\leq$ is checked entrywise according to~\cref{eq:preorderdef}, this follows directly by definition of~$\oplus$.  
\end{proof}

\subsection{Probability Distributions and the Distribution Monad}\label[app]{app:distrmonad}

\begin{definition}[Probability distributions]
	Given a set $X$, a \emph{(probability) distribution} on $X$ is a function $\dist\colon X \rightarrow [0,1]$ satisfying $\sum_{x \in X} \dist(x) = 1$. We call $\dist(x)$ the \emph{weight} of $\dist$ at $x$, and we write $\dist(S)$ for the weight of $\dist$ on $S \subseteq X$, namely, $\dist(S) = \sum_{x \in S}\dist(x)$. The \emph{support} of $\dist$ is the subset $\supp{\dist} \subseteq X$ containing all elements where $\dist$ assigns nonzero weight. In the sequel, distributions are assumed to have a finite support.
	
	We denote with $\Dset X$ the set of finitely supported distributions on $X$, formally,
	\[\Dset X = \{\dist\colon X \rightarrow [0,1] \mid \scalebox{0.8}{$\displaystyle\sum_{x \in X}$} \dist(x) = 1 \text{ and } |\supp{\dist}| < \infty\}.\]
	For any function $f\colon X \rightarrow Y$, $\Dset f \colon \Dset X \rightarrow \Dset Y$ denotes the \emph{pushforward} map, it sends a distribution $\dist \in \Dset X$ to the distribution $\Dset f(\dist) \colon Y \rightarrow [0,1]$ defined by
	\[\Dset f(\dist)(y) = \dist(f^{-1}(y)) = \sum_{x \in f^{-1}(y)} \dist(x).\]
	In words, the weight of $\Dset f(\dist)$ at $y \in Y$ is the total weight of $\dist$ on the preimage of $y$ under $f$. This yields a functor $\Dset\colon \Set \rightarrow \Set$.
\end{definition}

The functor $\Dset\colon \Set \rightarrow \Set$ has a monad structure. We will not need the details here, but it is useful to recall the notion of convex combination, which is part of the algebraic theory of the monad. For any $p \in [0,1]$ and $\dist, \distb \in \Dset X$, their \emph{convex combination} $\dist +_p \distb$ is a distribution over $X$ defined by 
\[(\dist+_p\distb)(x) = p\dist(x) + (1-p)\distb(x).\]

The Kleisli category of the monad $\Dset\colon \Set \rightarrow \Set$ is often seen as a category of (discrete) stochastic processes (sometimes also called Markov kernels~\cite{Fritz2020}). We will study its full subcategory containing only finite sets $\fset{n} = \{0,\dots, n-1\}$ as objects, which may be defined concretely as follows. Very concisely, $\FStoch$ is the wide symmetric monoidal subcategory of $\Mat_{[0,1]}$ containing all and only the matrices that are \emph{stochastic}, namely, whose columns all sum up to $1$. We report a direct definition of $\FStoch$ below, expanding the one provided in the main text as a subcategory of $\Mat_{[0,1]}$.

\begin{definition}[$\FStoch$]
	The symmetric monoidal category $\FStoch$ has objects the natural numbers, and a morphism $n \rightarrow m$ is an $m \times n$ \emph{stochastic matrix}, i.e.~a matrix with entries in the interval $[0,1]$ such that the sum of the entries in a column always equals $1$. The identity maps are identity matrices, and composition is computed by matrix multiplication, namely, if $A\colon n \rightarrow m$ and $A'\colon m \rightarrow \ell$ are two stochastic matrices, then $A;A'\colon n \rightarrow \ell$ is defined as $A'A$. When $n$ is $0$, there is a unique morphism $0 \rightarrow  m$ that we identify as the empty $m \times 0$ matrix denoted by $[]$. When $m = 0$, there are no morphisms $n \rightarrow 0$.
The monoidal product is given by addition on objects and by direct sum on matrices: given $A\colon n \rightsquigarrow  n'$ and $A' \colon m \rightsquigarrow  m'$, $A \oplus A' \colon (m+m') \times (n+n')$ is defined as $A \oplus A' \coloneqq \left[\begin{smallmatrix}
	\scriptstyle A & \scriptstyle\mathbf{0} \\
			\scriptstyle\mathbf{0} & \scriptstyle A'
\end{smallmatrix}\right]$. For each $n,m$, the symmetry $\sigma_{n,m}$ is $\left[\begin{smallmatrix}
	\scriptstyle\mathbf{0} & \scriptstyle\mathbf{I}_n\\
			\scriptstyle\mathbf{I}_m & \scriptstyle\mathbf{0}
\end{smallmatrix}\right]$.
\end{definition}

\subsection{Proofs of \Cref{sec:axiomatisetotalvariation}}\label[app]{app:proofstv}

The proof of \Cref{lem:fsenriched} relies on nonexpansiveness of sequential and parallel composition in $\FStoch$ with the chosen monoidal product being $\boxtimes_+$. These facts are encapsulated by the following two lemmas.

\begin{lemma}\label{lem:fsnexpcomp}
	For any matrices $A,B \in \FStoch(n,m)$ and $A',B'\in \FStoch(m,\ell)$,
	\[\tvmax(A;A', B;B') \leq \tvmax(A',A) + \tvmax(B',B).\]
\end{lemma}

\begin{lemma}\label{lem:fsnexptensor}
	For any matrices $A,B \in \FStoch(n,m)$ and $A',B' \in \FStoch(n',m')$,
	\[\tvmax(A\oplus A', B \oplus B') \leq \tvmax(A,B) + \tvmax(A',B').\]
\end{lemma}

In our proof of \Cref{lem:fsnexpcomp} we will make use of an equivalent definition of $\tv$, which relies on the notion of coupling. We recall it below.
\begin{definition}[Couplings]
	Given two distributions $\dist,\distb \in \Dset X$, a \emph{coupling} of $\dist$ and $\distb$ is a distribution in $\Dset (X\times X)$ such that, for any $x \in X$, the total weight on $\{x\} \times X$ is $\dist(x)$ and the total weight on $X \times \{x\}$ is $\distb(x)$. More concisely, the set of couplings of $\dist$ and $\distb$ can be defined as
	\[\Cpl(\dist,\distb) = \left\{ \omega \in \Dset(X\times X) \mid  \Dset \pi_1(\omega) = \dist, \Dset\pi_2(\omega) = \distb\right\},\]
	where $\pi_1,\pi_2\colon X\times X \rightarrow X$ are the projections. We also call $\Dset \pi_1(\omega)$ and $\Dset\pi_2(\omega)$ the \emph{marginals} of $\omega$.
\end{definition}
By \cite[Theorem 4]{Gibbs2002} or \cite[Proposition 5.2]{Jacobs2024}, the total variation distance between $\dist$ and $\distb$ is\begin{equation}\label{eqn:tvcoupling}
	\tv(\dist,\distb) = \inf\left\{ \sum_{x \neq x' \in X} \omega(x,x') \mid \omega \in \Cpl(\dist,\distb) \right\}.
\end{equation}

\begin{proof}[Proof of \Cref{lem:fsnexpcomp}]
	For any finite number of coefficients $p_1,\dots, p_m \in [0,1]$ with $\sum_{j \in \fset{m}}p_j = 1$, we use the notation $\bigplus_{j \in \fset{m}}p_j \cdot \dist_j$ for the convex combination of distributions $\dist_j$.

	We need two simple facts.
	
	\textbf{Fact 1.} Given two composable matrices $A \in \FStoch(n,m)$ and $A' \in \FStoch(m,\ell)$, the $i$th column of $A;A'$ is a convex combination of the columns of $A'$ given by
	\[(A;A')_{i} = \bigplus_{j \in \fset{m}}A_{ij}\cdot A'_{j},\]
	where $A_{ij}$ denotes the $j$th entry of the $i$th column. This is readily verified by direct computation (recall that $A;A'$ is the multiplied matrices $A'A$).
	
	\textbf{Fact 2.} The total variation distance is a convex function, namely, for any $p_1,\dots, p_m \in [0,1]$ with $\sum_{j \in \fset{m}}p_j = 1$, and $\dist_1,\dots, \dist_m, \distb_1,\dots, \distb_m, \in \Dset\fset{\ell}$,
	\begin{equation}\label{eqn:tvconvex}
		\tv(\bigplus_{j \in \fset{m}}p_j\cdot \dist_j,\bigplus_{j \in \fset{m}}p_j\cdot \distb_j) \leq \sum_{j \in \fset{m}}p_j\tv(\dist_j,\distb_j).
	\end{equation}
	This holds because $\tv$ is an instance of a Kantorovich metric, and those are always convex (see e.g.~\cite[Lemma 2.6.(4)]{Jacobs2020}).
	
	We are ready to prove that the map ${;}\colon \FS(n,m) \otimes \FS(m,\ell) \rightarrow \FS(n,\ell)$ is nonexpansive. We use the first fact to express the distance between composites as a maximum of distances between convex combinations:
	\begin{align*}
		\tvmax(A;A',B;B') &= \max_{i \in \fset{n}}\tv((A;A')_{i}, (B;B')_{i})\\
		&= \max_{i \in \fset{n}}\tv( \bigplus_{j \in \fset{m}}A_{ij}\cdot A'_{j},  \bigplus_{j' \in \fset{m}}B_{ij'}\cdot B'_{j'}).\\
	\end{align*}
	Then, for every $i \in \fset{n}$, if we pick $\omega_i \in \Cpl(A_{i},B_{i})$, we have the following derivation:
	\begin{align*}
		\tv( &\bigplus_{j \in \fset{m}}A_{ij}\cdot A'_{j},  \bigplus_{j' \in \fset{m}}B_{ij'}\cdot B'_{j'})\\
		&=\tv( \bigplus_{j,j' \in \fset{m}}\omega_i(j,j')\cdot A'_{j},  \bigplus_{j,j' \in \fset{m}}\omega_i(j,j')\cdot B'_{j'})\\
		&\leq \sum_{j,j' \in \fset{m}}\omega_i(j,j')\tv(A'_{j},B'_{j'}).
							\end{align*}
	The equality holds because $\omega_i$ is a coupling of $A_{i}$ and $B_{i}$, and the inequality holds by the second fact. We decompose the last sum in two parts, one where $j\neq j'$ and one where $j=  j'$:
	\[ \sum_{j \neq j' \in \fset{m}}\omega_i(j,j')\tv(A'_{j},B'_{j'}) +  \sum_{j \in \fset{m}}\omega_i(j,j)\tv(A'_{j},B'_{j}).\]
	We can loosely bound the first part because the total variation distance $\tv(A'_{j},B'_{j'})$ is never bigger than $1$. Hence, writing $\ndrel{\fset{m}} = \{(j,j')\mid j\neq j' \in \fset{m}\}$,
	\[\sum_{j \neq j' \in \fset{m}}\omega_i(j,j')\tv(A'_{j},B'_{j'}) \leq \sum_{j \neq j' \in \fset{m}}\omega_i(j,j') = \omega_i(\ndrel{\fset{m}}).\]
	We can also bound the second part because $\tv(A'_{j},B'_{j}) \leq \tvmax(A',B')$ for every $j$ by definition, and because the total weight of $\omega_i$ on $\drel{\fset{m}} = \{(j,j)\mid j \in \fset{m}\}$ is at most $1$:
	\[\sum_{j \in \fset{m}}\omega_i(j,j)\tv(A'_{j},B'_{j}) \leq \omega_i(\drel{\fset{m}})\tvmax(A',B') \leq \tvmax(A',B').\]
	To summarise, we have shown that 
	\[\tvmax(A;A',B;B') \leq \max_{i \in \fset{n}}\omega_i(\ndrel{\fset{m}}) + \tvmax(A',B'),\]
	for any choice of couplings $\omega_i$.
	We can conclude by using the equivalent definition of $\tv$ \cref{eqn:tvcoupling}: for any $i$ and choice of $\omega_i$, \[\omega_i(\ndrel{\fset{m}}) \overset{\cref{eqn:tvcoupling}}{\leq} \tv(A_{i},B_{i}) \overset{\cref{eqn:tvmax}}{\leq} \tvmax(A,B).\]
	Therefore, we have the desired inequation
	\[\tvmax(A;A',B;B') \leq \tvmax(A,B) + \tvmax(A',B').\qedhere\]
			
																				\end{proof}

\begin{proof}[Proof of \Cref{lem:fsnexptensor}]
	Notice that the distance between two columns of a stochastic matrix, seen as distributions, is invariant when adding rows of zeroes. This can easily be inferred from another equivalent formulation of $\tv$ in e.g.~\cite{Gibbs2002}: \begin{equation}\label{eqn:tvcharacsum}
		\tv(\dist,\distb) = \frac{1}{2}\sum_{x \in X} |\dist(x) - \distb(x)|.
	\end{equation}
	If we add (or remove) an element $x \in X$ which is assigned weight $0$ by both $\dist$ and $\distb$, then the distance between the resulting distributions is clearly the same. In particular, the distances between the columns of the direct sums $A\oplus A'$ and $B \oplus B'$ is the distance between the columns of $A$ and $B$ or $A'$ and $B'$. Formally,
	\[\tv((A\oplus A')_{i}, (B\oplus B')_{i}) = \begin{cases}
		\tv(A_i,B_i) & i \leq n\\
		\tv(A'_i,B'_i) & i > n
	\end{cases}.\]
	It readily follows that 
	\begin{equation}\label{eqn:tensorfsmaxenriched}
		\tvmax(A\oplus A', B \oplus B') = \max\{\tvmax(A,B), \tvmax(A',B')\}.
	\end{equation}
	This implies the desired inequation because the maximum of two positive numbers is always smaller than their sum.
	\end{proof}

\begin{proof}[Proof of \Cref{lemma:tvF}]
	Applying $F$ to the left-hand side of \cref{eqn:ca:tv} yields:
	\begin{align*}
		F(\cc \otimes \del) & = F(\cc) \otimes F(\del)
		\\
		&= \begin{bmatrix}
			\lambda\\1-\lambda
		\end{bmatrix} \oplus []
		= \begin{bmatrix}
			\lambda\\1-\lambda\\0
		\end{bmatrix}.
	\end{align*}
	Applying $F$ to the right-hand side of \cref{eqn:ca:tv} yields:
	\begin{align*}
		F(\del \otimes \ccOneMinus) &= F(\del) \otimes F(\ccOneMinus)
		\\
		&= [] \oplus \begin{bmatrix}
			\scriptstyle 1-\lambda\\ \scriptstyle 1-(1-\lambda)
		\end{bmatrix}
		= \begin{bmatrix}
			\scriptstyle 0 \\ \scriptstyle 1-\lambda \\ \scriptstyle \lambda
		\end{bmatrix}.
	\end{align*}
	Viewing these column vectors as distributions on $\fset{3}$, the total variation distance between them is $\lambda$ (compute it with \cref{eqn:tvcharacsum}), so \cref{eqn:ca:tv} is true in $F$.
\end{proof}
This leads to the following result.
\begin{lemma}\label{lem:Ftvisenriched} The functor $F \colon \sync{\ConvAlg} \to \FStoch$ is an enriched functor $F \colon \sync{\BarAlg} \to \FS$. 
\end{lemma}
\begin{proof}[Proof of \Cref{lem:Ftvisenriched}]
	We know that $F$ is a model of $\ConvAlg$, so we want to apply \cref{thm:model-to-enriched-model}. In order to do so, we recall that \RuleSeqsum, \RuleSym, and \RuleTensmeet\ are valid in $\FS$. 	It remains to verify that $F$ satisfies all the quantitative equations in $\BarAlg$, which is what \cref{lemma:tvF} achieves.
	
	We conclude that $F$ is enriched, in particular, for any $s,t \in \sync{\BarAlg}(n,m)$,
	\begin{equation}\label{eq:isometryTV} 
		\tvmax(F(s),F(t)) \leq d^{\BarAlg}(s,t).\qedhere
	\end{equation}
								\end{proof}

\begin{proof}[Proof of \Cref{lem:isoondist}]
													Given two morphisms $f,g\colon 1 \rightarrow m$, let their corresponding distributions be $\dist,\distb \in \Dset \fset{m}$. We will show that $d^{\BarAlg}(f,g) \leq \tv(\dist,\distb) = \tvmax(F(f),F(g))$. 	
	First, letting $\lambda \coloneqq \tv(\dist,\distb)$, \cref{lem:splitting} tells us that the distributions corresponding to $f$ and $g$ are equal to the distributions corresponding respectively to
		\[\scalebox{0.7}{\input{distsplitalone.tikz}} \quad \text{ and } \quad \scalebox{0.7}{\input{distbsplitalone.tikz}}.\]
		Since $F$ is fully faithful (by \cref{thm:fritzmain}), it means the following equalities between diagrams can be proven in $\ConvAlg$.
	\begin{gather*}
		\scalebox{0.7}{\input{distsplit.tikz}} \\ 
		\scalebox{0.7}{\input{distbsplit.tikz}}
	\end{gather*}
	It is straightforward then to apply \cref{eqn:deldistrib}, the thick versions of unitality and associativity, and finally \RuleRefl\ to prove the following in $\sync{\BarAlg}$.
	\begin{equation}\label{eqn:distsplitthree}
		\scalebox{0.85}{\input{distsplitthree.tikz}}
	\end{equation}
	\begin{equation}\label{eqn:distbsplitthree}
		\scalebox{0.85}{\input{distbsplitthree.tikz}}
	\end{equation}
	The only thing that differs in the right-hand sides of \cref{eqn:distsplitthree} and \cref{eqn:distbsplitthree} is the left part. More precisely, 	we recognise both sides of \cref{eqn:ca:tv}. 	Therefore, we have
	\[\scalebox{0.6}{\input{tvcoupling.tikz}} =_{\lambda} \scalebox{0.6}{\input{tvcouplingr.tikz}} \quad \text{and}\]
	\[ \quad \scalebox{0.6}{\input{distsplitthreealone.tikz}} =_0 \scalebox{0.6}{\input{distbsplitthreealone.tikz}},\]	
	so we can apply \RuleSeqsum, then \RuleTriang\ twice with \cref{eqn:distsplitthree} and \cref{eqn:distbsplitthree} to obtain 
	\[\input{dist.tikz}\ =_{\lambda}\ \input{distb.tikz}.\]
	We conclude that $d^{\BarAlg}(f,g) \leq \lambda = \tv(\dist,\distb)$.	
	The converse inequality holds because $F$ is an enriched functor, hence $F$ is an isometry on the hom-sets with domain $1$. We know it is bijective from \cref{thm:fritzmain}.
\end{proof}

\begin{proof}[Proof of \Cref{thm:baralg-enriched-iso}]
	We already know that $F$ is an isomorphism on the underlying categories by \cref{thm:fritzmain}. Moreover, by \cref{lem:Ftvisenriched}, $F$ is also an enriched functor, so for any $f,g \in \sync{\BarAlg}(n,m)$, we already know that $d^{\BarAlg}(f,g) \geq \tvmax(F(f),F(g))$. Therefore, we only need to prove the converse inequation to conclude that $F$ is an isomorphism of enriched categories thanks to \cref{cor:characenrichediso}.
	
	We use \cref{lem:fritzdecompose} to find the following decompositions of $f$ and $g$:
	\[f = (f_1 \otimes \cdots \otimes f_n);p_m^n \text{ and } g = (g_1 \otimes \cdots \otimes g_n);p_m^n,\]
	where $f_i\colon 1 \rightarrow m$ and $F(f_i)$ is the $i$th column of $F(f)$, and similarly for $g$. Since \RuleSeqsum\ and \RuleTensmeet are valid in $\sync{\BarAlg}$, we have
	\begin{align*}
		d^{\BarAlg}(f,g) &\leq \max_{i \in \fset{n}} d^{\BarAlg}(f_i,g_i) + d^{\BarAlg}(p_m^n,p_m^n)\\
		&= \max_{i \in \fset{n}} d^{\BarAlg}(f_i,g_i).
	\end{align*}
	Since the $f_i$s and $g_i$s are morphisms with domain $1$, we can use \cref{lem:isoondist} to find that
	\[d^{\BarAlg}(f_i,g_i) = \tv(F(f_i),F(g_i)).\]
	Thus, we conclude, as desired, that \[d^{\BarAlg}(f,g) \leq \max_{i \in \fset{n}} \tv(F(f_i),F(g_i)) = \tvmax(F(f),F(g)).\qedhere\]
\end{proof}

\subsection{Proofs of \cref{sec:related}}\label[app]{app:related}
We recall the recursive definition of a cartesian term over a signature $\Sigma$ whose operations all have coarity $1$. Any variable $x$ (taken from a fixed countable set) is a cartesian term, and for any cartesian terms $t_1,\dots, t_n$ and $n$-ary operation $o : n \in \Sigma$, $o(t_1,\dots, t_n)$ is a cartesian term.

Since we are working with unconditional quantitative algebraic theories, we are interested in a simpler logic than the quantitative equational logic in \cite{Mardare2016}. We take inspiration from the simpler rules in \cite[Figure~3.1]{Sarkis2024}, but we add back the requirement that operations are nonexpansive with a corresponding rule. Another related logic is fuzzy equational logic in \cite[p.143]{Belholavek2005}. In \cref{fig:qeqlog}, we present the rules of quantitative equational logic.
\begin{figure}[!ht]
	\begin{gather*}
																														\begin{bprooftree}
			\AxiomC{}
			\RightLabel{\textsc{Bot'}}
			\UnaryInfC{$s =_\bot t$}
		\end{bprooftree}\quad \begin{bprooftree}
			\AxiomC{$s =_{\varepsilon} t$}
			\AxiomC{$\varepsilon' \vle \varepsilon$}
			\RightLabel{\textsc{Mon'}}
			\BinaryInfC{$s =_{\varepsilon'} t$}
		\end{bprooftree}\\[0.5em]\begin{bprooftree}
			\AxiomC{$\forall i \in I, s =_{\varepsilon_i} t$}
			\RightLabel{\textsc{Cont'}}
			\UnaryInfC{$s =_{\vJoin_i \varepsilon_i} t$}
		\end{bprooftree}\\[0.5em]
		\begin{bprooftree}
			\AxiomC{$s = t$ is provable in equational logic}
			\RightLabel{\textsc{Refl'}}
			\UnaryInfC{$s =_\top t$}
		\end{bprooftree}\\[0.5em]
		\begin{bprooftree}
			\AxiomC{$s=_{\varepsilon} t$}
			\RightLabel{\textsc{Symm'}\footnotemark}
			\UnaryInfC{$t=_{\varepsilon} s$}
		\end{bprooftree}\quad \begin{bprooftree}
			\AxiomC{$t=_{\varepsilon} t'$}
			\AxiomC{$t' =_{\varepsilon'} t''$}
			\RightLabel{\textsc{Triang}}
			\BinaryInfC{$t=_{\varepsilon\vten\varepsilon'} t''$}
		\end{bprooftree}\\[0.5em]
		\begin{bprooftree}
			\AxiomC{$t=_{\varepsilon} t'$}
			\RightLabel{\textsc{SubQ}}
			\UnaryInfC{$t[s_i/x_i] =_{\varepsilon} t'[s_i/x_i]$}
		\end{bprooftree}\\[0.5em]
		\begin{bprooftree}
			\AxiomC{$o:n \in \Sigma$}
			\AxiomC{$\forall i, s_i =_{\varepsilon_i} t_i$}
			\RightLabel{\textsc{NExp}}
			\BinaryInfC{$o(s_1,\dots, s_n) =_{\vmeet_i \varepsilon_i} o(t_1,\dots, t_n)$}
		\end{bprooftree}
	\end{gather*}
\caption{Rules of unconditional quantitative equational logic. The symbols $s$, $t$, and variants are universally quantified over cartesian terms, while the symbol $\varepsilon$ and variants are universally quantified over $\qV$. The notation $t[s_i/x_i]$ indicates the substitution of all occurrences of $x_i$ with the term $s_i$ inside $t$ for all $i$.}\label{fig:qeqlog}
\end{figure}
\footnotetext{This rule is considered only when working with $\qVCatsym$ instead of $\qVCat$.}
Given an unconditional quantitative algebraic theory $\qU = (\Sigma,E,E_q)$, we write $\qU \vdash s=_{\varepsilon} t$, or say that $s=_{\varepsilon} t$ in $\qU$, if there is proof tree of finite height (possibly infinite branching) that uses the rules in \cref{fig:qeqlog} or the axioms in $E$ and $E_q$ and concludes $s=_{\varepsilon} t$.

With \cref{def:generated-dlawvth}, the above completes the definition of $\LT{\qU}$, so we can start rigorously proving things about it. We restate the definition of $d_{\LT{\qU}}$ for reference:
\begin{equation}\label{defn:distLT}
	d_{\LT{\qU}}(\langle f_i \rangle, \langle g_i \rangle) = \vmeet_{i} \vJoin \left\{ \varepsilon \mid \qU \vdash f_i =_{\varepsilon} g_i \right\}.
\end{equation}
\begin{lemma}
	The inference rule \RuleSeqsum\ is valid in $\LT{\qU}$.
\end{lemma}
\begin{proof}
	We need to show that for all morphisms $s,t\colon n \rightarrow m$ and $s',t'\colon m \rightarrow \ell$,
	\begin{equation}\label{eqn:seqsumvalid}
		d_{\LT{\qU}}(s,t) \vten d_{\LT{\qU}}(s',t') \vle d_{\LT{\qU}}(s;s',t;t').
	\end{equation}
	Unrolling the definition of $d_{\LT{\qU}}$ \cref{defn:distLT} and of composition in $\LT{\qU}$, we rewrite the R.H.S. as
	\[\vmeet_{i} \vJoin\left\{\varepsilon \mid \qU \vdash s_i[s'_j/x_j] =_{\varepsilon} t_i[t'_j/x_j] \right\}.\]
	Our goal then is to show that for any $i$, we have
	\[\qU \vdash s_i[s'_j/x_j] =_{d_{\LT{\qU}}(s,t) \vten d_{\LT{\qU}}(s',t')} t_i[t'_j/x_j].\]
	We decompose it in two proofs that yield,
	\begin{align}
		\qU &\vdash s_i[s'_j/x_j] =_{d_{\LT{\qU}}(s,t)} t_i[s'_j/x_j], \text{ and}\label{eqn:seqsumvalidfirst}\\
		\qU &\vdash t_i[s'_j/x_j] =_{d_{\LT{\qU}}(s',t')} t_i[t'_j/x_j].\label{eqn:seqsumvalidsecond}
	\end{align}
	then conclude by the triangle inequality rule \textsc{Triang}.

	For \cref{eqn:seqsumvalidfirst}, for any $\varepsilon$ such that $\qU \vdash s_i =_{\varepsilon} t_i$, we can use the substitution rule \textsc{SubQ} to get $\qU \vdash s_i[s'_j/x_j] =_{\varepsilon} t_i[s'_j/x_j]$. Therefore, using the continuity rule \textsc{Cont'}, we also have $\qU \vdash s_i[s'_j/x_j] =_{\vJoin\left\{ \varepsilon \mid s_i =_{\varepsilon} t_i \right\}} t_i[s'_j/x_j]$. We obtain \cref{eqn:seqsumvalidfirst} by the monotonicity rule \textsc{Mon'} because
	\[\vJoin\left\{ \varepsilon \mid \qU \vdash s_i =_{\varepsilon} t_i \right\} \vge \vmeet_{k} \vJoin\left\{ \varepsilon \mid \qU \vdash s_{k} =_{\varepsilon} t_{k} \right\} = d_{\LT{\qU}}(s,t).\]

	For \cref{eqn:seqsumvalidsecond}, we recall that by the IJD property of $\qV$, we have
	\begin{align*}
		d_{\LT{\qU}}(s',t') &= \vmeet_j \vJoin \left\{\varepsilon \mid \qU \vdash s'_j =_{\varepsilon} t'_j \right\}\\
		&= \vJoin \left\{ \vmeet_j \varepsilon_j \mid \forall j, \qU \vdash s'_j =_{\varepsilon_j} t'_j \right\}.
	\end{align*}
	Then, for all choices of $\varepsilon_j$, such that $\qU \vdash s'_j =_{\varepsilon_j} t'_j$ for all $j$, we will prove by induction on the structure of $t_i$ that $\qU \vdash t_i[s'_j/x_j] =_{\vmeet_j\varepsilon_j} t_i[t'_j/x_j]$, and \cref{eqn:seqsumvalidsecond} will follow by \textsc{Cont}.

	In the base case, we suppose that $t_i = x_k$, and we have to show $\qU \vdash s'_k =_{\vmeet_j \varepsilon_j} t'_k$. This holds by using \textsc{Mon'} because by hypothesis we have $\qU \vdash s'_k =_{\varepsilon_k} t'_k$ and $\varepsilon_k \vge \vmeet_j \varepsilon_j$.

	For the inductive case, we suppose that $t_i = o(u_1,\dots, u_p)$ and that $\qU \vdash u_q[s'_j/x_j] =_{\vmeet_j \varepsilon_j} u_q[t'_j/x_j]$ for all $q$. Then, by the nonexpansiveness rule \textsc{NExp} (and the fact that $\vmeet$ is idempotent), we have $\qU \vdash o(u_1,\dots,u_p)[s'_j/x_j] =_{\vmeet_j\varepsilon_j} o(u_1,\dots, u_p)[t'_j/x_j]$ as desired.
\end{proof}
\begin{lemma}
	The inference rule \RuleTensmeet\ is valid in $\LT{\qU}$.
\end{lemma}
\begin{proof}
	We need to show that for all morphisms $s,t\colon n \rightarrow m$ and $s',t'\colon m \rightarrow \ell$,
	\begin{equation}\label{eqn:seqmeetvalid}
		d_{\LT{\qU}}(s,t) \vmeet d_{\LT{\qU}}(s',t') \vle d_{\LT{\qU}}(s \otimes s',t \otimes t').
	\end{equation}
	This is relatively simple because the definition of $\otimes$ in $d_{\LT{\qU}}$ concatenates tuples. Hence, the inequation above is in fact an equation because $\vmeet$ is associative.
\end{proof}

Before moving on to the proof of this section's main theorem, we clarify the definition of $\qU'$ \cref{def:associated-cart-quant-theory}. How exactly are the (quantitative) equations in $E$ and $E_q$ between cartesian terms ported to $\qU'$. By \cite[Theorem 6.1]{Bonchi2018}, there is an isomorphism $\Phi\colon \syncat{\Sigma'}{E^c} \rightarrow \LT{\Sigma}$, so that any cartesian term $t$ is assigned a string diagram $\Phi^{-1}(t)$, and also a monoidal $\Sigma'$-term after choosing a representative. Thus, if $s=t \in E$, the equation we include in $E'$ is between representatives of $\Phi^{-1}(s)$ and $\Phi^{-1}(t)$, and similarly for quantitative equations in $E_q$. In the following, we will abusively omit the step of choosing a representative and write, e.g., \[E'_q = \{\Phi^{-1}(s) =_{\varepsilon} \Phi^{-1}(t) \mid s=_{\varepsilon} t \in E_q\}.\] 
\begin{proof}[Proof of \cref{thm:from-quant.mon.th-to-quant.mon.th}]
	By \Cref{cor:characenrichediso} we know it is enough to have an isomorphism of the underlying categories which is locally an isometry. As stated above, \cite[Theorem~6.1]{Bonchi2018} gives us an isomorphism that we denote with $\Phi\colon \syncat{\Sigma'}{E'} \to \LT{\Sigma,E}$.\footnote{We abuse notation and use the same letter $\Phi$ because, as spelled out in \cite[Lemma 6.5]{Bonchi2018}, both functors morally act in the same way.} We have left to prove that for any natural numbers $n$ and $m$, $\Phi$ is locally an isometry.
	
	Our first step is to consider $\Phi$ as a model of $(\Sigma',E')$ valued in $\LT{\qU}$ (recall that its underlying category is $\LT{\Sigma,E}$), where \RuleSeqsum\ and \RuleTensmeet\ are valid by the previous lemmas. Hence, if we can show that the equations in $E'_q$ are true in this model, \cref{thm:model-to-enriched-model} will imply that $\Phi$ is locally nonexpansive.

	A quantitative equation in $E'_q$ has the shape $\Phi^{-1}(s) =_{\varepsilon} \Phi^{-1}(t)$ for some cartesian terms $s$ and $t$. To show it is true in the model $\Phi$, we need to show that $\varepsilon \vle d_{\LT{\qU}}(\Phi(\Phi^{-1}(s)), \Phi(\Phi^{-1}(t)))$. This is true because $\Phi^{-1} \circ \Phi = \id$ and, by hypothesis, $\qU \vdash s=_{\varepsilon} t$ as that equation belongs to the axioms in $E_q$. We conclude by \cref{thm:model-to-enriched-model} that $\Phi: \sync{\qU'} \rightarrow \LT{\qU}$ is locally nonexpansive.

		We still need to show that $\forall f,g\colon n \rightarrow m \in \sync{\qU'}$,
	\begin{equation}\label{eqn:isometrylawvdecomposed}
		d_{\LT{\qU}}(\Phi f,\Phi g) \vle d_{\qU'}(f,g).
	\end{equation}

	Note that for each $i$, the $i$th element of the tuple $\Phi f$ is represented by a diagram $f_i$ (apply $\Phi^{-1}$) such that $f = \bigoplus_i f_i$. Moreover, by \RuleTensmeet, we have that $\vmeet_i d_{\qU'}(f_i,g_i) \vle d_{\qU'}(f,g)$. Therefore, by \cref{defn:distLT}, it suffices to show that $d_{\LT{\qU}}(\Phi f_i,\Phi g_i) \vle d_{\qU'}(f_i,g_i)$ for every $i$ to conclude \cref{eqn:isometrylawvdecomposed}.
	
	Since $d_{\LT{\qU}}$ and $d_{\qU'}$ are both defined as the join of derivable distances, we can prove that inequation by showing that for each $\varepsilon$ such that $\qU \vdash \Phi f_i =_{\varepsilon} \Phi g_i$, $f_i =_{\varepsilon} g_i \in \Hclosesummax{E_q}$.

	We proceed by structural induction on the proof witnessing $\qU \vdash \Phi f_i =_{\varepsilon} \Phi g_i$. For the axioms in $E_q$, we know that $f_i =_{\varepsilon} g_i$ is an axiom in $E'_q$ by constructions. Most of the other rules in \cref{fig:qeqlog} are directly replicated in the inference rules used to define the closure of $E'_q$. 
		The two complicated rules are \textsc{SubQ} and \textsc{NExp}.													
	For \textsc{SubQ}, we note that the substitution $t[\sigma(x_i)/x_i]$ corresponds to the composition \[\Phi^{-1} (\langle \sigma(x_1),\dots, \sigma(x_n) \rangle) ; \Phi^{-1} t,\]
	and similarly the substitution $t'[\sigma(x_i)/x_i]$ corresponds to \[\Phi^{-1} (\langle \sigma(x_1),\dots, \sigma(x_n) \rangle) ; \Phi^{-1} t'.\]
	Therefore, we can apply \RuleRefl\ to obtain \[\Phi^{-1} (\langle \sigma(x_1),\dots, \sigma(x_n) \rangle) =_{\top} \Phi^{-1} (\langle \sigma(x_1),\dots, \sigma(x_n) \rangle),\] the induction hypothesis to obtain $\Phi^{-1}t =_{\varepsilon} \Phi^{-1}t'$, and finally \RuleSeqsum, to obtain the desired $\Phi^{-1}t[\sigma(x_i)/x_i] =_{\varepsilon} t'[\sigma(x_i)/x_i]$.

	For \textsc{NExp}, $o(s_1,\dots, s_n)$ again corresponds to a composition
	\[\raisebox{0.3ex}{\scalebox{0.4}{\input{copy.tikz}}} ; (\Phi^{-1}s_1 \otimes \cdots \otimes \Phi^{-1}s_n) ; o,\]
	where the copy diagram actually represents the copy on $n$ wires. Similarly for $o(t_1,\dots, t_n)$. 
	By the induction hypothesis, we have for all $i$, $\Phi^{-1}s_i =_{\varepsilon_i} \Phi^{-1}t_i$, which we can combine with \RuleTensmeet\ to find
	\[\Phi^{-1}s_1 \otimes \cdots \otimes \Phi^{-1}s_n =_{\vmeet_i\varepsilon_i} \Phi^{-1}t_1 \otimes \cdots \otimes \Phi^{-1}t_n.\]
	Therefore, we can apply \RuleRefl\ to obtain \[\raisebox{0.3ex}{\scalebox{0.4}{\input{copy.tikz}}}  =_\top \raisebox{0.3ex}{\scalebox{0.4}{\input{copy.tikz}}}  \text{ and } o =_{\top} o,\]
	and finally \RuleSeqsum\ twice to get the desired $\Phi^{-1}(o(s_1,\dots, s_n)) =_{\vmeet_i\varepsilon_i} \Phi^{-1}(o(t_1,\dots, t_n))$.
	\end{proof}

\end{document}

%% file: sample.tikzstyles
\tikzstyle{black}=[fill=black, draw=black, shape=circle, minimum size=6pt, inner sep=0pt, outer sep=0pt]
\tikzstyle{conv}=[fill=none, draw=black, shape=rectangle, inner sep=3pt]
\tikzstyle{letter}=[fill=none, inner sep=0pt, draw=black, regular polygon, regular polygon sides=4,rounded corners]
\tikzstyle{dist}=[fill=none, draw=black, isosceles triangle, isosceles triangle apex angle=60, inner sep=4pt, rotate=180, anchor=center, shape border uses incircle]
\tikzstyle{white}=[fill=white, draw=black, shape=circle, minimum size=6pt, inner sep=0pt, outer sep=0pt, tikzit draw=black, tikzit fill=white]

\tikzstyle{basic box}=[draw, inner sep=2pt, fill=white, rectangle, minimum height=1.2em, minimum width=1em]

\tikzstyle{single}=[-, draw=black]
\tikzstyle{double}=[-, draw=black, line width=2pt, tikzit fill=cyan, tikzit draw=cyan]
\tikzstyle{boxes}=[-, fill=white, dashed]

%% file: figures/letterbox.tikz
\begin{tikzpicture}
	\begin{pgfonlayer}{nodelayer}
		\node [style=none] (1) at (-1, 0) {};
		\node [style=none] (2) at (1, 0) {};
		\node [style=letter] (3) at (0, 0) {$\letter$};
	\end{pgfonlayer}
	\begin{pgfonlayer}{edgelayer}
		\draw (1.center) to (3);
		\draw (3) to (2.center);
	\end{pgfonlayer}
\end{tikzpicture}

%% file: figures/del.tikz
\begin{tikzpicture}
	\begin{pgfonlayer}{nodelayer}
		\node [style=black] (5) at (-1, 0) {};
		\node [style=none] (6) at (1, 0) {};
	\end{pgfonlayer}
	\begin{pgfonlayer}{edgelayer}
		\draw (5) to (6.center);
	\end{pgfonlayer}
\end{tikzpicture}

%% file: figures/generator-notype.tikz
\begin{tikzpicture}
	\begin{pgfonlayer}{nodelayer}
		\node [style=basic box] (0) at (0, 0) {$g$};
		\node [style=none] (1) at (-1.5, 0) {};
		\node [style=none] (8) at (1.5, 0) {};
	\end{pgfonlayer}
	\begin{pgfonlayer}{edgelayer}
		\draw [style=single] (1.center) to (0);
		\draw [style=single] (0) to (8.center);
	\end{pgfonlayer}
\end{tikzpicture}

%% file: figures/id.tikz
\begin{tikzpicture}
	\begin{pgfonlayer}{nodelayer}
		\node [style=none] (71) at (1, 0) {};
		\node [style=none] (72) at (-1, 0) {};
	\end{pgfonlayer}
	\begin{pgfonlayer}{edgelayer}
		\draw (72.center) to (71.center);
	\end{pgfonlayer}
\end{tikzpicture}

%% file: figures/idn.tikz
\begin{tikzpicture}
	\begin{pgfonlayer}{nodelayer}
		\node [style=none] (71) at (0.75, 0) {};
		\node [style=none] (72) at (-0.75, 0) {};
		\node [style=none] (73) at (0, 0.25) {\scriptsize $n$};
	\end{pgfonlayer}
	\begin{pgfonlayer}{edgelayer}
		\draw (72.center) to (71.center);
	\end{pgfonlayer}
\end{tikzpicture}

%% file: figures/diagf.tikz
\begin{tikzpicture}
	\begin{pgfonlayer}{nodelayer}
		\node [style=basic box] (24) at (0, 0) {$f$};
		\node [style=none] (25) at (1, 0) {};
		\node [style=none] (27) at (-1, 0) {};
	\end{pgfonlayer}
	\begin{pgfonlayer}{edgelayer}
		\draw (24) to (25.center);
		\draw (27.center) to (24);
	\end{pgfonlayer}
\end{tikzpicture}

%% file: figures/diagg.tikz
\begin{tikzpicture}
	\begin{pgfonlayer}{nodelayer}
		\node [style=basic box] (24) at (0, 0) {$g$};
		\node [style=none] (25) at (1, 0) {};
		\node [style=none] (27) at (-1, 0) {};
	\end{pgfonlayer}
	\begin{pgfonlayer}{edgelayer}
		\draw (24) to (25.center);
		\draw (27.center) to (24);
	\end{pgfonlayer}
\end{tikzpicture}

%% file: figures/diagh.tikz
\begin{tikzpicture}
	\begin{pgfonlayer}{nodelayer}
		\node [style=basic box] (24) at (0, 0) {$h$};
		\node [style=none] (25) at (1, 0) {};
		\node [style=none] (27) at (-1, 0) {};
	\end{pgfonlayer}
	\begin{pgfonlayer}{edgelayer}
		\draw (24) to (25.center);
		\draw (27.center) to (24);
	\end{pgfonlayer}
\end{tikzpicture}

%% file: figures/diagf0.tikz
\begin{tikzpicture}
	\begin{pgfonlayer}{nodelayer}
		\node [style=basic box] (24) at (0, 0) {$f_0$};
		\node [style=none] (25) at (1, 0) {};
		\node [style=none] (27) at (-1, 0) {};
	\end{pgfonlayer}
	\begin{pgfonlayer}{edgelayer}
		\draw (24) to (25.center);
		\draw (27.center) to (24);
	\end{pgfonlayer}
\end{tikzpicture}

%% file: figures/diagf1.tikz
\begin{tikzpicture}
	\begin{pgfonlayer}{nodelayer}
		\node [style=basic box] (24) at (0, 0) {$f_1$};
		\node [style=none] (25) at (1, 0) {};
		\node [style=none] (27) at (-1, 0) {};
	\end{pgfonlayer}
	\begin{pgfonlayer}{edgelayer}
		\draw (24) to (25.center);
		\draw (27.center) to (24);
	\end{pgfonlayer}
\end{tikzpicture}

%% file: figures/diagg0.tikz
\begin{tikzpicture}
	\begin{pgfonlayer}{nodelayer}
		\node [style=basic box] (24) at (0, 0) {$g_0$};
		\node [style=none] (25) at (1, 0) {};
		\node [style=none] (27) at (-1, 0) {};
	\end{pgfonlayer}
	\begin{pgfonlayer}{edgelayer}
		\draw (24) to (25.center);
		\draw (27.center) to (24);
	\end{pgfonlayer}
\end{tikzpicture}

%% file: figures/diagg1.tikz
\begin{tikzpicture}
	\begin{pgfonlayer}{nodelayer}
		\node [style=basic box] (24) at (0, 0) {$g_1$};
		\node [style=none] (25) at (1, 0) {};
		\node [style=none] (27) at (-1, 0) {};
	\end{pgfonlayer}
	\begin{pgfonlayer}{edgelayer}
		\draw (24) to (25.center);
		\draw (27.center) to (24);
	\end{pgfonlayer}
\end{tikzpicture}

%% file: figures/delete.tikz
\begin{tikzpicture}
	\begin{pgfonlayer}{nodelayer}
		\node [style=none] (2) at (0, 0) {};
		\node [style=black] (3) at (0.75, 0) {};
		\node [style=none] (4) at (-0.75, 0) {};
	\end{pgfonlayer}
	\begin{pgfonlayer}{edgelayer}
		\draw [style=single] (3) to (4.center);
	\end{pgfonlayer}
\end{tikzpicture}

%% file: figures/zer.tikz
\begin{tikzpicture}
	\begin{pgfonlayer}{nodelayer}
		\node [style=none] (2) at (0, 0) {};
		\node [style=white] (3) at (-1, 0) {};
		\node [style=none] (4) at (1, 0) {};
	\end{pgfonlayer}
	\begin{pgfonlayer}{edgelayer}
		\draw [style=single] (3) to (4.center);
	\end{pgfonlayer}
\end{tikzpicture}

%% file: figures/scalar.tikz
\begin{tikzpicture}
	\begin{pgfonlayer}{nodelayer}
		\node [style=none] (0) at (1.5, 0) {};
		\node [style=none] (1) at (-1.5, 0) {};
		\node [style=conv] (3) at (0, 0) {$\scalar{k}$};
	\end{pgfonlayer}
	\begin{pgfonlayer}{edgelayer}
		\draw [style=single] (1.center) to (3);
		\draw [style=single] (1.center) to (3);
		\draw [style=single] (3) to (0.center);
	\end{pgfonlayer}
\end{tikzpicture}

%% file: figures/distnode.tikz
\begin{tikzpicture}
	\begin{pgfonlayer}{nodelayer}
		\node [style=dist, label={center:$\dist$}] (20) at (0, 0) {};
		\node [style=none] (21) at (-2, 0) {};
		\node [style=none, label={120:${m}$}] (22) at (2, 0) {};
	\end{pgfonlayer}
	\begin{pgfonlayer}{edgelayer}
		\draw (21.center) to (20);
		\draw [style=double] (20) to (22.center);
	\end{pgfonlayer}
\end{tikzpicture}

%% file: figures/thickdel.tikz
\begin{tikzpicture}
	\begin{pgfonlayer}{nodelayer}
		\node [style=black] (0) at (-1, 0) {};
		\node [style=none, label={120:{${n}$}}] (1) at (1, 0) {};
	\end{pgfonlayer}
	\begin{pgfonlayer}{edgelayer}
		\draw [style=double] (0) to (1.center);
	\end{pgfonlayer}
\end{tikzpicture}

%% file: figures/dist.tikz
\begin{tikzpicture}
	\begin{pgfonlayer}{nodelayer}
		\node [style=dist, label={center:$\dist$}] (20) at (0, 0) {};
		\node [style=none] (21) at (-2, 0) {};
		\node [style=none] (22) at (2, 0) {};
	\end{pgfonlayer}
	\begin{pgfonlayer}{edgelayer}
		\draw (21.center) to (20);
		\draw [style=double] (20) to (22.center);
	\end{pgfonlayer}
\end{tikzpicture}

%% file: figures/distb.tikz
\begin{tikzpicture}
	\begin{pgfonlayer}{nodelayer}
		\node [style=dist, label={center:$\distb$}] (20) at (0, 0) {};
		\node [style=none] (21) at (-2, 0) {};
		\node [style=none] (22) at (2, 0) {};
	\end{pgfonlayer}
	\begin{pgfonlayer}{edgelayer}
		\draw (21.center) to (20);
		\draw [style=double] (20) to (22.center);
	\end{pgfonlayer}
\end{tikzpicture}

%% file: figures/id2.tikz
\begin{tikzpicture}
	\begin{pgfonlayer}{nodelayer}
		\node [style=none] (0) at (2, -0.75) {};
		\node [style=none] (1) at (-2, -0.75) {};
		\node [style=none] (2) at (-2, 0.5) {};
		\node [style=none] (3) at (2, 0.5) {};
	\end{pgfonlayer}
	\begin{pgfonlayer}{edgelayer}
		\draw (0.center) to (1.center);
		\draw (3.center) to (2.center);
	\end{pgfonlayer}
\end{tikzpicture}

%% file: main.bbl
\begin{thebibliography}{10}

\bibitem{ABRAMSKY200637}
Samson Abramsky.
\newblock What are the fundamental structures of concurrency?: We still don't know!
\newblock {\em Electronic Notes in Theoretical Computer Science}, 162:37--41, 2006.
\newblock Proceedings of the Workshop "Essays on Algebraic Process Calculi" (APC 25).
\newblock URL: \url{https://www.sciencedirect.com/science/article/pii/S1571066106004105}, \href {https://doi.org/10.1016/j.entcs.2005.12.075} {\path{doi:10.1016/j.entcs.2005.12.075}}.

\bibitem{Adamek2022}
Jiří Ad{\'a}mek.
\newblock Varieties of quantitative algebras and their monads.
\newblock In {\em Proceedings of the 37th Annual ACM/IEEE Symposium on Logic in Computer Science}, LICS '22, New York, NY, USA, 2022. Association for Computing Machinery.
\newblock \href {https://doi.org/10.1145/3531130.3532405} {\path{doi:10.1145/3531130.3532405}}.

\bibitem{Arkor2024}
Nathanael Arkor and Dylan McDermott.
\newblock The formal theory of relative monads.
\newblock {\em Journal of Pure and Applied Algebra}, 228(9):107676, 2024.
\newblock URL: \url{https://www.sciencedirect.com/science/article/pii/S0022404924000732}, \href {https://doi.org/10.1016/j.jpaa.2024.107676} {\path{doi:10.1016/j.jpaa.2024.107676}}.

\bibitem{Arkor2025}
Nathanael Arkor and Dylan McDermott.
\newblock Relative monadicity.
\newblock {\em Journal of Algebra}, 663:399--434, 2025.
\newblock URL: \url{https://www.sciencedirect.com/science/article/pii/S0021869324005167}, \href {https://doi.org/10.1016/j.jalgebra.2024.08.040} {\path{doi:10.1016/j.jalgebra.2024.08.040}}.

\bibitem{BacciinBook}
Giorgio Bacci, Radu Mardare, Prakash Panangaden, and Gordon Plotkin.
\newblock Quantitative equational reasoning.
\newblock In {\em Foundations of probabilistic programming}, pages 333--360. Cambridge: Cambridge University Press, 2021.
\newblock \href {https://doi.org/10.1017/9781108770750.011} {\path{doi:10.1017/9781108770750.011}}.

\bibitem{Bacci2018a}
Giorgio Bacci, Radu Mardare, Prakash Panangaden, and Gordon~D. Plotkin.
\newblock An algebraic theory of markov processes.
\newblock In Anuj Dawar and Erich Grädel, editors, {\em Proceedings of the 33rd Annual {ACM/IEEE} Symposium on Logic in Computer Science, {LICS} 2018, Oxford, UK, July 09-12, 2018}, page 679–688. {ACM}, 2018.
\newblock \href {https://doi.org/10.1145/3209108.3209177} {\path{doi:10.1145/3209108.3209177}}.

\bibitem{Bacci2021}
Giorgio Bacci, Radu Mardare, Prakash Panangaden, and Gordon~D. Plotkin.
\newblock Tensor of quantitative equational theories.
\newblock In Fabio Gadducci and Alexandra Silva, editors, {\em 9th Conference on Algebra and Coalgebra in Computer Science, {CALCO} 2021, August 31 to September 3, 2021, Salzburg, Austria}, volume 211 of {\em LIPIcs}, page 7:1–7:17. Schloss Dagstuhl - Leibniz-Zentrum für Informatik, 2021.
\newblock \href {https://doi.org/10.4230/LIPIcs.CALCO.2021.7} {\path{doi:10.4230/LIPIcs.CALCO.2021.7}}.

\bibitem{Baez2018}
John~C. Baez, Brandon Coya, and Franciscus Rebro.
\newblock Props in network theory.
\newblock {\em Theory Appl. Categ.}, 33:Paper No. 25, 727--783, 2018.

\bibitem{Baez2023}
John~C. Baez, Xiaoyan Li, Sophie Libkind, Nathaniel~D. Osgood, and Eric Redekopp.
\newblock A categorical framework for modeling with stock and flow diagrams.
\newblock In Jummy David and Jianhong Wu, editors, {\em Mathematics of Public Health: Mathematical Modelling from the Next Generation}, pages 175--207. Springer International Publishing, Cham, 2023.
\newblock \href {https://doi.org/10.1007/978-3-031-40805-2_8} {\path{doi:10.1007/978-3-031-40805-2_8}}.

\bibitem{Belholavek2005}
Radim B{\v{e}}lohl{\'a}vek and Vil{\'e}m Vychodil.
\newblock {\em Fuzzy equational logic}, volume 186 of {\em Stud. Fuzziness Soft Comput.}
\newblock Berlin: Springer, 2005.
\newblock \href {https://doi.org/10.1007/b105121} {\path{doi:10.1007/b105121}}.

\bibitem{Bonchi2019b}
Filippo Bonchi, Joshua Holland, Robin Piedeleu, Pawe\l{} Soboci\'{n}ski, and Fabio Zanasi.
\newblock Diagrammatic algebra: from linear to concurrent systems.
\newblock {\em Proc. ACM Program. Lang.}, 3(POPL), January 2019.
\newblock \href {https://doi.org/10.1145/3290338} {\path{doi:10.1145/3290338}}.

\bibitem{Bonchi2015}
Filippo Bonchi, Pawel Sobocinski, and Fabio Zanasi.
\newblock Full abstraction for signal flow graphs.
\newblock In {\em Proceedings of the 42nd ACM SIGPLAN-SIGACT symposium on principles of programming languages, POPL '15, Mumbai, India, January 12--18, 2015}, pages 515--526. New York, NY: Association for Computing Machinery (ACM), 2015.
\newblock \href {https://doi.org/10.1145/2676726.2676993} {\path{doi:10.1145/2676726.2676993}}.

\bibitem{Bonchi2018}
Filippo Bonchi, Pawe\l{} Soboci\'nski, and Fabio Zanasi.
\newblock Deconstructing {L}awvere with distributive laws.
\newblock {\em J. Log. Algebr. Methods Program.}, 95:128--146, 2018.
\newblock \href {https://doi.org/10.1016/j.jlamp.2017.12.002} {\path{doi:10.1016/j.jlamp.2017.12.002}}.

\bibitem{Breiner2019}
Spencer Breiner, Carl~A. Miller, and Neil~J. Ross.
\newblock Graphical {M}ethods in {D}evice-{I}ndependent {Q}uantum {C}ryptography.
\newblock {\em {Quantum}}, 3:146, May 2019.
\newblock \href {https://doi.org/10.22331/q-2019-05-27-146} {\path{doi:10.22331/q-2019-05-27-146}}.

\bibitem{BroadbentK23}
Anne Broadbent and Martti Karvonen.
\newblock Categorical composable cryptography: extended version.
\newblock {\em Log. Methods Comput. Sci.}, 19(4), 2023.

\bibitem{Capucci2022}
Matteo Capucci, Bruno Gavranovi\'c, Jules Hedges, and Eigil Fjeldgren~Rischel.
\newblock Towards foundations of categorical cybernetics.
\newblock In {\em Proceedings of the {F}ourth {I}nternational {C}onference on {A}pplied {C}ategory {T}heory}, volume 372 of {\em Electron. Proc. Theor. Comput. Sci. (EPTCS)}, pages 235--248. EPTCS, [place of publication not identified], 2022.

\bibitem{Coecke2010}
B.~Coecke, M.~Sadrzadeh, and S.~Clark.
\newblock Mathematical foundations for distributed compositional model of meaning.
\newblock {\em Linguistic Analysis}, 36:345--384, 2010.
\newblock Lambek Festschrift, J. van Benthem, M. Moortgat, and W. Buszkowski (eds.).

\bibitem{Coecke2008}
Bob Coecke and Ross Duncan.
\newblock Interacting quantum observables.
\newblock In {\em Proceedings of the 35th International Colloquium on Automata, Languages and Programming, Part II}, ICALP '08, page 298–310, Berlin, Heidelberg, 2008. Springer-Verlag.
\newblock \href {https://doi.org/10.1007/978-3-540-70583-3_25} {\path{doi:10.1007/978-3-540-70583-3_25}}.

\bibitem{Coecke2017}
Bob Coecke and Aleks Kissinger.
\newblock {\em Picturing quantum processes. {A} first course in quantum theory and diagrammatic reasoning}.
\newblock Cambridge: Cambridge University Press, 2017.
\newblock \href {https://doi.org/10.1017/9781316219317} {\path{doi:10.1017/9781316219317}}.

\bibitem{crescenzi2024categoricalfoundationdeeplearning}
Francesco~Riccardo Crescenzi.
\newblock Towards a categorical foundation of deep learning: A survey, 2024.
\newblock URL: \url{https://arxiv.org/abs/2410.05353}, \href {https://arxiv.org/abs/2410.05353} {\path{arXiv:2410.05353}}.

\bibitem{Crubille15}
Rapha\"elle Crubill\'e and Ugo Dal~Lago.
\newblock Metric reasoning about {$\lambda$}-terms: the affine case.
\newblock In {\em 2015 30th {A}nnual {ACM}/{IEEE} {S}ymposium on {L}ogic in {C}omputer {S}cience ({LICS} 2015)}, pages 633--644. IEEE Computer Soc., Los Alamitos, CA, 2015.
\newblock \href {https://doi.org/10.1109/LICS.2015.64} {\path{doi:10.1109/LICS.2015.64}}.

\bibitem{Cruttwell2022}
Geoffrey S.~H. Cruttwell, Bruno Gavranovi{\'{c}}, Neil Ghani, Paul Wilson, and Fabio Zanasi.
\newblock Categorical foundations of gradient-based learning.
\newblock In Ilya Sergey, editor, {\em Programming Languages and Systems}, pages 1--28, Cham, 2022. Springer International Publishing.

\bibitem{DalLago2022b}
Ugo Dal~Lago and Francesco Gavazzo.
\newblock Effectful program distancing.
\newblock {\em Proc. ACM Program. Lang.}, 6(POPL), January 2022.
\newblock \href {https://doi.org/10.1145/3498680} {\path{doi:10.1145/3498680}}.

\bibitem{DalLago2022}
Ugo Dal~Lago, Furio Honsell, Marina Lenisa, and Paolo Pistone.
\newblock {On Quantitative Algebraic Higher-Order Theories}.
\newblock In Amy~P. Felty, editor, {\em 7th International Conference on Formal Structures for Computation and Deduction (FSCD 2022)}, volume 228 of {\em Leibniz International Proceedings in Informatics (LIPIcs)}, pages 4:1--4:18, Dagstuhl, Germany, 2022. Schloss Dagstuhl -- Leibniz-Zentrum f{\"u}r Informatik.
\newblock URL: \url{https://drops.dagstuhl.de/entities/document/10.4230/LIPIcs.FSCD.2022.4}, \href {https://doi.org/10.4230/LIPIcs.FSCD.2022.4} {\path{doi:10.4230/LIPIcs.FSCD.2022.4}}.

\bibitem{Desharnais99}
Jos\'ee Desharnais, Vineet Gupta, Radha Jagadeesan, and Prakash Panangaden.
\newblock Metrics for labeled {M}arkov systems.
\newblock In {\em C{ONCUR}'99: concurrency theory ({E}indhoven)}, volume 1664 of {\em Lecture Notes in Comput. Sci.}, pages 258--273. Springer, Berlin, 1999.
\newblock URL: \url{https://doi.org/10.1007/3-540-48320-9_19}, \href {https://doi.org/10.1007/3-540-48320-9\_19} {\path{doi:10.1007/3-540-48320-9\_19}}.

\bibitem{Dilworth52}
R.~P. Dilworth and J.~E. McLaughlin.
\newblock Distributivity in lattices.
\newblock {\em Duke Math. J.}, 19:683--693, 1952.
\newblock \href {https://doi.org/10.1215/S0012-7094-52-01972-8} {\path{doi:10.1215/S0012-7094-52-01972-8}}.

\bibitem{Dwork06}
Cynthia Dwork.
\newblock Differential privacy.
\newblock In Michele Bugliesi, Bart Preneel, Vladimiro Sassone, and Ingo Wegener, editors, {\em Automata, Languages and Programming}, pages 1--12, Berlin, Heidelberg, 2006. Springer Berlin Heidelberg.

\bibitem{Fox1976}
Thomas Fox.
\newblock Coalgebras and {C}artesian categories.
\newblock {\em Comm. Algebra}, 4(7):665--667, 1976.
\newblock \href {https://doi.org/10.1080/00927877608822127} {\path{doi:10.1080/00927877608822127}}.

\bibitem{Fritz09}
Tobias Fritz.
\newblock A presentation of the category of stochastic matrices, March 2009.
\newblock URL: \url{http://arxiv.org/abs/0902.2554v2; http://arxiv.org/pdf/0902.2554v2}, \href {https://arxiv.org/abs/0902.2554v2} {\path{arXiv:0902.2554v2}}.

\bibitem{Fritz2020}
Tobias Fritz.
\newblock A synthetic approach to markov kernels, conditional independence and theorems on sufficient statistics.
\newblock {\em Advances in Mathematics}, 370:107239, 2020.
\newblock URL: \url{https://www.sciencedirect.com/science/article/pii/S0001870820302656}, \href {https://doi.org/10.1016/j.aim.2020.107239} {\path{doi:10.1016/j.aim.2020.107239}}.

\bibitem{Lobski2023b}
Ella Gale, Leo Lobski, and Fabio Zanasi.
\newblock A categorical approach to synthetic chemistry.
\newblock In {\em Theoretical aspects of computing -- ICTAC 2023. 20th international colloquium, Lima, Peru, December 4--8, 2023. Proceedings}, pages 276--294. Cham: Springer, 2023.
\newblock \href {https://doi.org/10.1007/978-3-031-47963-2_17} {\path{doi:10.1007/978-3-031-47963-2_17}}.

\bibitem{Gavazzo2023}
Francesco Gavazzo and Cecilia {Di Florio}.
\newblock Elements of quantitative rewriting.
\newblock {\em Proc. ACM Program. Lang.}, 7(POPL), jan 2023.
\newblock \href {https://doi.org/10.1145/3571256} {\path{doi:10.1145/3571256}}.

\bibitem{Gibbs2002}
Alison~L. Gibbs and Francis~Edward Su.
\newblock On choosing and bounding probability metrics.
\newblock {\em Int. Stat. Rev.}, 70(3):419--435, 2002.
\newblock URL: \url{scholarship.claremont.edu/hmc_fac_pub/680}, \href {https://doi.org/10.2307/1403865} {\path{doi:10.2307/1403865}}.

\bibitem{Hoffmann1981}
Rudolf-E. Hoffmann.
\newblock Continuous posets, prime spectra of completely distributive complete lattices, and {Hausdorff} compactifications.
\newblock Continuous lattices, {Proc}. {Conf}., {Bremen} 1979, {Lect}. {Notes} {Math}. 871, 159-208 (1981)., 1981.

\bibitem{HLarsen2021}
Nicholas~Gauguin Houghton-Larsen.
\newblock {\em A Mathematical Framework for Causally Structured Dilations and its Relation to Quantum Self-Testing}.
\newblock Phd thesis, University of Copenhagen, Copenhagen, Denmark, February 2021.
\newblock URL: \url{https://arxiv.org/abs/2103.02302}, \href {https://arxiv.org/abs/2103.02302} {\path{arXiv:2103.02302}}.

\bibitem{HylandPowerSketches}
Martin Hyland and John Power.
\newblock Symmetric monoidal sketches.
\newblock In {\em Proceedings of the 2nd ACM SIGPLAN International Conference on Principles and Practice of Declarative Programming}, PPDP '00, pages 280--288, New York, NY, USA, 2000. Association for Computing Machinery.
\newblock \href {https://doi.org/10.1145/351268.351299} {\path{doi:10.1145/351268.351299}}.

\bibitem{Jacobs2024}
Bart Jacobs.
\newblock Drawing with distance, 2024.
\newblock URL: \url{https://arxiv.org/abs/2405.18182}, \href {https://arxiv.org/abs/2405.18182} {\path{arXiv:2405.18182}}.

\bibitem{JacobsKZ21}
Bart Jacobs, Aleks Kissinger, and Fabio Zanasi.
\newblock Causal inference via string diagram surgery: {A} diagrammatic approach to interventions and counterfactuals.
\newblock {\em Math. Struct. Comput. Sci.}, 31(5):553--574, 2021.

\bibitem{Jacobs2020}
Bart Jacobs and Abraham Westerbaan.
\newblock {Distances between States and between Predicates}.
\newblock {\em {Logical Methods in Computer Science}}, {Volume 16, Issue 1}, February 2020.
\newblock URL: \url{https://lmcs.episciences.org/6154}, \href {https://doi.org/10.23638/LMCS-16(1:26)2020} {\path{doi:10.23638/LMCS-16(1:26)2020}}.

\bibitem{Jurka2024}
Jan Jurka, Stefan Milius, and Henning Urbat.
\newblock Algebraic reasoning over relational structures.
\newblock {\em Electronic Notes in Theoretical Informatics and Computer Science}, Volume 4 - Proceedings of MFPS XL, Dec 2024.
\newblock URL: \url{https://entics.episciences.org/14598}, \href {https://doi.org/10.46298/entics.14598} {\path{doi:10.46298/entics.14598}}.

\bibitem{lambeq}
Dimitri Kartsaklis, Ian Fan, Richie Yeung, Anna Pearson, Robin Lorenz, Alexis Toumi, Giovanni de~Felice, Konstantinos Meichanetzidis, Stephen Clark, and Bob Coecke.
\newblock lambeq: An efficient high-level python library for quantum {NLP}.
\newblock {\em CoRR}, abs/2110.04236, 2021.
\newblock URL: \url{https://arxiv.org/abs/2110.04236}, \href {https://arxiv.org/abs/2110.04236} {\path{arXiv:2110.04236}}.

\bibitem{KellyBook}
G.~M. Kelly.
\newblock Basic concepts of enriched category theory.
\newblock {\em Repr. Theory Appl. Categ.}, 2005(10):1--136, 2005.
\newblock Originally published by: Cambridge Univ. Press, Cambridge, 1992.

\bibitem{kissinger2017pictureperfectquantumkeydistribution}
Aleks Kissinger, Sean Tull, and Bas Westerbaan.
\newblock Picture-perfect quantum key distribution, 2017.
\newblock URL: \url{https://arxiv.org/abs/1704.08668}, \href {https://arxiv.org/abs/1704.08668} {\path{arXiv:1704.08668}}.

\bibitem{KongLiang2024}
Liang Kong, Wei Yuan, Zhi-Hao Zhang, and Hao Zheng.
\newblock Enriched monoidal categories {I}: Centers.
\newblock {\em Quantum topology}, 2024.

\bibitem{LafontCombinators}
Yves Lafont.
\newblock Interaction combinators.
\newblock {\em Inf. Comput.}, 137(1):69--101, August 1997.
\newblock \href {https://doi.org/10.1006/inco.1997.2643} {\path{doi:10.1006/inco.1997.2643}}.

\bibitem{Larsen2011}
Kim~G. Larsen, Uli Fahrenberg, and Claus Thrane.
\newblock Metrics for weighted transition systems: axiomatization and complexity.
\newblock {\em Theor. Comput. Sci.}, 412(28):3358--3369, 2011.
\newblock \href {https://doi.org/10.1016/j.tcs.2011.04.003} {\path{doi:10.1016/j.tcs.2011.04.003}}.

\bibitem{Lawvere1963}
F.~W. Lawvere.
\newblock Functorial semantics of algebraic theories.
\newblock {\em Proc. Natl. Acad. Sci. USA}, 50:869–872, 1963.
\newblock \href {https://doi.org/10.1073/pnas.50.5.869} {\path{doi:10.1073/pnas.50.5.869}}.

\bibitem{Lawvere73}
F~William Lawvere.
\newblock Metric spaces, generalized logic, and closed categories.
\newblock {\em Reprints in Theory and Applications of Categories}, 1:1–37, 2002.
\newblock Originally published in Rendiconti del seminario matématico e fisico di Milano, XLIII (1973).

\bibitem{Libkind2022}
Sophie Libkind, Andrew Baas, Micah Halter, Evan Patterson, and James~P. Fairbanks.
\newblock An algebraic framework for structured epidemic modelling.
\newblock {\em Philosophical Transactions of the Royal Society A: Mathematical, Physical and Engineering Sciences}, 380(2233), August 2022.
\newblock URL: \url{http://dx.doi.org/10.1098/rsta.2021.0309}, \href {https://doi.org/10.1098/rsta.2021.0309} {\path{doi:10.1098/rsta.2021.0309}}.

\bibitem{Lobski2023}
Leo Lobski and Fabio Zanasi.
\newblock String diagrams for layered explanations.
\newblock In {\em Proceedings of the fifth international conference on applied category theory, ACT 2022, Glasgow, United Kingdom, July 18--22, 2022}, pages 362--382. Waterloo: Open Publishing Association (OPA), 2023.
\newblock \href {https://doi.org/10.4204/EPTCS.380.21} {\path{doi:10.4204/EPTCS.380.21}}.

\bibitem{LorenzTull-causalmodels}
Robin Lorenz and Sean Tull.
\newblock Causal models in string diagrams.
\newblock {\em CoRR}, abs/2304.07638, 2023.

\bibitem{Mardare2016}
Radu Mardare, Prakash Panangaden, and Gordon~D. Plotkin.
\newblock Quantitative algebraic reasoning.
\newblock In Martin Grohe, Eric Koskinen, and Natarajan Shankar, editors, {\em Proceedings of the 31st Annual {ACM/IEEE} Symposium on Logic in Computer Science, {LICS} '16, New York, NY, USA, July 5-8, 2016}, page 700–709. {ACM}, 2016.
\newblock \href {https://doi.org/10.1145/2933575.2934518} {\path{doi:10.1145/2933575.2934518}}.

\bibitem{Mardare2017}
Radu Mardare, Prakash Panangaden, and Gordon~D. Plotkin.
\newblock On the axiomatizability of quantitative algebras.
\newblock In {\em 32nd Annual {ACM/IEEE} Symposium on Logic in Computer Science, {LICS} 2017, Reykjavik, Iceland, June 20-23, 2017}, page 1–12. {IEEE} Computer Society, 2017.
\newblock \href {https://doi.org/10.1109/LICS.2017.8005102} {\path{doi:10.1109/LICS.2017.8005102}}.

\bibitem{MiliusU19}
Stefan Milius and Henning Urbat.
\newblock Equational axiomatization of algebras with structure.
\newblock In Mikolaj Bojanczyk and Alex Simpson, editors, {\em Foundations of Software Science and Computation Structures - 22nd International Conference, {FOSSACS} 2019, Held as Part of the European Joint Conferences on Theory and Practice of Software, {ETAPS} 2019, Prague, Czech Republic, April 6-11, 2019, Proceedings}, volume 11425 of {\em Lecture Notes in Computer Science}, pages 400--417. Springer, 2019.
\newblock \href {https://doi.org/10.1007/978-3-030-17127-8\_23} {\path{doi:10.1007/978-3-030-17127-8\_23}}.

\bibitem{Mittal2016}
Sparsh Mittal.
\newblock A survey of techniques for approximate computing.
\newblock {\em ACM Comput. Surv.}, 48(4), March 2016.
\newblock \href {https://doi.org/10.1145/2893356} {\path{doi:10.1145/2893356}}.

\bibitem{Morrison2017}
Scott Morrison and David Penneys.
\newblock Monoidal categories enriched in braided monoidal categories.
\newblock {\em International Mathematics Research Notices}, 2019(11):3527--3579, 10 2017.
\newblock \href {https://arxiv.org/abs/https://academic.oup.com/imrn/article-pdf/2019/11/3527/28757603/rnx217.pdf} {\path{arXiv:https://academic.oup.com/imrn/article-pdf/2019/11/3527/28757603/rnx217.pdf}}, \href {https://doi.org/10.1093/imrn/rnx217} {\path{doi:10.1093/imrn/rnx217}}.

\bibitem{Orchard2019}
Dominic Orchard, Vilem-Benjamin Liepelt, and Harley Eades~III.
\newblock Quantitative program reasoning with graded modal types.
\newblock {\em Proc. ACM Program. Lang.}, 3(ICFP), July 2019.
\newblock \href {https://doi.org/10.1145/3341714} {\path{doi:10.1145/3341714}}.

\bibitem{Panangaden2009}
Prakash Panangaden.
\newblock {\em Labelled Markov Processes}.
\newblock Imperial College Press, GBR, 2009.

\bibitem{Perrone2024}
Paolo Perrone.
\newblock Markov categories and entropy.
\newblock {\em IEEE Trans. Inf. Theory}, 70(3):1671--1692, 2024.
\newblock \href {https://doi.org/10.1109/TIT.2023.3328825} {\path{doi:10.1109/TIT.2023.3328825}}.

\bibitem{Piedeleu2025b}
Robin Piedeleu, Mateo Torres-Ruiz, Alexandra Silva, and Fabio Zanasi.
\newblock A complete axiomatisation of equivalence for discrete probabilistic programming.
\newblock In Viktor Vafeiadis, editor, {\em Programming Languages and Systems}, pages 202--229, Cham, 2025. Springer Nature Switzerland.
\newblock \href {https://doi.org/10.1007/978-3-031-91121-7_9} {\path{doi:10.1007/978-3-031-91121-7_9}}.

\bibitem{PiedeleuZanasi2025}
Robin Piedeleu and Fabio Zanasi.
\newblock {\em An Introduction to String Diagrams for Computer Scientists}.
\newblock Elements in Applied Category Theory. Cambridge University Press, 2025.

\bibitem{Power2005}
John Power.
\newblock Discrete lawvere theories.
\newblock In {\em Proceedings of the First International Conference on Algebra and Coalgebra in Computer Science}, CALCO'05, page 348–363, Berlin, Heidelberg, 2005. Springer-Verlag.
\newblock \href {https://doi.org/10.1007/11548133_22} {\path{doi:10.1007/11548133_22}}.

\bibitem{Rosicky2024}
Jiří Rosický.
\newblock Discrete equational theories.
\newblock {\em Mathematical Structures in Computer Science}, 34(2):147–160, 2024.
\newblock \href {https://doi.org/10.1017/S096012952400001X} {\path{doi:10.1017/S096012952400001X}}.

\bibitem{Rosicky2023}
Jiří Rosický and Giacomo Tendas.
\newblock Enriched universal algebra, 2023.
\newblock \href {https://arxiv.org/abs/2310.11972} {\path{arXiv:2310.11972}}.

\bibitem{Sarkis2024}
Ralph Sarkis.
\newblock {\em Lifting Algebraic Reasoning to Generalized Metric Spaces}.
\newblock Phd thesis, ENS de Lyon, Lyon, France, September 2024.
\newblock \href {https://doi.org/10.5281/zenodo.14001076} {\path{doi:10.5281/zenodo.14001076}}.

\bibitem{Selinger_2010}
P.~Selinger.
\newblock {\em A Survey of Graphical Languages for Monoidal Categories}, pages 289--355.
\newblock Springer Berlin Heidelberg, 2010.
\newblock URL: \url{http://dx.doi.org/10.1007/978-3-642-12821-9_4}, \href {https://doi.org/10.1007/978-3-642-12821-9_4} {\path{doi:10.1007/978-3-642-12821-9_4}}.

\bibitem{shiebler2021categorytheorymachinelearning}
Dan Shiebler, Bruno Gavranovic, and Paul Wilson.
\newblock Category theory in machine learning, 2021.
\newblock URL: \url{https://arxiv.org/abs/2106.07032}, \href {https://arxiv.org/abs/2106.07032} {\path{arXiv:2106.07032}}.

\bibitem{Stone1949PostulatesFT}
Michael~H. Stone.
\newblock Postulates for the barycentric calculus.
\newblock {\em Annali di Matematica Pura ed Applicata}, 29:25--30, 1949.
\newblock URL: \url{https://api.semanticscholar.org/CorpusID:122252152}.

\bibitem{DBLP:conf/concur/BreugelW01}
Franck van Breugel and James Worrell.
\newblock An algorithm for quantitative verification of probabilistic transition systems.
\newblock In {\em {CONCUR}}, volume 2154 of {\em Lecture Notes in Computer Science}, pages 336--350. Springer, 2001.

\bibitem{vBW2001}
Franck van Breugel and James Worrell.
\newblock Towards quantitative verification of probabilistic transition systems.
\newblock In Fernando Orejas, Paul~G. Spirakis, and Jan van Leeuwen, editors, {\em Automata, Languages and Programming}, page 421–432, Berlin, Heidelberg, 2001. Springer Berlin Heidelberg.

\bibitem{vBWorrell2005}
Franck van Breugel and James Worrell.
\newblock A behavioural pseudometric for probabilistic transition systems.
\newblock {\em Theoret. Comput. Sci.}, 331(1):115--142, 2005.
\newblock \href {https://doi.org/10.1016/j.tcs.2004.09.035} {\path{doi:10.1016/j.tcs.2004.09.035}}.

\bibitem{Villani2009}
C\'{e}dric Villani.
\newblock {\em Optimal transport: old and new}, volume 338 of {\em Grundlehren der mathematischen Wissenschaften [Fundamental Principles of Mathematical Sciences]}.
\newblock Springer-Verlag, Berlin, 2009.
\newblock Old and new.
\newblock \href {https://doi.org/10.1007/978-3-540-71050-9} {\path{doi:10.1007/978-3-540-71050-9}}.

\bibitem{Wilson2022}
Paul Wilson and Fabio Zanasi.
\newblock Categories of differentiable polynomial circuits for machine learning.
\newblock In {\em Graph transformation}, volume 13349 of {\em Lecture Notes in Comput. Sci.}, pages 77--93. Springer, Cham, 2022.
\newblock URL: \url{https://doi.org/10.1007/978-3-031-09843-7_5}, \href {https://doi.org/10.1007/978-3-031-09843-7\_5} {\path{doi:10.1007/978-3-031-09843-7\_5}}.

\bibitem{Wootters1982Single}
W.~K. Wootters and W.~H. Zurek.
\newblock A single quantum cannot be cloned.
\newblock {\em Nature}, 299(5886):802--803, October 1982.
\newblock URL: \url{http://dx.doi.org/10.1038/299802a0}, \href {https://doi.org/10.1038/299802a0} {\path{doi:10.1038/299802a0}}.

\bibitem{zanasi:tel-01218015}
Fabio Zanasi.
\newblock {\em {Interacting Hopf Algebras- the Theory of Linear Systems}}.
\newblock Theses, {Ecole normale sup{\'e}rieure de lyon - ENS LYON}, October 2015.
\newblock URL: \url{https://theses.hal.science/tel-01218015}.

\end{thebibliography}
